\documentclass[aps,prx,twocolumn,reprint,superscriptaddress,nofootinbib,floatfix,longbibliography]{revtex4-2}
\pdfoutput=1

\usepackage{xcolor}
\usepackage{comment}

\usepackage[colorlinks=true, citecolor=blue, linkcolor=red, urlcolor=blue, pdfborder={0 0 0}]{hyperref}

\usepackage[T1]{fontenc}
\usepackage{amsmath}
\usepackage{amsthm}
\usepackage{amssymb}
\usepackage{braket}
\usepackage{tikz}
\usepackage{qcircuit}
\usepackage{tabularx}
\usepackage{algorithm}
\usepackage{algpseudocode}
\usepackage{multirow}
\usepackage{float}
\usepackage{dsfont}
\usepackage{xfrac}
\usepackage[
  format=plain,
  justification=justified,
  singlelinecheck=false,
  labelsep=period,
  font=small,
]{caption}
\usepackage{subcaption}
\usepackage{wrapfig}
\usepackage{nicefrac}

\makeatletter
\AtBeginDocument{%
  \captionsetup{format=plain, justification=justified, singlelinecheck=false}%
  \captionsetup[sub]{justification=justified, singlelinecheck=false}%
  \long\def\@makecaption#1#2{%
    \vskip\abovecaptionskip
    \begingroup
      \small
      \leftskip\z@skip
      \rightskip\z@skip
      \parfillskip\@flushglue
      \parindent\z@
      \noindent #1.~#2\par
    \endgroup
    \vskip\belowcaptionskip
  }%
}
\makeatother

\newtheorem{theorem}{Theorem}
\newtheorem{proposition}{Proposition}

\newtheorem{lemma}{Lemma}

\newtheorem{corollary}{Corollary}
\newtheorem{assumption}{Assumption}
\DeclareMathOperator*{\argmin}{arg\,min}
\DeclareMathOperator*{\argmax}{arg\,max}
\DeclareMathOperator{\Tr}{Tr}

\newcommand{\secref}[1]{Section~\ref{#1}}
\newcommand{\thmref}[1]{Theorem~\ref{#1}}
\newcommand{\lemref}[1]{Lemma~\ref{#1}}
\newcommand{\assumpref}[1]{Assumption~\ref{#1}}
\newcommand{\propref}[1]{Proposition~\ref{#1}}
\newcommand{\corref}[1]{Corollary~\ref{#1}}
\newcommand{\algoref}[1]{Algorithm~\ref{#1}}
\definecolor{qcgreen}{RGB}{0,128,0}
\algrenewcommand\algorithmicindent{1em}

\newcommand{\tabref}[1]{Table~\ref{#1}}
\renewcommand{\eqref}[1]{eq.~(\ref{#1})}
\newcommand{\figref}[1]{Fig.~\ref{#1}}
\newcommand{\appref}[1]{App.~\ref{#1}}

\def\orcid#1{\kern -0.4em\href{https://orcid.org/#1}{\includegraphics[keepaspectratio,width=0.7em]{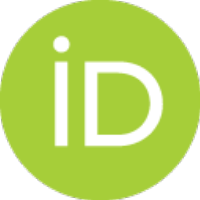}}} 

\usepackage{amsmath} 

\newcommand{\computerfont}[1]{{\fontfamily{cmtt}\selectfont #1}}
\definecolor{darkgreen}{rgb}{0.0, 0.5, 0.0} 

\newcommand{\EDIN}{\affiliation{School of Informatics, University of Edinburgh, Edinburgh, UK}}
\newcommand{\FUJITSU}{\affiliation{Fujitsu Research of Europe Ltd., United Kingdom}}
\newcommand{\LIP}{\affiliation{LIP6, CNRS, Sorbonne Universit\'{e}, Paris, France}}
\newcommand{\QCWARE}{\affiliation{QC Ware, Palo Alto, USA and Paris, France}}
\newcommand{\SIGNALS}{\affiliation{Quantum Signals, Paris, France}}

\begin{document}

\title{Adaptive directional gradients for parameterised quantum circuits}

\author{Brian Coyle}
\EDIN
\FUJITSU

\author{Snehal Raj}
\LIP
\QCWARE

\author{Virag Umathe}
\EDIN

\author{El Amine Cherrat}
\SIGNALS

\author{Elham Kashefi}
\EDIN
\LIP

\begin{abstract}
    Training parameterised quantum circuits (PQCs) on quantum hardware is bottlenecked by the measurement cost of gradient estimation, which under the parameter-shift rule scales linearly in the number of trainable parameters and dominates the total shot budget of training at scale. In this work, we propose a framework of \emph{forward gradient} estimators for PQCs, based on the forward mode of automatic differentiation, that yields an unbiased estimator of the gradient by averaging a freely tunable number of random directional derivatives and recovers SPSA, random coordinate descent, and the parameter-shift rule as limiting cases, with no ancilla qubits or controlled-gate overhead. We prove that stochastic quantum forward gradient descent converges under standard assumptions, with an explicit second-moment expansion that interpolates between the single-direction extreme of SPSA and the full-gradient extreme of parameter-shift. Within this framework we derive QUIVER (Quantum Iterative V-adaptive Estimator Rule), an adaptive optimiser for parameterised circuits whose update rule follows from a closed-form minimum measurement-cost allocation. We show numerically that forward gradients train Hamming-weight-preserving orthogonal quantum neural networks with up to 60 qubits and 1770 parameters on the ECG5000 and MNIST datasets orders of magnitude more efficiently than the parameter-shift rule. We also demonstrate that our proposed QUIVER optimiser can outperform iCANS and gCANS measurement-frugal optimisers on optimisation problems using the quantum approximate optimisation algorithm and quantum simulation with the variational quantum eigensolver.
\end{abstract}

\maketitle

\section{Introduction}
Gradient-based training is the dominant paradigm for parameterised quantum circuits (PQCs)~\cite{cerezo_variational_2021, bharti_noisy_2022}, from the variational quantum eigensolver (VQE) in chemistry to quantum neural networks (QNNs) in quantum machine learning (QML). As problem sizes grow, so does the number of trainable parameters, a trend compounded by recent evidence that quantum models, like their classical counterparts, benefit from overparameterisation~\cite{larocca_theory_2023, delgado_identifying_2023, garcia-martin_effects_2024}. The cost of extracting gradient information from such models is itself a function of the ansatz expressibility~\cite{holmes_connecting_2022}, the measurement cost per gradient step therefore grows as $\mathcal{O}(N)$ under the parameter-shift rule, placing gradient-based PQC training under increasing pressure as models scale~\cite{schuld_quantum_2022}. Classically, the efficiency of the backpropagation algorithm~\cite{rumelhart_learning_1986} and its generalisation into automatic differentiation (AD) pipelines~\cite{baydin_automatic_2018} ensures that gradient computation scales modestly with the number of parameters. No such general efficiency exists for quantum models.

All known approaches to efficient quantum gradient extraction require either fault-tolerant circuits with deep ancilla overhead~\cite{abbas_quantum_2023}, architectures with restricted expressibility~\cite{bowles_backpropagation_2023, coyle_training-efficient_2024, chinzei_trade-off_2024}, or \emph{approximate} gradient estimators that remove the explicit $\mathcal{O}(N)$ parameter dependence from the gradient rule itself. The third class is attractive for near-term devices, but the parameter dependence is not eliminated. In SPSA~\cite{spall_multivariate_1992}, a single random perturbation replaces $N$ parameter-shift evaluations, yet the $N$-dependence re-emerges in the estimator variance. Similarly, random coordinate descent (RCD)~\cite{ding_random_2024} computes one exact gradient component per step, but convergence slows by $\mathcal{O}(N)$.

In this work, we show that these and other approximate estimators are special cases of a single framework: \emph{forward gradients}~\cite{baydin_gradients_2022}\footnote{Closely related ideas appear under different names in the classical literature, including directional gradient descent for reinforcement learning~\cite{silver_learning_2022} and sketched gradient methods for large-scale optimisation~\cite{hanzely_sega_2018}; these differ in scope but share the idea of estimating gradients from a small number of random projections.}. Forward gradients reconstruct the full gradient from $V$ random directional derivatives, where $V$ is a free parameter independent of $N$. By varying $V$ and the direction distribution, one recovers SPSA ($V=1$, Rademacher), RCD ($V=1$, basis vectors), and the parameter-shift rule ($V=N$, basis vectors) as limiting cases. For classical neural networks, forward gradients have shown backpropagation-competitive performance without the backward lock~\cite{hinton_forward-forward_2022} or memory overhead~\cite{fournier_can_2023, ren_scaling_2023}. Quantum circuits admit no known analogue of backpropagation, making the forward mode a natural gradient strategy for this setting. Casting SPSA, RCD, and the parameter-shift rule as limiting cases of this framework provides a unified treatment of estimators that have until now been analysed separately, and exposes $V$ as an explicit lever through which one can interpolate between the cheapest and most expensive gradient strategies.

We make three contributions. First, we establish a convergence bound for stochastic gradient descent with the central-difference forward gradient estimator on PQCs, incorporating both the $\mathcal{O}(\varepsilon^2)$ finite-difference bias and quantum measurement shot noise (\secref{sec:convergence}), which to our knowledge is the first such result for this estimator class on parameterised quantum circuits. The bound reveals an exact no-free-lunch result on convex losses: the $V$-fold saving in per-step cost is exactly compensated by a $V$-fold increase in the number of steps required, so the total shot cost is $V$-independent. The bound also identifies the finite-difference step size $\varepsilon$ as the dominant hyperparameter: shot noise is amplified by $1/\varepsilon^2$, and the bias--variance trade-off induced by $\varepsilon$ is empirically more consequential than the choice of $V$.

Second, we derive \textsc{quiver} (\textbf{Qu}antum \textbf{I}terative \textbf{V}-adaptive \textbf{E}stimator \textbf{R}efinement), an adaptive optimiser built on the forward gradient estimator (\secref{sec:adaptive_optimizer_forward}). Extending the CABS framework~\cite{balles_coupling_2017} and the quantum iCANS/gCANS family~\cite{kubler_adaptive_2020, gu_adaptive_2021} to the directional-derivative setting reveals a structural property absent from the parameter-shift rule: measurement noise concentrates uniformly across random Rademacher directions, neutralising the per-direction shot-allocation lever that makes iCANS effective. Maximising expected gain per measurement shot over the remaining lever (the number of directions $V$) yields a closed-form joint $(V, M)$ rule that is variance-optimal (Rademacher directions uniquely minimise the estimator second moment at fixed $V$) and saturates the Cram\'er--Rao lower bound for gradient recovery from a quantum shot oracle up to a constant that vanishes as $N \to \infty$. We demonstrate numerically that \textsc{quiver} outperforms iCANS and gCANS on VQE ground-state optimisation and QAOA MaxCut benchmarks.

Third, we validate the forward gradient estimator empirically at parameter counts $N \sim 10^2$--$10^3$, training orthogonal quantum neural networks~\cite{landman_quantum_2022, monbroussou_trainability_2023, coyle_training-efficient_2024} with up to $60$ qubits ($N = 1770$ parameters) on ECG5000 time-series and MNIST image classification (\secref{sec:forward_at_scale}). Across four problem classes, a fixed $V \ll N$ estimator reaches the accuracy of the parameter-shift rule at a fraction of its total shot budget, with the cost saving growing with $N$. We also show heuristically that a decaying $V$-schedule outperforms both fixed-$V$ endpoints at matched budget (\secref{sec:v_scheduling_results}); we present this as a motivating observation rather than a systematic result, and as the empirical precursor to the \textsc{quiver} treatment of \secref{sec:adaptive_optimizer_forward}.
The remainder of this paper is organised as follows. \secref{sec:background} reviews gradient descent, automatic differentiation, and quantum gradient extraction including the parameter-shift rule and resource-frugal optimisers. \secref{sec:quantum_forward}--\secref{sec:forward_gradients_quantum_ml} propose forward gradients for quantum circuits and develop finite-difference directional derivatives. \secref{sec:forward_at_scale} and \secref{sec:v_scheduling_results} validate the estimator on classification, ground-state estimation, and combinatorial optimisation benchmarks, and \secref{sec:adaptive_optimizer_forward} develops the gain-theoretic analysis, noise concentration result, and the \textsc{quiver} optimiser.

\begin{figure*}[!t]
    \centering
    \includegraphics[width=\textwidth]{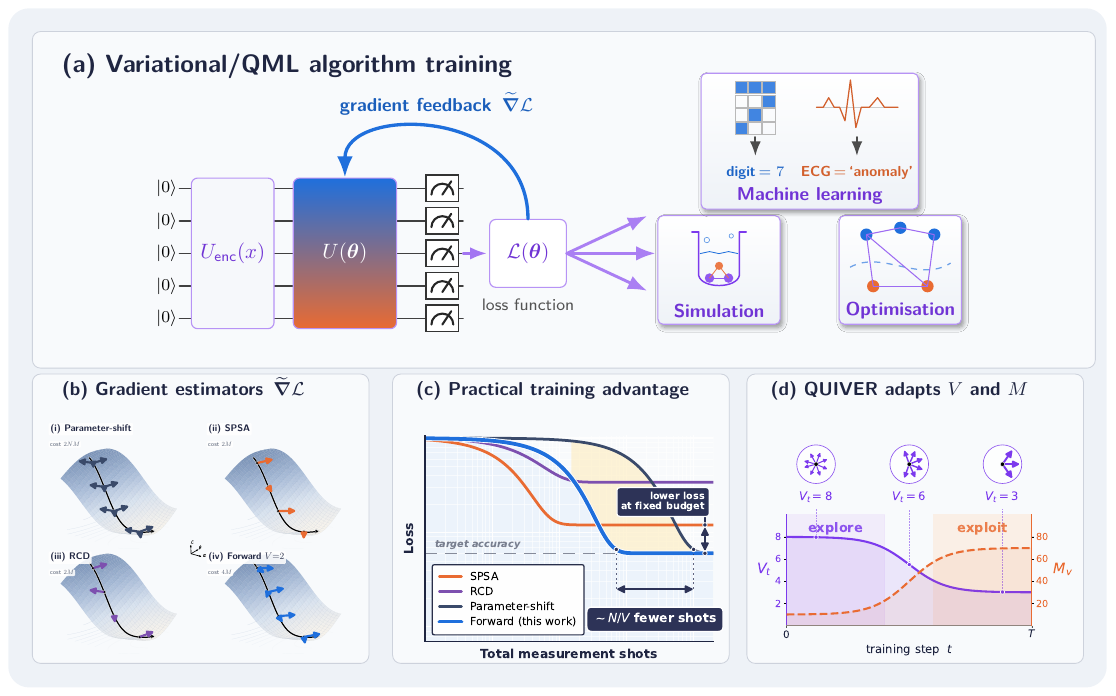}
    \caption{\textbf{Overview of the forward gradient framework and the \textsc{quiver} optimiser.} (a)~Variational/QML algorithm training loop: a parametric quantum model $U(\boldsymbol{\theta})$ (preceded by a data-encoding unitary $U_{\text{enc}}(x)$) is measured to produce a loss $\mathcal{L}(\boldsymbol{\theta})$ suited to a range of tasks (quantum simulation, combinatorial optimisation, machine learning). A gradient estimator $\widetilde{\boldsymbol{\nabla}}\mathcal{L}$ feeds back into $U(\boldsymbol{\theta})$ to update its parameters. (b)~Gradient estimators $\widetilde{\boldsymbol{\nabla}}\mathcal{L}$, shown as random directions probed in the $N$-dimensional parameter space at fixed per-direction shot count $M$. The parameter-shift rule (PS) evaluates $2N$ circuits per step covering all $N$ basis directions (cost $2NM$). SPSA draws a single random direction ($2M$ shots). RCD draws a single basis-aligned direction ($2M$ shots). The proposed forward gradient estimator (this work) draws $V$ random directions ($2VM$ shots), with $V$ as a free parameter interpolating between SPSA ($V=1$) and PS ($V=N$). (c)~Practical training advantage (cartoon): forward gradients with $V \ll N$ reach the same loss as PS in approximately $N/V$ times fewer total shots, and achieve lower loss than SPSA and RCD at fixed budget. (d)~The \textsc{quiver} optimiser adapts the number of directions $V_t$ and shots-per-direction $M_v$ over training, where $t$ indexes the training step and $v$ a particular sampled direction: early in training, many noisy directions are sampled at low shot count. Later, $V_t$ is pruned to a few directions with high shot count, concentrating the budget along the gradient direction $\boldsymbol{\nabla}\mathcal{L}$.}
    \label{fig:schematic}
\end{figure*}

\section{Background} \label{sec:background}

\subsection{Gradient descent and automatic differentiation} \label{ssec:gradient_descent}
The goal of gradient-based optimisation is to find parameters $\boldsymbol{\theta}^* = \argmin_{\boldsymbol{\theta}} \mathcal{L}(\boldsymbol{\theta})$ that minimise a loss function $\mathcal{L}(\boldsymbol{\theta})$ computed from the output of some model $f(\boldsymbol{\theta})$. In gradient descent, one moves iteratively in the direction of steepest descent with step size $\eta$:
\begin{equation} \label{eqn:gradient_descent}
        \boldsymbol{\theta}^{t+1} = \boldsymbol{\theta}^{t} - \eta \boldsymbol{\nabla} \mathcal{L}(\boldsymbol{\theta}^{t})
\end{equation}
where $\boldsymbol{\theta} \in \mathbb{R}^N$ and $[\boldsymbol{\nabla} \mathcal{L}(\boldsymbol{\theta})]_j := \sfrac{\partial \mathcal{L}(\boldsymbol{\theta})}{\partial \theta_j}$. For vector-valued losses the gradient generalises to the Jacobian $\boldsymbol{\mathcal{J}}_{ij} := \sfrac{\partial \boldsymbol{\mathcal{L}}_i}{\partial \theta_j}$. In practice we can only compute \emph{estimators} $\widetilde{\boldsymbol{g}}(\boldsymbol{\theta})$\footnote{We use \(\widetilde{\boldsymbol{g}}(\boldsymbol{\theta})\) and \(\widetilde{\boldsymbol{\nabla}}\mathcal{L}\) interchangeably throughout.} of the true gradient, ideally \emph{unbiased}: $\mathbb{E}[\widetilde{\boldsymbol{g}}(\boldsymbol{\theta})] = \boldsymbol{\nabla} \mathcal{L}(\boldsymbol{\theta})$. In classical ML the stochasticity arises from data subsampling; for quantum models, measurement shot noise introduces a second source. The most widely used variant, Adam~\cite{kingma_adam_2017}, replaces the raw gradient with exponentially weighted first and second moment estimates:
\begin{align}
    \boldsymbol{\theta}^{t+1} &= \boldsymbol{\theta}^{t} - \eta \frac{\boldsymbol{\widehat{m}}^{t}}{\sqrt{\boldsymbol{\widehat{v}}^{t}} + \varepsilon}, \\
    \boldsymbol{\widehat{m}}^{t} &:= \beta_1\boldsymbol{\widehat{m}}^{t-1} + (1-\beta_1)\widetilde{\boldsymbol{g}}(\boldsymbol{\theta}^{t}), \\
    \boldsymbol{\widehat{v}}^{t} &:= \beta_2\boldsymbol{\widehat{v}}^{t-1} + (1-\beta_2)\widetilde{\boldsymbol{g}}^2(\boldsymbol{\theta}^{t})
\end{align}
computed elementwise. In practice, learning rate scheduling $\eta \rightarrow \eta(t)$ is also applied.

The efficiency of gradient descent depends on how the gradient is computed. In modern machine learning, automatic differentiation (AD)~\cite{baydin_automatic_2018} is the standard approach, implemented natively in JAX~\cite{bradbury_jax_2018}, PyTorch~\cite{paszke_pytorch_2019}, and TensorFlow~\cite{martin_abadi_tensorflow_2015}. AD operates in two modes: the \emph{forward} mode computes Jacobian-vector products $\boldsymbol{\mathcal{J}} \boldsymbol{v}$ (yielding one column of $\boldsymbol{\mathcal{J}}$ per pass), while the \emph{reverse} mode computes vector-Jacobian products $\boldsymbol{v}^{\top}\boldsymbol{\mathcal{J}}$ (one row per pass)\footnote{In JAX, these correspond to \computerfont{.jvp()} and \computerfont{.vjp()} respectively.}. Since most ML losses map $\mathbb{R}^N \rightarrow \mathbb{R}^M$ with $N \gg M$, the reverse mode is overwhelmingly preferred: a single pass computes $\partial \boldsymbol{\mathcal{L}}_i / \partial \theta_j$ for all $j$ simultaneously. This is the backpropagation algorithm~\cite{rumelhart_learning_1986, griewank_numerical_2012, schmidhuber_deep_2015}, though it requires storing intermediate activations and is inherently sequential. We will return to the forward mode in~\secref{sec:forward_ml}, where it plays a central role.

\subsection{Computing quantum gradients} \label{ssec:quantum_gradients}
For parameterised quantum circuits, the situation is fundamentally different. A trainable unitary $\mathcal{U}(\boldsymbol{\theta}) := U(\theta_1) \cdots U(\theta_d)$ is applied to an initial state $\ket{\psi}$\footnote{Typically $\ket{0}^{\otimes n}$, but for QML applications one may encode data via $\ket{\psi} = U_{\text{enc}}(\boldsymbol{x})\ket{0}^{\otimes n}$, possibly with data reuploading~\cite{perez-salinas_data_2020}.} and measured with an observable $\mathcal{O}$, yielding output functions $f(\boldsymbol{\theta}) = \bra{\psi}\mathcal{U}^{\dagger}(\boldsymbol{\theta})\mathcal{O}\,\mathcal{U}(\boldsymbol{\theta})\ket{\psi}$. Unlike classical models, we cannot directly access intermediate states to apply the chain rule, so dedicated gradient extraction methods are required.

The Hadamard test~\cite{romero_strategies_2018, farhi_classification_2018, abbas_quantum_2023} is the most general approach, computing gradients via controlled unitaries and an ancillary qubit, but the controlled operations are expensive on near-term hardware. The \emph{parameter-shift} rule~\cite{mitarai_quantum_2018} avoids this overhead entirely, computing gradients using the \emph{same circuits} as function evaluation with no ancillae or controlled gates. The most general form handles arbitrary gate generators~\cite{wierichs_general_2022, kyriienko_generalized_2021}; for the common case of gates $U(\theta_j) = e^{i\theta_j G_j}$ where $G_j$ has two distinct eigenvalues (e.g. Pauli matrices), the rule simplifies to
\begin{equation} \label{eqn:parameter_shift_rule}
    \frac{\partial f(\boldsymbol{\theta})}{\partial \theta_j} = \frac{1}{2}\left[f\left(\boldsymbol{\theta} + \tfrac{\pi}{2}\boldsymbol{e}_j\right) - f\left(\boldsymbol{\theta} - \tfrac{\pi}{2}\boldsymbol{e}_j\right)\right]
\end{equation}
requiring just two circuit evaluations per parameter. More generally, the gradient can be written as a linear combination of $R$ shifted function evaluations, $\sfrac{\partial f}{\partial \theta_j} = \sum_{r=1}^{R} y^j_r f^j_r(\boldsymbol{\theta})$, where $R$ depends on the spectrum of the generator (e.g. $R=4$ for beam-splitter gates~\cite{anselmetti_local_2021, coyle_training-efficient_2024}). Crucially, parameter-shift rules yield \emph{exact} derivatives with non-infinitesimal shifts, distinguishing them from finite-difference approximations.

\subparagraph{Quantum backpropagation.} Both the Hadamard test and parameter-shift rule compute each gradient component independently, giving an $\mathcal{O}(N)$ lower bound on gradient evaluation for $N$-parameter models\footnote{More precisely, $\mathcal{O}(N)$ applies per parameter; the full gradient costs $\mathcal{O}(N^2)$ if each shifted circuit is a separate execution. See~\cite{abbas_quantum_2023, bowles_backpropagation_2023} for precise statements.}. Ref.~\cite{abbas_quantum_2023} developed a quantum backpropagation algorithm achieving $\mathcal{O}(\text{polylog}(N))$ gradient scaling, but the construction requires a fault-tolerant device with coherent quantum random-access memory and $\mathcal{O}(\mathrm{poly}(n))$ ancilla overhead, neither of which is available on current hardware. In a complementary direction, Refs.~\cite{bowles_backpropagation_2023, coyle_training-efficient_2024, chinzei_trade-off_2024} identify model families with efficient gradient scaling by restricting the gate set to subgroups that admit classical back-substitution, at the cost of reduced representational capacity.

\subsection{Estimators for quantum gradients}  \label{ssec:quantum_grad_estimators}
Unlike classical models where gradients can often be computed exactly, each circuit evaluation in the parameter-shift rule~\eqref{eqn:parameter_shift_rule} must be estimated from finite measurement samples. Following~\cite{sweke_stochastic_2020}, let $o_m(\boldsymbol{\theta})$ denote the outcome of a single measurement trial, and define the $M$-shot sample mean
\begin{equation} \label{eqn:vqe_loss_functions_estimator}
    o^{M}(\boldsymbol{\theta}) := \frac{1}{M}\sum_{m=1}^M o_{m}(\boldsymbol{\theta}), \qquad \mathbb{E}[o^M] = \bra{\psi}\mathcal{O}\ket{\psi}_{\boldsymbol{\theta}}
\end{equation}
which is an unbiased estimator of the expectation value. For bounded observables $\|\mathcal{O}\| \leq 1$ (e.g.\ a Pauli string, or a normalised Hamiltonian term), each single-shot outcome has variance at most one, so $\operatorname{Var}[o^M] \leq 1/M$, and shot noise is fully controlled by $M$. This is the $\sigma^2_{\text{shot}}/M$ term that will reappear in the forward gradient convergence analysis of~\secref{ssec:convergence_shots}. Each shifted function evaluation in~\eqref{eqn:parameter_shift_rule} is estimated by such a sample mean, so the gradient estimator for the $j$-th component becomes
\begin{equation} \label{eqn:parameter_shift_rule_simple_estimator}
    [\widetilde{\boldsymbol{\nabla}} f(\boldsymbol{\theta})]_j := \sum_{r=1}^{R} y^j_r\, o_r^{j, M}(\boldsymbol{\theta})
\end{equation}
where $o_r^{j,M}$ is the $M$-shot sample mean for the $r$-th shifted circuit of parameter $j$, and unbiasedness follows by linearity. In practice one typically sets $M_r^j := M$ uniformly, though this can be relaxed~\cite{moussa_resource_2023}. By standard Hoeffding arguments, estimating~\eqref{eqn:vqe_loss_functions_estimator} to within precision $\xi$ with probability $1-\delta$ requires $M = \mathcal{O}(\xi^{-2}\log(1/\delta))$ shots. For the full $N$-parameter gradient~\eqref{eqn:parameter_shift_rule_simple_estimator}, each of $N$ components requires $R$ shifted circuits, giving a total cost of $\mathcal{O}(RN\xi^{-2}\log(N/\delta))$ shots. This $\mathcal{O}(N)$ scaling is the fundamental bottleneck of parameter-shift gradient estimation.

\subsection{Approximate gradient estimators: SPSA and RCD} \label{ssec:spsa_rcd}
Two well-known alternatives avoid the explicit $\mathcal{O}(N)$ per-step cost by estimating a single random gradient component or direction rather than the full gradient vector.

\subparagraph{SPSA.} The SPSA algorithm~\cite{spall_stochastic_1987, spall_multivariate_1992, bhatnagar_stochastic_2013} estimates the gradient via a single random perturbation with Rademacher vector $\boldsymbol{u}$ ($u_i \in \{-1, +1\}$):
\begin{equation} \label{eqn:spsa_gradient_rule}
    \widetilde{\boldsymbol{g}}^{\textsf{SPSA}}(\boldsymbol{\theta}) = \frac{f(\boldsymbol{\theta} + \varepsilon\boldsymbol{u}) - f(\boldsymbol{\theta} - \varepsilon\boldsymbol{u})}{2\varepsilon}\, \boldsymbol{u}
\end{equation}
requiring only two circuit evaluations per step, independent of $N$. SPSA has been applied to VQE~\cite{cade_strategies_2020}, quantum natural gradients~\cite{gacon_simultaneous_2021}, combinatorial optimisation~\cite{jain_graph_2022}, and quantum control with very few shots per evaluation~\cite{sauvage_optimal_2020}. However, the $N$-dependence removed from the per-step cost reappears in the estimator variance, which scales linearly in $N$~\cite{abbas_quantum_2023}. The total cost to convergence is therefore comparable to the parameter-shift rule, but in practice SPSA performs well at the scales of variational quantum algorithms studied in the literature thus far~\cite{cade_strategies_2020, bonet-monroig_performance_2023}.

\subparagraph{Random coordinate descent.} RCD~\cite{nesterov_efficiency_2012, richtarik_iteration_2014, ding_random_2024} takes a complementary approach: it selects a single parameter index $j$ uniformly at random and computes the \emph{exact} gradient in that component via the parameter-shift rule:
\begin{equation} \label{eqn:rcd_gradient_estimator}
    \widetilde{\boldsymbol{g}}^{\textsf{RCD}}(\boldsymbol{\theta}) = [\nabla f]_j\, \boldsymbol{e}_j, \qquad j \sim U(\{1, \dots, N\})
\end{equation}
As with SPSA, the per-step saving is offset by $\mathcal{O}(N)$ slower convergence~\cite{ding_random_2024}. Both methods highlight a fundamental tension: removing the explicit $N$-dependence from the gradient rule shifts it to the variance or convergence rate.

\subsection{Adaptive shot allocation: the xCANS family} \label{ssec:iCANS_gCANS}
A complementary approach to reducing measurement cost is to keep the $\mathcal{O}(N)$ per-step structure of the parameter-shift rule but allocate shots \emph{adaptively} across gradient components, spending more on noisy or informative directions and less on well-determined ones. The xCANS family of optimisers~\cite{kubler_adaptive_2020, gu_adaptive_2021} formalises this via the \emph{expected gain} of a gradient step.

\subsubsection{Expected gain} \label{sssec:expected_gain}
For gradient descent with an unbiased estimator $\widetilde{\boldsymbol{\nabla}} \mathcal{L}$ on an $L$-Lipschitz loss, the improvement in the loss from one step to the next can be lower bounded~\cite{balles_coupling_2017}. We define the \emph{gain} $\mathcal{G}$ as this lower bound:
\begin{equation} \label{eqn:gain_definition}
    \mathcal{G} := \eta \boldsymbol{\nabla} \mathcal{L}(\boldsymbol{\theta}^{t})^{\top}\widetilde{\boldsymbol{\nabla}} \mathcal{L}(\boldsymbol{\theta}^{t})  - \frac{L\eta^2}{2}\|\widetilde{\boldsymbol{\nabla}} \mathcal{L}(\boldsymbol{\theta}^{t})\|^2
\end{equation}
satisfying $\mathcal{L}(\boldsymbol{\theta}^t) - \mathcal{L}(\boldsymbol{\theta}^{t+1}) \geq \mathcal{G}$.
Taking the expectation (over data subsampling classically, or measurement shots quantumly) and using unbiasedness, the expected gain becomes
\begin{equation} \label{eqn:expected_gain_icans}
    \mathbb{E}\left[\mathcal{G}\right]
    = \left(\eta - \frac{L\eta^2}{2}\right)\|\boldsymbol{\nabla} \mathcal{L}\|^2
    - \frac{L\eta^2}{2M}\Tr(\Sigma(\boldsymbol{\theta}^t))
\end{equation}
where $\Sigma(\boldsymbol{\theta})$ is the covariance of the gradient estimator and $M$ the number of measurement shots. The gain is positive when $\eta < 2/L$ and the noise term $\Tr(\Sigma)/M$ is sufficiently small relative to $\|\boldsymbol{\nabla}\mathcal{L}\|^2$.

\subsubsection{Shot allocation rules} \label{ssec:xcans}

iCANS~\cite{kubler_adaptive_2020} maximises~\eqref{eqn:expected_gain_icans} \emph{per gradient component}, yielding
\begin{equation} \label{eqn:icans_shot_rule}
        M_i = \frac{2L \eta}{2- L\eta} \frac{\widetilde{\sigma}_i}{\widetilde{g}^2_i}
\end{equation}
where $\widetilde{g}_i$ and $\widetilde{\sigma}_i$ are empirical estimates of the $i$-th gradient component and its variance. Components with large signal-to-noise ratio receive fewer shots; noisy components receive more. The global variant gCANS~\cite{gu_adaptive_2021} couples the allocation across all parameters by replacing the per-component denominator with the full gradient norm:
\begin{equation} \label{eqn:gcans_shot_rule}
        M_i = \frac{2L \eta}{2- L\eta} \frac{\widetilde{\sigma}_i \sum_{n=1}^N \widetilde{\sigma}_n}{\|\widetilde{\boldsymbol{g}}\|^2}
\end{equation}
In practice, both algorithms use bias-corrected exponential moving averages for $\widetilde{g}_i$ and $\widetilde{\sigma}_i$. The critical feature that makes these allocations effective is that measurement variances $\widetilde{\sigma}_i$ differ across parameters in the parameter-shift setting, since each shifted circuit perturbs a different gate.

\subsubsection{Extensions} \label{ssec:rosalin}

Several works extend the xCANS framework by introducing additional sources of adaptivity. Random operator sampling~\cite{arrasmith_operator_2020, sweke_stochastic_2020} distributes shots across terms in a Hamiltonian decomposition $\mathcal{O} = \sum_k c_k O_k$, weighting each term by $|c_k|$; when combined with iCANS, this yields the Rosalin optimiser~\cite{arrasmith_operator_2020}. The doubly-stochastic parameter-shift rule~\cite{sweke_stochastic_2020} goes further, sampling both Hamiltonian terms \emph{and} parameter-shift terms, achieving unbiased gradient estimates from a single measurement shot per step. Refoqus~\cite{moussa_resource_2023} generalises these ideas to data-dependent QML losses (e.g.\ MSE), where shot budgets must also be distributed across data points. Closely related is the metric-aware shot-budgeting analysis of~\cite{vanstraaten_measurement_2021}, which derives an optimal allocation between matrix and vector entries of the quantum natural gradient. A complementary line of work reduces the per-component shot cost not by reallocating shots but by reusing measurement data: classical-shadow methods such as CoVaR~\cite{boyd_training_2022} and adaptive informationally complete POVMs~\cite{garcia-perez_learning_2021} can estimate many observables from a single measurement record, attacking the same $\mathcal{O}(N)$ bottleneck from a different direction. All of these optimisers, however, assume per-parameter gradient computation via the parameter-shift rule.

\tabref{tab:methods} (\secref{sec:quantum_forward}) summarises the methods reviewed in this section together with the forward-gradient framework and the \textsc{quiver} optimiser developed in this work.

\section{Forward gradients} \label{sec:quantum_forward}
The gradient methods reviewed above each pay an $\mathcal{O}(N)$ cost in some form: the parameter-shift rule explicitly, in $N$ separate component evaluations per step; SPSA implicitly, through an $\mathcal{O}(N)$ amplification of the estimator variance from packing the full gradient into a single random perturbation; and RCD through an $\mathcal{O}(N)$ slow-down in convergence from updating one coordinate at a time. In this section we propose \emph{forward gradients}, an estimator class that unifies these methods as special cases of a single random-directional-derivative estimator with a tunable parameter $V$ (the number of random directions sampled per step). The $\mathcal{O}(N)$ trade-off persists, but is consolidated into a single parameter $V$ through which one interpolates between these trade-off extremes, as the remainder of the paper develops. \tabref{tab:methods} maps the landscape: the top block places parameter-shift, SPSA, and RCD as fixed-$(V, \text{direction})$ instances of the estimator; the middle block lists existing shot-adaptive optimisers, all built on the parameter-shift estimator; and the bottom block summarises the \textsc{quiver} family introduced in this work, which adapts $V$ (and $M$) on top of the forward-gradient estimator.

\begin{table*}[t]
\centering
\caption{\textbf{Unified taxonomy of gradient estimators and shot-adaptive optimisers.} Gradient estimators and shot-adaptive optimisers for variational quantum algorithms, framed as instances of the forward-gradient estimator~\eqref{eqn:forward_gradient_estimator}. Top block: gradient estimators, parameterised by $(V, \text{direction distribution})$, with the directional derivative computed via the parameter-shift rule (PSR) for basis directions and central finite difference (FD) for Rademacher directions. Middle block: existing shot-adaptive optimisers, all built on the parameter-shift estimator ($V = N$, basis, PSR). Bottom block: the \textsc{quiver} family introduced in this work, built on the forward-gradient estimator (tunable $V$, Rademacher, FD). Loss type: ``Lin.''~$=$ linear in circuit expectation values; ``Any''~$=$ general. $^\ast$SSD~\cite{pramanik_stochastic_2025} computes each directional derivative \emph{exactly} via an inner product circuit with $\lceil\log_2 N\rceil$ ancilla qubits and controlled gates, rather than by central FD. See~\secref{sec:forward_gradients_quantum_ml} for a full comparison.}
\label{tab:methods}
\footnotesize
\setlength{\tabcolsep}{5pt}
\begin{tabular}{lccccc}
\hline
\textbf{Method} & $\boldsymbol{V}$ & \textbf{Directions $\boldsymbol{v}$} & \textbf{Adapts} & \textbf{Loss} & \textbf{Derivation principle} \\
\hline
\multicolumn{6}{l}{\emph{Gradient estimators}} \\
Parameter-shift~\cite{mitarai_quantum_2018} & $N$ & basis $\boldsymbol{e}_j$ & $-$ & Any & $-$ \\
SPSA~\cite{spall_multivariate_1992} & $1$ & Rademacher & $-$ & Any & $-$ \\
RCD~\cite{ding_random_2024} & $1$ & $N\boldsymbol{e}_j$, $j$ uniform & $-$ & Any & $-$ \\
SSD$^\ast$~\cite{pramanik_stochastic_2025} & $1$ & $\mu{=}0,\,\sigma{=}1$ & $-$ & Lin. & $-$ \\
\textbf{Forward gradient (this work)} & tunable & $\mu{=}0,\,\sigma{=}1$ & $-$ & Any & $-$ \\
\hline
\multicolumn{6}{l}{\emph{Shot-adaptive optimisers, parameter-shift base}} \\
Op.\ sampling~\cite{arrasmith_operator_2020} & $N$ & basis $\boldsymbol{e}_j$ & $M_{\boldsymbol{e}_j}$ per op & Lin. & min-variance op weights \\
Double stoch.~\cite{sweke_stochastic_2020} & $N$ & basis $\boldsymbol{e}_j$ & $M_{\boldsymbol{e}_j}$ $+$ ops & Lin. & random-coord.\ per shift \\
iCANS~\cite{kubler_adaptive_2020} & $N$ & basis $\boldsymbol{e}_j$ & $M_{\boldsymbol{e}_j}$ per param & Lin. & gain per shot \\
gCANS~\cite{gu_adaptive_2021} & $N$ & basis $\boldsymbol{e}_j$ & $M_{\boldsymbol{e}_j}$ per param (global) & Lin. & gain per shot \\
Rosalin~\cite{arrasmith_operator_2020} & $N$ & basis $\boldsymbol{e}_j$ & $M_{\boldsymbol{e}_j}$ $+$ ops & Lin. & gain per shot \\
Refoqus~\cite{moussa_resource_2023} & $N$ & basis $\boldsymbol{e}_j$ & $M_{\boldsymbol{e}_j}$ $+$ batch & Any & gain per shot \\
\hline
\multicolumn{6}{l}{\emph{Shot-adaptive optimisers, forward-gradient base (this work)}} \\
\textbf{\textsc{quiver} fixed-$M$} & tunable & $\mu{=}0,\,\sigma{=}1$ & $V$ & Lin. & gain per shot \\
\textbf{\textsc{quiver} joint $(V, M)$} & tunable & $\mu{=}0,\,\sigma{=}1$ & $V$, $M_{\boldsymbol{v}}$ & Lin. & min-cost at MSE target (CRB) \\
\hline
\end{tabular}
\end{table*}

\subsection{Forward gradients in machine learning} \label{sec:forward_ml}
The forward mode of AD computes Jacobian-vector products $\boldsymbol{\mathcal{J}}\boldsymbol{v} = \nabla_{\boldsymbol{v}} f$ for a given direction $\boldsymbol{v}$. As discussed in~\secref{ssec:gradient_descent}, this is less efficient than backpropagation for the typical $N \gg M$ regime, since each forward pass yields only one directional derivative rather than the full gradient. Nevertheless, a productive line of work~\cite{baydin_gradients_2022, silver_learning_2022, hanzely_sega_2018, flugel_beyond_2024} has shown that a useful gradient estimator can be constructed from a small number of such directional derivatives, with multi-tangent aggregation~\cite{flugel_beyond_2024} explicitly studying $V > 1$ via orthogonal projection of the random tangents:
\begin{equation} \label{eqn:forward_gradient}
    \widetilde{\boldsymbol{\nabla}}^{\mathsf{F}}_{\boldsymbol{v}} f(\boldsymbol{\theta}) := (\nabla_{\boldsymbol{v}} f)\, \boldsymbol{v}
\end{equation}
The full gradient estimator is obtained by averaging over $V$ sampled directions $\mathcal{V} := \{\boldsymbol{v}^\ell\}_{\ell=1}^V$:
\begin{equation} \label{eqn:forward_gradient_estimator}
    \widetilde{\boldsymbol{\nabla}}^{\mathsf{F}} f(\boldsymbol{\theta})
    = \frac{1}{V} \sum_{\ell=1}^V (\nabla_{\boldsymbol{v}^\ell} f)\, \boldsymbol{v}^\ell
\end{equation}
When the components of each $\boldsymbol{v}^\ell$ are drawn independently with zero mean and unit variance, \eqref{eqn:forward_gradient_estimator} is an unbiased estimator of $\boldsymbol{\nabla} f(\boldsymbol{\theta})$~\cite{baydin_gradients_2022}. In the classical setting, the estimator's variance scales as $\mathcal{O}(N)$~\cite{ren_scaling_2023, bos_convergence_2024}, so more directions are needed as the number of parameters grows; this dimension penalty can be partially circumvented by repeated sampling of each direction~\cite{dexheimer_improving_2024}, the classical analogue of the per-direction shot count $M$ in the quantum setting, and the multi-tangent extension that averages over $V > 1$ random directions with explicit orthogonal projection~\cite{flugel_beyond_2024} is the closest classical antecedent of the $V$-interpolation unification of~\secref{ssec:unification}. We establish the corresponding bound for the quantum estimator (with the additional contribution from measurement shot noise) in~\secref{sec:convergence}, which to our knowledge is the first such result for the central-difference forward gradient estimator on parameterised quantum circuits.

A separate line of work addresses variance reduction by exploiting the layered structure of classical neural networks: activity perturbation~\cite{ren_scaling_2023, singhal_how_2023}, local losses with auxiliary per-layer networks~\cite{ren_scaling_2023, fournier_can_2023}, and Jacobian manipulation that reduces both bias and variance by exploiting low-dimensional gradient structure~\cite{wang_towards_2025}. These mitigations do not transfer to quantum circuits, where the loss is computed by a single non-decomposable circuit evaluation~\cite{panchal_cost_2025}: classical forward-mode AD makes a single directional derivative $\mathcal{O}(1)$ at the cost of one function evaluation, whereas no known quantum subroutine recovers this scaling, so the quantum saving comes instead from the $V \ll N$ random-projection reduction we develop below. Two classical extensions are natural follow-ups but lie outside the scope of this work: the second-order forward-mode AD of~\cite{cobb_second_2025} adds directional curvature information at the cost of one hyper-dual-number forward pass per direction, and the random-subspace exact-Generalized-Gauss--Newton construction of SOFO~\cite{yu_sofo_2024} would require a quantum analogue of the in-subspace Hessian inverse. Quantum-side counterparts to these second-order constructions already exist (quantum natural gradient~\cite{stokes_quantum_2020}, simultaneous-perturbation natural gradient (QN-SPSA)~\cite{gacon_simultaneous_2021}, and parameter-shift Hessian estimators~\cite{mari_estimating_2021}), and a head-to-head comparison of a directional-Hessian forward-gradient construction against these baselines is a natural follow-up that we do not pursue here.

\subsection{Existing gradient estimators as special cases} \label{ssec:unification}
The forward gradient framework~\eqref{eqn:forward_gradient_estimator} recovers SPSA, RCD, and the parameter-shift rule as special cases under specific choices of $V$ and the direction distribution. SPSA~\eqref{eqn:spsa_gradient_rule} corresponds to $V = 1$ with a Rademacher direction and central finite-difference directional derivative. RCD~\eqref{eqn:rcd_gradient_estimator} corresponds to $V = 1$ with $\boldsymbol{v} = \boldsymbol{e}_j$ drawn uniformly from the $N$ basis vectors, rescaled by a factor of $N$ to give an unbiased estimate, and the parameter-shift rule for the directional derivative. The full parameter-shift rule~\eqref{eqn:parameter_shift_rule} is then recovered in the deterministic limit where all $N$ basis vectors are enumerated in a single step rather than sampled; this is the $V \to N$ deterministic limit of the framework, in which direction randomness is replaced by full enumeration. As a consistency check, the embedding reproduces the known behaviour of its limiting cases: SPSA's $O(N)$ variance amplification and RCD's $O(N)$ slow-down relative to parameter-shift both emerge directly from the second-moment expansion developed below. The new ingredient the common framing supplies is an explicit $V$ that interpolates between the cheapest single-direction estimator ($V = 1$, highest variance) and the exact full gradient ($V = N$, zero direction variance, highest cost), and it is this lever that the rest of the paper explores, both as a fixed hyperparameter in~\secref{sec:forward_at_scale} and as an adaptive quantity in~\secref{sec:adaptive_optimizer_forward}.

\subsection{Quantum directional derivatives} \label{sec:forward_gradients_quantum_ml}
Applying~\eqref{eqn:forward_gradient_estimator} to quantum circuits requires an efficient method for computing directional derivatives $\nabla_{\boldsymbol{v}} f$. Computing the exact directional derivative via the parameter-shift rule requires evaluating the full gradient first (at cost $\mathcal{O}(RN)$) and then projecting onto $\boldsymbol{v}$, which negates any saving from the forward gradient framework. However, Ref.~\cite{parrish_analytical_2021} observed that for gates of the form $e^{i\theta G}$, loss functions expressible as linear measurements of observables are trigonometric polynomials in the parameters, with maximum angular frequency bounded by $\sqrt{N}$. This band-limitation means that finite-difference approximations with large step sizes $\varepsilon$ can achieve high accuracy using only $Q$ circuit evaluations per direction, with $Q \ll RN$. In particular, defining the finite-difference estimator:
\begin{equation}\label{eqn:directional_derivative_approx_quantum}
    \widetilde{\nabla}_{\boldsymbol{v}}^{\varepsilon} f := \frac{f^M(\boldsymbol{\theta} + \varepsilon\boldsymbol{v}) - f^M(\boldsymbol{\theta} - \varepsilon\boldsymbol{v})}{2\varepsilon}
\end{equation}
where $f^M(\boldsymbol{\theta}') := \sfrac{1}{M}\sum_{m=1}^M o_m(\boldsymbol{\theta}')$ is the $M$-shot estimate of the circuit output at parameters $\boldsymbol{\theta}'$. Ref.~\cite{parrish_analytical_2021} validated this approach numerically, showing that $Q = 2$ (central difference) with step sizes $\varepsilon \in \{0.05, 0.1, 0.2, 0.3\}$ suffices for accurate nuclear gradient estimation in VQE. Higher-order stencils ($Q = 4, 6$) can improve accuracy at the cost of additional circuit evaluations; in this work we use the central difference throughout. Combining~\eqref{eqn:directional_derivative_approx_quantum} with~\eqref{eqn:forward_gradient_estimator}, the quantum forward gradient estimator becomes
\begin{equation}\label{eqn:forward_gradient_approx_quantum}
    \widetilde{\boldsymbol{\nabla}}^{\mathsf{F}} f = \frac{1}{V}\sum_{\ell=1}^V (\widetilde{\nabla}_{\boldsymbol{v}^\ell}^{\varepsilon} f)\, \boldsymbol{v}^\ell
\end{equation}
requiring $2V$ circuit evaluations per gradient step (two per direction for the central difference), each with $M$ shots. The total measurement cost is $2VM$ shots per step, far below the parameter-shift cost of $2RNM$ when $V \ll N$.

The choice of finite-difference step $\varepsilon$ governs a bias--variance trade-off in the central-difference estimator~\eqref{eqn:directional_derivative_approx_quantum}: shot noise is amplified by $1/\varepsilon^2$, while the leading bias scales as $\mathcal{O}(\varepsilon^2)$. We find $\varepsilon = 0.1$ to be the best-performing value on a representative VQE example, and use it for all subsequent experiments. \appref{app:eps_star} derives the corresponding closed-form optimum and shows that it agrees with this empirical choice.

A recent complementary approach~\cite{pramanik_stochastic_2025} (\emph{stochastic shadows descent}) computes $\boldsymbol{\nabla} f \cdot \boldsymbol{v}$ as a single block-encoded subroutine using $N$ controlled copies of the ansatz and $\log_2 N$ ancillas, removing the finite-difference bias. The practical advantage of forward gradients comes from a structural classical fact: forward-mode automatic differentiation computes a directional derivative in one pass at the cost of a single function evaluation, so reconstructing the gradient from $V \ll N$ random projections gives an $N/V$ saving. No known quantum subroutine recovers this scaling; every exact directional-derivative method on a PQC pays $\mathcal{O}(N)$, either as $N$ parameter-shift queries or as the depth-$N$ block encoding of Ref.~\cite{pramanik_stochastic_2025}. The circuit of Ref.~\cite{pramanik_stochastic_2025} produces one exact DD per query, but a $V$-projection reconstruction has total gate count $V N \cdot \mathrm{depth}(U)$, a factor $V$ worse than plain parameter-shift, with equality only at $V = 1$. The bias removal therefore moves the $N$ cost from query count to query depth without saving anything over PS, and pushes the construction past the NISQ regime. The quantum approximate forward gradient estimator proposed here accepts an $\mathcal{O}(\varepsilon^2)$ bias on the directional derivative instead, recovering the classical $V \ll N$ saving with $\mathcal{O}(\mathrm{depth}(U))$ per query; the bias is analytically controlled via the optimal step size $\varepsilon^\star$ that balances finite-difference bias against shot noise (\appref{app:eps_star}).

\section{Convergence of quantum forward gradient descent} \label{sec:convergence}
The reduced per-step cost of the forward gradient estimator~\eqref{eqn:forward_gradient_approx_quantum} naturally raises the question of whether this saving translates to a reduction in \emph{total} measurement cost to reach a fixed accuracy, or whether the increased variance of the estimator compensates. We analyse this in stages, first assuming access to exact directional derivatives, then incorporating the finite-difference approximation and quantum measurement noise.

\subsection{Convergence with exact directional derivatives} \label{ssec:convergence_exact}
Consider the idealised setting where each directional derivative $\nabla_{\boldsymbol{v}^\ell} f = \boldsymbol{v}^\ell \cdot \boldsymbol{\nabla} f$ is computed exactly (no finite-difference error, no shot noise). The forward gradient estimator~\eqref{eqn:forward_gradient_estimator} is then an unbiased estimator of $\boldsymbol{\nabla} f$. Its second moment is characterised by the following lemma (proved in~\appref{app:second_moment}).

\begin{lemma}[Second-moment expansion]\label{lem:exp_over_vectors_norm_grad_squared}
Let $\{\boldsymbol{v}^\ell\}_{\ell=1}^V$ be drawn i.i.d.\ with $\mathbb{E}[v_i^\ell] = 0$, $\mathbb{E}[v_i^\ell v_{i'}^\ell] = \delta_{ii'}$, and let $\kappa := \mathbb{E}[(v_i^\ell)^4]$ be the fourth moment of one component ($\kappa = 1$ for Rademacher, $\kappa = 3$ for standard Gaussian). Then
\begin{equation} \label{eqn:second_moment_exact_fwd}
    \mathbb{E}_{\boldsymbol{v}}\left[\|\widetilde{\boldsymbol{\nabla}}^{\mathsf{F}} f\|^2\right] = \frac{N + V + \kappa - 2}{V}\|\boldsymbol{\nabla} f\|^2.
\end{equation}
\end{lemma}

Throughout the remainder of this paper we use Rademacher directions ($\kappa = 1$) unless stated otherwise, so the prefactor reduces to $(N + V - 1)/V$. The convergence rates for forward gradient descent follow from \lemref{lem:exp_over_vectors_norm_grad_squared} and are stated as a proposition (proved in~\appref{app:pl_convergence_fwd}).

\begin{proposition}[Forward gradient convergence]\label{prop:fwd_convergence}
Let $f \in C^1(\mathbb{R}^N)$ be $L$-smooth and satisfy the PL inequality with constant $\mu > 0$, and let directions be drawn from the Rademacher distribution ($\kappa = 1$).

\emph{(i) Convergence of forward gradient descent.} With learning rate $\eta = V / (L(N + V - 1))$, forward gradient descent satisfies
\begin{multline} \label{eqn:fwd_convergence_exact}
    \mathbb{E}[f(\boldsymbol{\theta}^{(T)})] - f^* \leq \\
    \left(1 - \frac{\mu V}{L(N{+}V{-}1)}\right)^{\!T} \!(f(\boldsymbol{\theta}^{(0)}) - f^*)
\end{multline}

\emph{(ii) Convergence of quantum forward gradient descent.} With per-direction shot noise variance $\sigma^2_{\emph{shot}}$ and learning rate $\eta \in [0, 1/(2\mu)]$, the expected suboptimality satisfies
\begin{multline} \label{eqn:convergence_noisy}
    \mathbb{E}[f(\boldsymbol{\theta}^{(T)})] - f^* \leq \left(1 - 2\mu\eta\right)^T (f_0 - f^*) \\
    + \frac{L\eta}{4\mu} \cdot \frac{(N{+}V{-}1)\,\sigma^2_{\emph{shot}}}{VM}
\end{multline}
\end{proposition}

\begin{proof}[Proof sketch]
The forward gradient estimator is unbiased (\propref{prop:fwd_unbiased}). \lemref{lem:exp_over_vectors_norm_grad_squared} bounds its second moment by $\beta^2 \|\nabla f\|^2$ with $\beta^2 = (N+V-1)/V$ (using $\kappa=1$). Substituting into the standard PL-SGD bound~\cite{wolf_mathematical_2023, sweke_stochastic_2020} yields~(i). For~(ii), adding the constant shot-noise contribution $\sigma^2_{\mathrm{shot}}/M$ to the second moment gives a bounded-variance estimator; the bounded-variance form of the same SGD bound then yields the asymptotic floor directly. The formal derivation is in~\appref{app:pl_convergence_fwd}.
\end{proof}

To reach $f(\boldsymbol{\theta}^{(T)}) - f^* \leq \delta$ under~(i), we therefore require
\begin{equation} \label{eqn:steps_exact_fwd}
    T = \mathcal{O}\left(\frac{N}{V} \cdot \frac{L}{\mu} \log \frac{1}{\delta}\right)
\end{equation}
steps, where the last equality holds for $V \ll N$. Since each step costs $\mathcal{O}(V)$ circuit evaluations, the total circuit cost is
\begin{equation} \label{eqn:total_cost_exact_fwd}
    T \times V = \mathcal{O}\left(\frac{L N}{\mu} \log \frac{1}{\delta}\right)
\end{equation}
which is independent of $V$. This is the no-free-lunch result for forward gradients: the $V$-fold saving in per-step cost is exactly compensated by a $V$-fold increase in the number of steps required, and the total circuit evaluation budget is the same as for the full gradient ($V = N$). In this exact setting, the choice of $V$ affects only the granularity of progress (many cheap steps vs few expensive ones), not the total cost.

\begin{figure*}[!t]
    \centering
    \includegraphics[width=\textwidth]{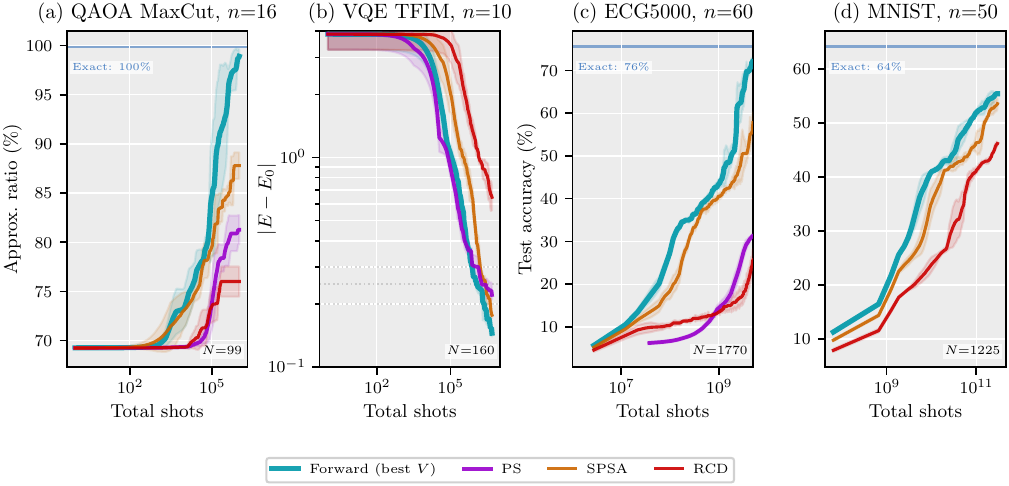}
    \caption{\textbf{Forward gradients match or exceed all baselines across four problem domains: quantum optimisation, quantum simulation, time-series classification, and image classification.} Each panel compares the best forward gradient configuration against the parameter-shift rule (PS), SPSA ($V = 1$) and RCD on the total measurement shots axis. All stochastic methods are run at a matched per-step shot budget $B = 1000$; PS uses $M = 10$ shots per parameter and its per-step cost therefore grows as $2NM$. (a)~QAOA MaxCut, $n = 16$, depth 3, $N = 99$, best forward $V = 10$. (b)~VQE ground-state estimation (transverse-field Ising model, $n = 10$, $d = 8$), $N = 160$, best forward $V = 10$. (c)~ECG5000 time-series classification, $n = 60$, $N = 1770$ parameters, best forward $V = 25$. (d)~MNIST image classification, $n = 50$, $N = 1225$, best forward $V = 10$. Forward gradients reach near-optimal performance on all four domains, at a fraction of the PS shot cost on the three panels (a, b, c) for which a PS baseline is shown, while SPSA and RCD degrade at large $N$. The MNIST panel (d) reports forward, SPSA, and RCD curves; PS is omitted from this panel because its per-step cost $2NM$ is computationally prohibitive at these system sizes, and the PS-vs-forward comparison on MNIST is reported separately at $n \in \{10, 20, 40\}$ on a $10\times$-subsampled training set in \figref{fig:mnist_merged}. Shaded bands: one standard deviation over three random seeds (all four panels). The QAOA panel is shown on a random instance of a weighted Erd\H{o}s--R\'enyi graph. All methods use the Adam optimiser at identical hyperparameters; only the gradient estimator differs.}
    \label{fig:hero}
\end{figure*}

\begin{figure}[!t]
    \centering
    \includegraphics[width=\columnwidth]{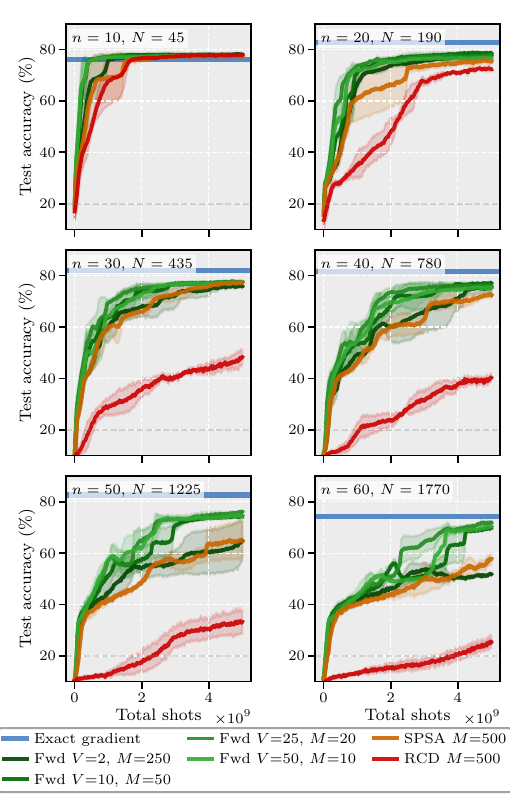}
    \caption{\textbf{Forward gradients consistently outperform SPSA and RCD at matched shot budget.} ECG5000: forward gradient methods ($V > 1$) vs SPSA ($V = 1$) and RCD. Each panel shows a different system size $n$. Forward gradients with $V = 10$ consistently outperform single-direction methods.}
    \label{fig:ecg_spsa_rcd}
\end{figure}

\begin{figure}[!t]
    \centering
    \includegraphics[width=\columnwidth]{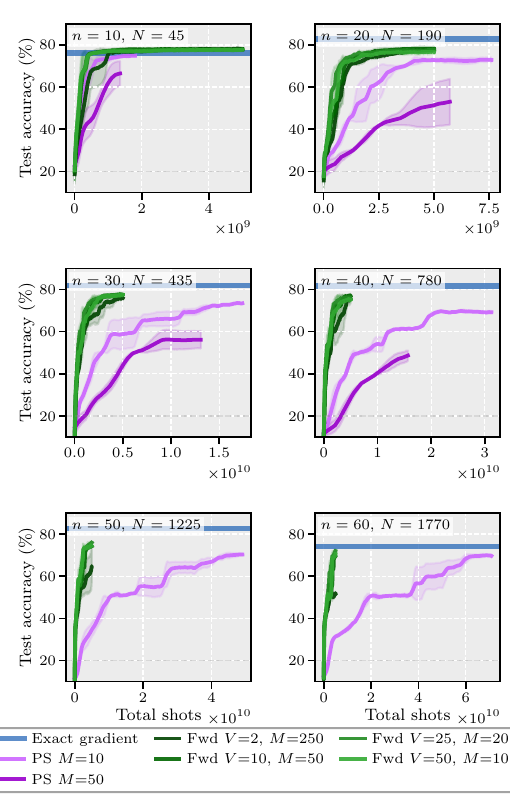}
    \caption{\textbf{Forward gradients reach parameter-shift accuracy at a fixed, $N$-independent shot cost.} ECG5000: forward gradient methods vs parameter-shift rule (PS) with $M$ shots per parameter. Forward methods use a fixed per-step budget independent of $N$, yet reach comparable accuracy. We do not include the larger-shot $M = 50$ baseline at $n = 50, 60$ due to the extensive runtimes required ($2NM$ shots per step for $N \in \{1225, 1770\}$).}
    \label{fig:ecg_ps}
\end{figure}

\subsection{Convergence under quantum measurement noise} \label{ssec:convergence_shots}
In practice, each directional derivative is estimated from $M$ measurement shots, contributing a per-direction variance $\sigma^2_{\text{shot}}/M$ to the estimator, where $\sigma^2_{\text{shot}} := \mathbb{E}_m[\operatorname{Var}_{\boldsymbol{v}}[\widetilde{\nabla}_{\boldsymbol{v}} f_m]]$. \propref{prop:fwd_convergence}(ii) captures this: the first term of~\eqref{eqn:convergence_noisy} is the optimisation error, which decreases geometrically and is independent of how the shot budget is allocated, while the second term is an asymptotic error floor (a residual suboptimality set by the shot noise that cannot be averaged away within a single gradient step).

The floor is proportional to $(N + V + \kappa - 2)/(VM)$. To compare methods at equal cost, we impose a fixed total per-step shot budget $B$ and allocate it as $M = B/(2V)$ shots per direction, so that each of the $V$ directions contributes two circuit evaluations each estimated with $M$ shots. Substituting into the floor gives
\begin{equation} \label{eqn:noise_floor}
    \frac{L\eta}{4\mu} \cdot \frac{2(N{+}V{+}\kappa{-}2)\,\sigma^2_{\text{shot}}}{B}
\end{equation}
For $V \ll N$ the $+V$ term in the numerator is negligible, so the floor reduces to approximately $\nicefrac{L\eta}{4\mu} \cdot \nicefrac{2N\sigma^2_{\text{shot}}}{B}$ with a correction of order $V/N$; crucially, this expression is independent of $V$, since increasing $V$ reduces the direction-averaging variance by $1/V$ while the corresponding decrease in $M = B/(2V)$ raises the per-direction shot noise by the same factor, and the two effects cancel exactly.

This is the no-free-lunch result: at fixed per-step shot budget, no choice of $V$ lowers the asymptotic shot-noise residual. The experimental savings reported in~\secref{sec:forward_at_scale} and~\secref{sec:v_scheduling_results} therefore cannot come from reducing this residual; they arise before it is reached, where the non-convex landscape still provides strong gradient signal and a larger $V$ finds better descent directions early in training.

\section{Forward gradients at scale} \label{sec:forward_at_scale}
To validate these practical savings, we train Hamming-weight preserving orthogonal circuits~\cite{landman_quantum_2022, monbroussou_trainability_2023, coyle_training-efficient_2024} on ECG5000 time-series classification ($n$ up to 60, $N$ up to 1770 parameters) and MNIST image classification ($n$ up to 50, $N$ up to 1225), comparing forward gradient estimators against the parameter-shift rule, SPSA, and RCD. All stochastic methods (forward, SPSA, RCD) are run at a fixed per-step shot budget of $B = 1000$ shots, so that comparisons between them are at matched measurement cost; the parameter-shift rule is run with $M = 10$ shots per parameter, so its per-step cost grows as $2NM$ and is \emph{not} matched. Results are mean accuracy over three random seeds, reported against total shots used. \figref{fig:hero} summarises the results, confirming empirically the cartoon behaviour of \figref{fig:schematic}(c). Full hyperparameters are in~\appref{app:hyperparams}.

\begin{figure}[t]
    \centering
    \includegraphics[width=\columnwidth]{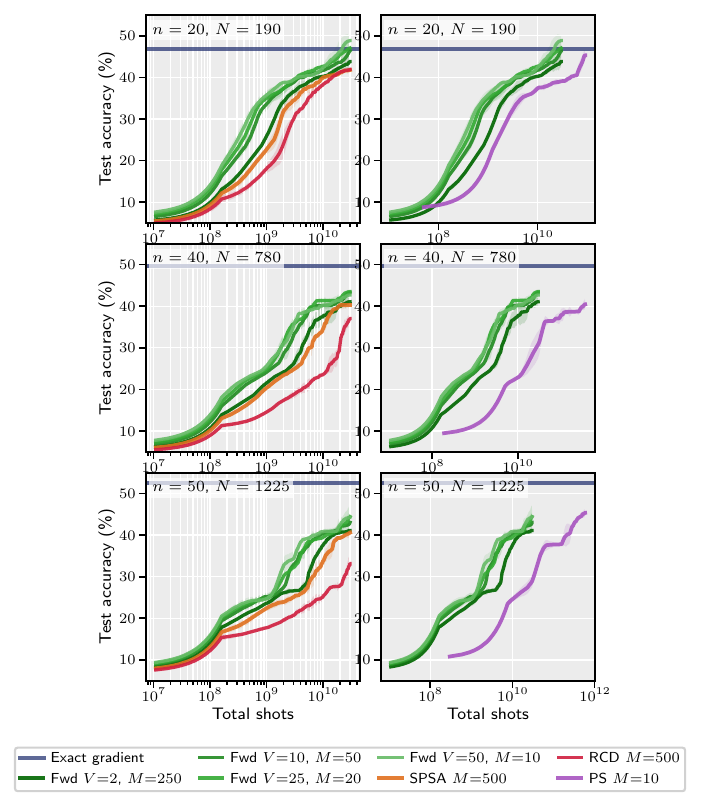}
    \caption{\textbf{MNIST image classification reproduces the ECG5000 ordering across system sizes.} MNIST 10-class classification reproduces the ECG5000 ordering on image data. Rows: system sizes $n$. Left column: forward gradient family vs SPSA and RCD; right column: forward gradients vs PS ($M = 10$ shots per parameter), trained on a $10\times$-subsampled training set so the larger PS comparison is tractable. Shaded bands: $\pm 1\sigma$ across three random seeds. Same hyperparameters and matched $B = 1000$-shot per-step budget as \figref{fig:hero}.}
    \label{fig:mnist_merged}
\end{figure}

\subsection{ECG5000} \label{ssec:ecg_results}
\figref{fig:ecg_spsa_rcd} compares forward gradient methods against SPSA and RCD at a fixed shot budget per step. Forward gradients with $V \geq 10$ outperform single-direction methods ($V = 1$), with the gap widening at larger $N$. At $n = 60$ ($N = 1770$), SPSA collapses to 63.8\% while forward methods with $V = 25$ match the exact gradient baseline at 76.0\%.

\figref{fig:ecg_ps} compares forward methods against the parameter-shift rule. Despite using a fixed per-step budget (independent of $N$), forward gradients reach comparable or superior accuracy to PS at all system sizes. The accuracy gap between forward and exact gradients remains small (1--2 percentage points) up to $n = 60$, while the total shot cost is a fraction of that required by PS. The best-performing forward configuration uses $V = 10$ at most system sizes, with $V = 25$ preferred at $n = 60$.

\subparagraph{MNIST.} On MNIST image classification the same ordering holds (\figref{fig:mnist_merged}, \appref{app:hyperparams}): across system sizes $n \in \{10, 20, 40, 50\}$ ($N$ up to $1225$), forward gradients with $V = 10$ track the exact-gradient baseline at the matched $B = 1000$-shot per-step budget, while SPSA and RCD degrade with $n$ and the parameter-shift-versus-forward shot ratio scales linearly in $N$.
\subsection{Ground-state estimation and combinatorial optimisation} \label{ssec:vqe_qaoa_results}
Beyond classification, we validate forward gradients on two further problem classes for which the optimisation landscape and the measurement noise are qualitatively different. The ground-state estimation benchmark is a transverse-field Ising model (TFIM) on $n = 10$ qubits with $d = 8$ layers ($N = 160$ parameters), trained to minimise the energy expectation value of a hardware-efficient ansatz. The combinatorial optimisation benchmark is a depth-3 QAOA for MaxCut on 16-vertex weighted Erd\H{o}s--R\'enyi graphs with $N = 99$ parameters, trained to maximise the approximation ratio. Both are shown in panels (c) and (d) of \figref{fig:hero}, under the same matched per-step budget $B = 1000$ convention as the classification benchmarks.

On VQE, forward gradients with $V = 10$ reach a mean best-over-training energy error $|E_{\mathrm{best}} - E_{\mathrm{exact}}| = 0.115 \pm 0.010$ at a $5\times 10^6$-shot budget (mean and standard deviation over three random seeds; energy readout from a $10^4$-shots/term noisy oracle excluded from the training budget), compared with $0.203 \pm 0.021$ for Adam-optimised parameter-shift at $M=10$ shots per parameter over three seeds and $0.215$ for SPSA at the same budget (single-seed reference configuration).

On QAOA, forward gradients with $V = 10$ reach an approximation ratio of $0.988$ at a $5\times 10^6$-shot budget, within $\sim 0.01$ of the exact-gradient baseline and well ahead of SPSA at the same budget. At this system size ($N = 99$), the parameter-shift rule remains a viable comparison point, but its per-step cost $2NM$ grows rapidly with system size and becomes the dominant factor at the classification scales of~\secref{ssec:ecg_results}. The four panels of \figref{fig:hero} span ground-state estimation, combinatorial optimisation, and supervised learning, and the matched-budget forward gradient estimator is competitive with or better than the best available baseline on every one.

\section{Adaptive direction scheduling} \label{sec:v_scheduling_results}
\subsection{Motivation} \label{ssec:v_schedule_motivation}
The results of the previous section show that forward gradients with a fixed number of directions $V$ can train quantum circuits at scale with practical shot savings. A question this raises is whether $V$ itself can be adapted during training.

The convergence analysis of~\secref{sec:convergence} provides the motivation. The noise amplification factor $(N + V + \kappa - 2)/V$ in the gain expression~\eqref{eqn:second_moment_exact_fwd} decreases with $V$, suggesting that fewer directions yield more precise gradient estimates per shot. However, the number of steps to convergence grows as $N/V$ (\eqref{eqn:steps_exact_fwd}), favouring more directions. The optimal balance between these competing effects shifts during training: early on, when the gradient is large and the loss landscape has many viable descent directions, a large $V$ provides broad coverage at low cost per direction. As training progresses and the gradient shrinks, precision matters more than coverage, and concentrating shots on fewer directions becomes preferable.

This suggests a simple scheduling strategy: decay $V$ from a large initial value $V_{\max}$ to a small final value $V_{\min}$ over the course of training, while keeping the total shot budget per step $B = 2VM$ fixed, so that as $V$ decreases, the shots-per-direction $M = B/(2V)$ increases automatically. We consider three monotone schedule shapes, parameterised by training step $t \in [0, T]$:
\begin{flalign}
&\text{linear:} \quad V_t = V_{\max} - (V_{\max} - V_{\min})\,\tfrac{t}{T}, & \label{eqn:vsched_linear} \\
&\text{cosine:} \quad V_t = V_{\min} + \tfrac{1}{2}(V_{\max} - V_{\min})\left(1 + \cos\!\tfrac{\pi t}{T}\right), & \label{eqn:vsched_cosine} \\
&\text{exponential:} \quad V_t = V_{\min}\left(\tfrac{V_{\max}}{V_{\min}}\right)^{\!(T - t)/T}. & \label{eqn:vsched_exp}
\end{flalign}
In each case $V_t$ is rounded to the nearest integer and clamped to $[V_{\min}, V_{\max}]$. The per-step shot cost $B = 2V_t M_t$ is unchanged; only the allocation between the number of directions and shots per direction varies.

\subsection{Results} \label{ssec:v_schedule_results}
We illustrate this across four benchmarks (\figref{fig:v_schedule}). In each case, a decreasing $V$-schedule outperforms both fixed-$V$ endpoints at matched shot budget: exponential decay on VQE TFIM ($n = 10$, $d = 8$, $N = 160$), linear decay on QAOA MaxCut ($n = 16$, $N = 99$), and linear or cosine decay on ECG5000 and MNIST ($n = 50$, $N = 1225$). The primary driver of the improvement is the direction of decay rather than the precise functional form, though the optimal shape and range of $V_{\max} \to V_{\min}$ are problem-dependent and the same schedule does not achieve the same margin on every benchmark. These results are presented as a heuristic illustration rather than a systematic study: the observed improvement is problem- and configuration-dependent, and we make no claim that a decaying $V$-schedule will consistently outperform fixed-$V$ endpoints in general. Rather, the observation that the training signal itself drives a natural preference for fewer directions as training progresses motivates the \textsc{quiver} optimiser of \secref{sec:adaptive_optimizer_forward}, which determines $V$ automatically from the measurement data.

\begin{figure*}[t]
    \centering
    \includegraphics[width=\textwidth]{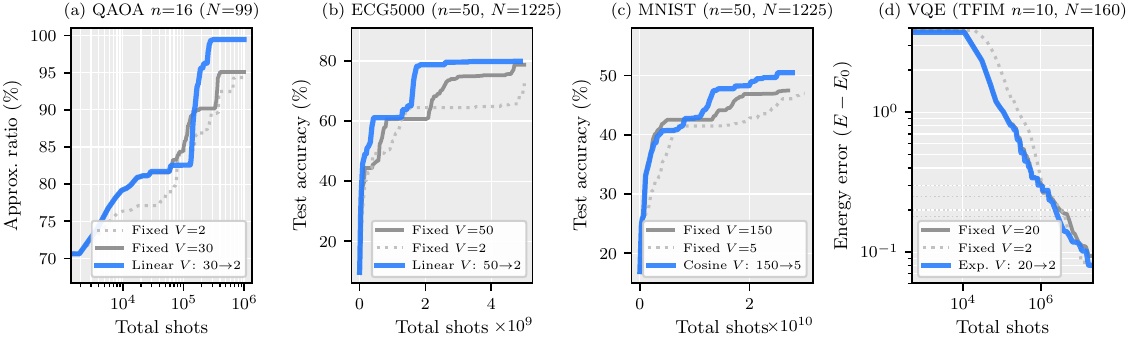}
    \caption{\textbf{Decreasing $V$-schedules outperform fixed-$V$ endpoints across all problem domains.} Each panel shows a scheduled $V$ (blue) against the two fixed-$V$ endpoints (grey solid = $V_{\max}$, grey dashed = $V_{\min}$). The scheduled curve reaches higher performance at matched shot budgets than either endpoint on all four benchmarks: (a) QAOA MaxCut ($n=16$, $N=99$), linear $V: 30\to2$; (b) ECG5000 ($n=50$, $N=1225$), linear $V: 50\to2$; (c) MNIST ($n=50$, $N=1225$), cosine $V: 150\to5$; (d) VQE TFIM ($n=10$, $N=160$), exponential $V: 20\to2$. The direction of decay is the primary driver of the improvement; the optimal functional form and $V$-range are problem-dependent.}
    \label{fig:v_schedule}
\end{figure*}

\section{Quantum-native optimisation for forward gradients} \label{sec:adaptive_optimizer_forward}
The forward gradient results of~\secref{sec:forward_at_scale} and~\secref{sec:v_scheduling_results} treat the direction vectors as purely random: drawn from an isotropic Rademacher distribution, with $V$ either fixed or scheduled by a classical rule that ignores the measurement outcomes themselves. A natural question, in the spirit of the CABS~\cite{balles_coupling_2017} and iCANS/gCANS~\cite{kubler_adaptive_2020, gu_adaptive_2021} families of shot-adaptive optimisers for the parameter-shift rule, is whether the \emph{quantum} measurement statistics generated during training can themselves be used to guide the estimator. In this section we answer that question by constructing \textsc{quiver} (\textbf{Qu}antum \textbf{I}terative \textbf{V}-adaptive \textbf{E}stimator \textbf{R}efinement), an adaptive optimiser for forward gradients. The derivation that follows assumes a loss linear in circuit expectation values, as is the case for the VQE and QAOA benchmarks of~\secref{sec:forward_at_scale}; an extension to non-linear losses (such as the MSE and cross-entropy classification losses), in the spirit of Refoqus~\cite{moussa_resource_2023}, is left to future work.

The construction proceeds in five steps. In~\secref{ssec:gain_forward_gradients} we derive the expected gain of a forward gradient step, extending the per-component gain decomposition of iCANS to the directional-derivative setting. The resulting expression exposes two adaptive levers: the number of shots $M$ allocated to each direction, and the number of directions $V$ itself. In~\secref{ssec:noise_concentration} we show that the per-direction shot lever (the natural forward-gradient analogue of iCANS) is neutralised by a concentration-of-measure phenomenon that is structural to random-direction estimators and absent from the parameter-shift setting. In~\secref{ssec:quiver} we turn to the remaining lever and find that the natural ``fixed-$M$, optimise $V$ for gain-per-shot'' construction yields a closed-form rule with a signal-independent floor: it picks a reasonable constant $V$ automatically but does not adapt within a run. We then decouple $V$ and $M$ into two independent targets, a per-direction signal budget and an absolute estimator-variance budget, and derive a joint closed-form for $(V^\star_t, M^\star_t)$ that genuinely adapts. \secref{ssec:quiver_optimality} re-derives this rule from a single optimisation principle (minimum shot cost at a target estimator accuracy with a per-direction shot floor) and shows that the resulting allocation (\thmref{thm:joint_optimum}) uses Rademacher directions that minimise variance among isotropic distributions (\propref{prop:rademacher_min}) and matches the Cram\'er--Rao lower bound on gradient recovery from a shot-noise oracle up to a vanishing constant (\corref{cor:crb}). Finally \secref{ssec:quiver_validation} validates the joint rule on VQE TFIM and compares it head-to-head against fixed parameter-shift, iCANS, gCANS, and the fixed-$M$ variant at matched total shot budget.

\subsection{Gain expression for forward gradients} \label{ssec:gain_forward_gradients}
Our starting point is a per-direction decomposition of the expected gain~\eqref{eqn:gain_definition}, which makes the contribution of each random direction explicit as a competition between a signal term and a noise term. This is the forward-gradient analogue of the per-component gain decomposition used by CABS~\cite{balles_coupling_2017} and iCANS~\cite{kubler_adaptive_2020} for classical and parameter-shift gradients respectively.

\begin{lemma}[Per-direction gain decomposition]\label{lem:per_direction_gain}
Substituting \lemref{lem:exp_over_vectors_norm_grad_squared} into~\eqref{eqn:gain_definition} yields:

\emph{(i)} $\mathbb{E}_m\mathbb{E}_{\boldsymbol{v}}[\mathcal{G}^{\mathsf{F}}]$ decomposes as
\begin{multline} \label{eqn:gain_decomposition}
     \mathbb{E}_m\mathbb{E}_{\boldsymbol{v}}\left[\mathcal{G}^{\mathsf{F}}\right]
    \approx \frac{1}{V}\sum_{\ell=1}^V \bigg(\eta \|\boldsymbol{\nabla} \mathcal{L}\|^2 \\
    - \frac{L\eta^2}{2}\frac{N + V + \kappa - 2}{V} \left((\nabla_{\boldsymbol{v}^{\ell}}\mathcal{L})^2
    + \frac{1}{M}\operatorname{Var}_m\left[\widetilde{\nabla}_{\boldsymbol{v}^{\ell}}\mathcal{L}_{m}\right]\right) \\
    \bigg) =: \frac{1}{V}\sum_{\ell=1}^V \gamma_{\boldsymbol{v}^\ell}.
\end{multline}

\emph{(ii)} The gain-per-shot optimum $M^*_{\ell} := \argmax_{M_\ell}[\gamma_{\boldsymbol{v}^\ell}/M_\ell]$ is
\begin{multline} \label{eqn:maximising_meas_for_each_direction_main}
   M^*_\ell = \frac{2\,\mathrm{Var}_m\left[\widetilde{\nabla}_{\boldsymbol{v}^\ell}\mathcal{L}_m\right]}{C\|\boldsymbol{\nabla}\mathcal{L}\|^2 - (\nabla_{\boldsymbol{v}^\ell}\mathcal{L})^2}, \\
   C := \tfrac{2V}{L\eta(N+V+\kappa-2)}.
\end{multline}
\end{lemma}

Each direction $\ell$ contributes a signal term $(\nabla_{\boldsymbol{v}^{\ell}}\mathcal{L})^2$ and a noise term $\nicefrac{1}{M}\operatorname{Var}_m[\widetilde{\nabla}_{\boldsymbol{v}^{\ell}}\mathcal{L}_{m}]$ to the gain penalty, amplified by $(N + V + \kappa - 2)/V$. In~\eqref{eqn:maximising_meas_for_each_direction_main}, directions better aligned with the gradient receive more shots (smaller denominator); in practice the unknown true quantities are replaced by empirical estimates (see \appref{app:meas_alloc} for positivity conditions). The allocation~\eqref{eqn:maximising_meas_for_each_direction_main} is the forward-gradient analogue of iCANS: shots are distributed across directions in proportion to the ratio of per-direction variance to a signal-dependent denominator. The next subsection shows why this per-direction allocation is structurally inert for forward gradients, at which point we return to the remaining lever (the number of directions $V$) in~\secref{ssec:quiver}.

\subsection{Why iCANS-style allocation fails for forward gradients} \label{ssec:noise_concentration}
The allocation~\eqref{eqn:maximising_meas_for_each_direction_main} has the same structure as the iCANS allocation~\cite{kubler_adaptive_2020}: shots are distributed according to the ratio of per-direction measurement variance (numerator) to a signal-dependent denominator. For iCANS, both quantities vary across parameters, because each canonical direction $\boldsymbol{e}_j$ shifts a single gate while all others remain fixed, producing quantum states that depend on which parameter is being differentiated.

For forward gradients, each direction $\boldsymbol{v}^\ell$ perturbs all $N$ parameters simultaneously with i.i.d.\ components, and the measurement variance of the resulting directional derivative is, under smoothness assumptions on the parameterised expectation value, a Lipschitz function of $\boldsymbol{v}^\ell$. Classical concentration-of-measure results for Lipschitz functions of independent bounded random variables then yield the following heuristic, which we treat as a working assumption and verify empirically below.

\begin{assumption}[Noise concentration across random directions]\label{assump:noise_concentration_main}
Let $\boldsymbol{v}^\ell$ have i.i.d.\ components (Rademacher or Gaussian) and let $\sigma^2_{\nabla,\ell} := \operatorname{Var}_m[\widetilde{\nabla}_{\boldsymbol{v}^{\ell}}\mathcal{L}_m]$ denote the per-direction measurement variance. If $\sigma^2_{\nabla,\ell}$ is Lipschitz in $\boldsymbol{v}^\ell$ with Lipschitz constant independent of $N$, then
\begin{equation} \label{eqn:noise_concentration_main}
    \sigma^2_{\nabla,\ell} \approx \bar{\sigma}^2_{\nabla}\qquad \text{for all } \ell,
\end{equation}
with $\bar{\sigma}^2_{\nabla} := \mathbb{E}_{\boldsymbol{v}}[\sigma^2_{\nabla,\ell}]$. The Lipschitz condition ensures $\sigma^2_{\nabla,\ell}$ concentrates around $\bar\sigma^2_\nabla$ with fluctuations $\mathcal{O}(L_\sigma)$ independent of $N$. We prove this for local Hamiltonians and bounded-depth ans\"atze in \appref{app:noise_concentration_proof} and verify it empirically below.
\end{assumption}

\begin{figure}[t]
    \centering
    \includegraphics[width=0.9\linewidth]{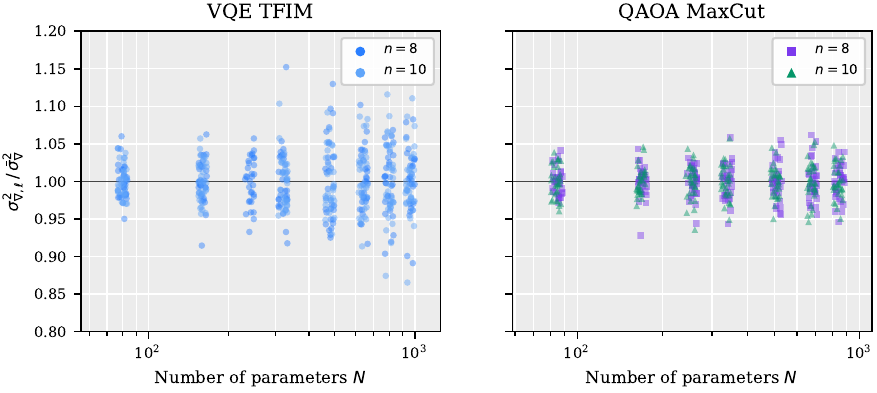}
    \caption{\textbf{Directional-derivative variance concentrates uniformly across random directions.} For each configuration ($V_{\mathrm{probe}} = 32$ Rademacher directions; $\varepsilon = 0.1$, $M = 50$, $R = 5000$), the 32 individual per-direction variances $\sigma^2_{\nabla,\ell}$ are plotted normalised by their configuration mean $\bar\sigma^2_\nabla$. The solid line at $y = 1$ is the configuration mean; all points lie within $\pm 15\%$ of unity, and the spread is typically below $\pm 10\%$. Left: VQE TFIM ($n = 8, 10$). Right: QAOA MaxCut ($n = 8, 10$). Each series sweeps $N$ up to $\sim 10^3$ by increasing depth $d$. The tight clustering around unity, across all $N$ and both problem types, confirms \assumpref{assump:noise_concentration_main}.}
    \label{fig:noise_concentration}
\end{figure}

\noindent Because each $\boldsymbol{v}^\ell$ averages over $N$ i.i.d.\ perturbations, the per-gate variance structure that iCANS exploits is washed out. The numerator of~\eqref{eqn:maximising_meas_for_each_direction_main} becomes approximately $\bar{\sigma}^2_\nabla$ for all $\ell$, leaving only the denominator as a source of variation. But the denominator variation comes from the signal $g_\ell^2 = (\nabla_{\boldsymbol{v}^\ell}\mathcal{L})^2$, which for random directions is a random projection of the gradient rather than a structured per-parameter quantity. The allocation reduces to $M^*_\ell \approx \bar{\sigma}^2_\nabla / (\text{const} - g_\ell^2)$, which provides only weak, stochastic variation compared to the structured numerator-driven allocation of iCANS.

We verify this empirically in \figref{fig:noise_concentration}. For each configuration we draw $V_{\mathrm{probe}} = 32$ independent random directions at the initial parameters and estimate $\operatorname{Var}_m[\widetilde{\nabla}_{\boldsymbol{v}^\ell}\mathcal{L}_m]$ for each by repeating the noisy central-difference estimator $R = 5000$ times at $M = 50$ shots per evaluation, with $\varepsilon = 0.1$. Series cover VQE TFIM and QAOA MaxCut at $n \in \{8, 10\}$ each, with layer depth $d$ increasing along each series to sweep $N$ up to $\sim 10^3$. For each configuration we plot the 32 individual values $\sigma^2_{\nabla,\ell}$ normalised by their mean $\bar\sigma^2_\nabla$; the points cluster tightly around unity (within $\pm 10\%$ typically) for all $N$ and all four series, confirming $\sigma^2_{\nabla,\ell} \approx \bar\sigma^2_\nabla$ regardless of problem type or qubit count. Substituting~\eqref{eqn:noise_concentration_main} into the per-direction gain makes the structure explicit:
\begin{equation} \label{eqn:gain_signal_noise_main}
    \gamma_{\boldsymbol{v}^{\ell}} = \eta \|\boldsymbol{\nabla} \mathcal{L}\|^2 - A\,g_\ell^2 - \frac{A\,\bar{\sigma}^2_{\nabla}}{M_\ell}, \quad A := \frac{L\eta^2(N{+}V{+}\kappa{-}2)}{2V}.
\end{equation}
The noise penalty (third term) is identical for all $\ell$, while the signal cost (second term) varies through $g_\ell^2$. Since the only adaptive lever in the allocation (the numerator) has been removed by concentration, a different strategy is needed: rather than allocating shots non-uniformly across directions, we adapt the \emph{number of directions itself}. This is the remaining degree of freedom in the forward gradient estimator, and the one that the $V$-scheduling heuristic of~\secref{sec:v_scheduling_results} already exploits empirically. In the next subsection we derive a principled rule for this adaptation.

\subsection{The \textsc{quiver} optimiser} \label{ssec:quiver}
The per-direction allocation of~\eqref{eqn:maximising_meas_for_each_direction_main} having been neutralised by the noise concentration, the remaining adaptive lever is the number of directions $V$ itself, and optionally the scalar per-direction shot count $M$. Each epoch, the optimiser maintains two exponential moving averages at no extra measurement cost: $\widehat{g}^2_t$ of the squared forward gradient norm $\|\widetilde{\boldsymbol{\nabla}}^{\mathsf F}\mathcal{L}\|^2$ and $\widehat{\sigma}^2_t$ of the mean per-direction measurement variance $\bar\sigma_\nabla^2$. We use these EMAs to update $(V_t, M_t)$ in two stages: a gain-per-shot closed-form that yields a reasonable but static allocation, and a two-target refinement that recovers genuine within-run adaptation.

\subparagraph{Gain-per-shot optimum at fixed $M$.} With $M$ held fixed, the gain expression~\eqref{eqn:gain_signal_noise_main} is a function of $V$ alone.

\begin{lemma}[Fixed-$M$ optimal $V$]\label{lem:fixed_m_v_star}
Taking the expectation of $\gamma_{\boldsymbol v^\ell}$ over isotropic directions under \assumpref{assump:noise_concentration_main} and maximising $\mathbb{E}[\mathcal{G}^{\mathsf F}]/(2VM)$ over $V$, the first-order condition yields
\begin{equation} \label{eqn:adaptive_v_rule}
    V^{\star}
    =
    \frac{2L\eta\,(M\|\boldsymbol{\nabla}\mathcal{L}\|^2 + \bar\sigma^2_\nabla)(N+\kappa-2)}
         {M\|\boldsymbol{\nabla}\mathcal{L}\|^2(2 - L\eta) - L\eta\,\bar\sigma^2_\nabla},
\end{equation}
clamped to $[V_{\min}, V_{\max}] = [2, N]$, with $V^\star = V_{\max}$ when the denominator is non-positive (proved in~\appref{app:fixed_m_v_star_proof}).
\end{lemma}

$V^\star$ scales with $N + \kappa - 2$, reflecting the $O(N)$ variance amplification of forward gradients; both inputs are quantities already returned by the forward gradient step; and $M$ is the single hyperparameter, with per-step cost $B_t = 2V_t M$ varying naturally with $V_t$.

\subparagraph{The fixed-$M$ optimum does not adapt.} On VQE TFIM, per-epoch logging of $V_t$ from~\eqref{eqn:adaptive_v_rule} shows that $V_t$ converges to a fixed value on the first update step and does not move thereafter, because at fixed $M$ the optimum~\eqref{eqn:adaptive_v_rule} has a high-SNR floor
\begin{equation} \label{eqn:v_floor}
    V^{\star}_{\mathrm{floor}} = \frac{2L\eta(N + \kappa - 2)}{2 - L\eta}
\end{equation}
which depends only on $(L, \eta, N)$, not on $\widehat{g}^2_t$ or $\widehat{\sigma}^2_t$. The formula leaves this floor only when $\widehat{\sigma}^2 \gtrsim M\widehat{g}^2$, a condition that never holds during VQE training at the conservative learning rate $\eta = 3\times 10^{-3}$ required for $N \geq 80$ convergence: across $500$ logged training epochs on VQE TFIM $n = 8$, $d = 10$ the signal-to-noise ratio $M\widehat{g}^2/\widehat{\sigma}^2$ decreases from $\approx 4300$ to $\approx 1400$, an order-of-magnitude drop but still well above the threshold for adaptation. At $L\eta = 0.045$, \eqref{eqn:v_floor} evaluates to $V \approx 15$ at $N = 320$, which is precisely the steady-state value observed in the run data.

The fixed-$M$ rule therefore sets a constant $V$ calibrated to $(L, \eta, N)$ with no additional tuning, and, as we show in the large-$N$ panel of \figref{fig:regime_crossover}, this alone is enough to outperform iCANS/gCANS at the same shot budget when $L\eta \ll 1$. It does not, however, respond to training dynamics: within a run, $V_t$ is flat.

\subparagraph{The joint $(V, M)$ rule.} When the gradient signal dominates the measurement noise ($\widehat{\sigma}^2 \ll M\widehat{g}^2$), holding $M$ fixed collapses the gain-per-shot objective to a function of $(L, \eta, N)$ alone, so $V$ never moves in response to training dynamics. Allowing $M$ to vary jointly with $V$ restores this dependence: the pair $(V^\star_t, M^\star_t)$ is then the solution of a minimum-cost allocation problem whose inputs are the current signal and noise estimates, derived in \secref{ssec:quiver_optimality} (\thmref{thm:joint_optimum}),
\begin{equation} \label{eqn:joint_vm}
    V^{\star}_t = \frac{(N - 1 + \alpha)\, \widehat{g}^2_t}{\tau^2},
    \qquad
    M^{\star}_t = \frac{N\, \widehat{\sigma}^2_t}{\alpha\, \widehat{g}^2_t}
\end{equation}
with two hyperparameters: a dimensionless ratio $\alpha > 0$ that controls the per-direction shot count and a target absolute variance $\tau^2 > 0$ for the reconstructed gradient estimate. In this form the qualitative behaviour of the rule is immediate: $M^\star$ \emph{grows} as training reduces the signal (since $\widehat{g}^2_t$ decreases), while $V^\star$ \emph{shrinks} with $\widehat{g}^2_t$. Noise dependence is preserved in $M^\star$ through $\widehat{\sigma}^2_t$, and $V^\star$ tracks $\widehat{g}^2_t$ directly. Both are computed once per epoch from the EMAs already maintained by the optimiser. We take~\eqref{eqn:joint_vm} as the defining update of \textsc{quiver}; the fixed-$M$ rule~\eqref{eqn:adaptive_v_rule} is recovered as a special case when $M$ is held constant, and as we show in~\secref{ssec:quiver_optimality}, both rules are instances of the same minimum-cost allocation principle with different constraint choices. The closed form $V^\star_t \propto (N - 1 + \alpha)\widehat{g}^2_t/\tau^2$ inherits its $N$-dependence from both the explicit $N$ prefactor and the gradient-norm scaling of the deployed problem; substituting the empirically observed $\widehat{g}^2 \propto N^{-1/2}$ scaling for our trained PQCs recovers the empirical $V^\star \propto \sqrt{N}$ trend reported in our $V$-sweep experiments without further tuning.

Because $\widehat{g}^2_t$ and $\widehat{\sigma}^2_t$ are noisy early in training, unclamped updates to $V_t$ and $M_t$ can swing by an order of magnitude in a single epoch; we therefore apply a multiplicative clamp on each update and hold $(V_t, M_t)$ fixed during a short initial warmup while the estimates stabilise. Neither modification changes the fixed point of the rule. The full update is in \algoref{alg:quiver_joint}.

\begin{algorithm}[H]
\caption{\textsc{quiver} joint $(V, M)$ optimiser}
\label{alg:quiver_joint}
\begin{algorithmic}[1]
\Require learning rate $\eta$, number of parameters $N$, targets $(\alpha, \tau^2)$, warmup $T_w$, EMA decay $\beta$, rate-limit bounds $(r_{\downarrow}, r_{\uparrow}) = (0.7, 1.5)$, clamps $[V_{\min}, V_{\max}]$, $[M_{\min}, M_{\max}]$, initial $(V_0, M_0)$, finite-difference step $\varepsilon$
\State Initialise $\boldsymbol{\theta}_0$, Adam state $\phi_0$; $V \gets V_0$, $M \gets M_0$
\State $\widehat{g}^2 \gets \mathrm{None}$; $\widehat{\sigma}^2 \gets \mathrm{None}$
\For{$t = 0, 1, 2, \ldots$}
    \State Sample $V$ random directions $\{\boldsymbol{v}^{(\ell)}\}_{\ell=1}^{V}$ (Rademacher)
    \State Measure directional derivatives
    \Statex \hfill $d_\ell \gets \bigl[\widetilde{\mathcal{L}}^M(\boldsymbol{\theta}_t + \varepsilon \boldsymbol{v}^{(\ell)}) - \widetilde{\mathcal{L}}^M(\boldsymbol{\theta}_t - \varepsilon \boldsymbol{v}^{(\ell)})\bigr] / (2\varepsilon)$
    \State $\widetilde{\boldsymbol{\nabla}}\mathcal{L} \gets \frac{1}{V}\sum_{\ell=1}^{V} d_\ell\, \boldsymbol{v}^{(\ell)}$ \Comment{reconstructed gradient}
    \State $g^2_t \gets \lVert\widetilde{\boldsymbol{\nabla}}\mathcal{L}\rVert^2$; $\sigma^2_t \gets \mathrm{Var}_\ell(d_\ell)$
    \If{$\widehat{g}^2 = \mathrm{None}$}
        \State $\widehat{g}^2 \gets g^2_t$; $\widehat{\sigma}^2 \gets \sigma^2_t$
    \Else
        \State $\widehat{g}^2 \gets \beta\, \widehat{g}^2 + (1-\beta)\, g^2_t$
        \State $\widehat{\sigma}^2 \gets \beta\, \widehat{\sigma}^2 + (1-\beta)\, \sigma^2_t$
    \EndIf
    \State $(\boldsymbol{\theta}_{t+1}, \phi_{t+1}) \gets \mathrm{Adam}(\boldsymbol{\theta}_t, \phi_t, \widetilde{\boldsymbol{\nabla}}\mathcal{L}; \eta)$
    \If{$t \ge T_w$} \Comment{update $(V, M)$ after warmup}
        \State $M^{\star} \gets N \widehat{\sigma}^2 / (\alpha\, \widehat{g}^2)$
        \State $V^{\star} \gets (N - 1 + \alpha)\, \widehat{g}^2 / \tau^2$
        \State $V \gets \mathrm{clip}\!\bigl[V^{\star},\; \max(r_{\downarrow} V, V_{\min}),\; \min(r_{\uparrow} V, V_{\max})\bigr]$
        \State $M \gets \mathrm{clip}\!\bigl[M^{\star},\; \max(r_{\downarrow} M, M_{\min}),\; \min(r_{\uparrow} M, M_{\max})\bigr]$
    \EndIf
\EndFor
\end{algorithmic}
\end{algorithm}

\subsection{Optimal allocation} \label{ssec:quiver_optimality}
The joint rule~\eqref{eqn:joint_vm} was stated in the previous subsection as the defining update of \textsc{quiver}, but not derived. We now derive it from a single optimisation principle: \emph{among all $(V, M)$ pairs that achieve a prescribed gradient-estimation accuracy with at least a minimum shot count per direction, choose the one that uses the fewest total measurements.} For any zero-mean unit-variance direction distribution the minimisation yields a rule of the same form, with the component kurtosis $\kappa$ appearing in place of $1$; substituting Rademacher directions ($\kappa = 1$, \propref{prop:rademacher_min}) gives exactly~\eqref{eqn:joint_vm}. The derivation has three parts. First, among isotropic independent-component direction distributions, Rademacher uniquely minimises the estimator second moment at fixed $V$ (\propref{prop:rademacher_min}). Second, solving the cost-minimisation problem over $(V, M)$ when the accuracy constraint is tight gives the joint rule~\eqref{eqn:joint_vm} in closed form (\thmref{thm:joint_optimum}); the two hyperparameters $(\alpha, \tau^2)$ correspond to a reparameterisation of the minimum-shot floor and the accuracy target respectively. Third, the resulting shot budget matches the Cram\'er--Rao lower bound for unbiased gradient estimation from a shot-noise oracle, up to a constant that vanishes in $N$ (\corref{cor:crb}). Parameter-shift saturates the same CRB on the same oracle, so both estimators are CRB-optimal and the practical advantage of~\secref{sec:forward_at_scale} is a constant-factor gap within that class, not an information-theoretic separation.

\subparagraph{Rademacher is variance-optimal.}
Write $\widehat{\boldsymbol{g}}_V$ for the forward gradient estimator~\eqref{eqn:forward_gradient_estimator} with exact directional derivatives and $\boldsymbol{g} := \boldsymbol{\nabla}\mathcal{L}(\boldsymbol{\theta})$. The squared error $\mathbb{E}\|\widehat{\boldsymbol{g}}_V - \boldsymbol{g}\|^2$ depends on the choice of direction distribution; the following proposition identifies which distribution minimises it at fixed $V$.

\begin{proposition}[Minimum fourth-moment direction distribution] \label{prop:rademacher_min}
Let $p$ be an isotropic distribution on $\mathbb{R}^N$ with independent components of zero mean and unit variance, and write $\kappa := \mathbb{E}[v_i^4]$ for the component kurtosis. Then
\begin{equation} \label{eqn:var_direction_formula}
    \mathbb{E}\|\widehat{\boldsymbol{g}}_V - \boldsymbol{g}\|^2 = \frac{\|\boldsymbol{g}\|^2\,(\kappa + N - 2)}{V}.
\end{equation}
By $\mathbb{E}[v_i^4] \ge 1$ (Jensen, with equality iff $|v_i|$ is almost-surely constant), Rademacher directions ($v_i \in \{\pm 1\}$) uniquely minimise the right-hand side at fixed $V$, giving $\mathbb{E}\|\widehat{\boldsymbol{g}}_V - \boldsymbol{g}\|^2 \ge \|\boldsymbol{g}\|^2(N-1)/V$ (proved in~\appref{app:rademacher_proof}).
\end{proposition}
\noindent The result is stated within the class of isotropic \emph{independent-component} distributions with unit second moment, which contains Rademacher and standard Gaussian and excludes unit-norm distributions such as directions drawn uniformly from the sphere $S^{N-1}$ (the sphere distribution has dependent components). The classical forward-gradient literature~\cite{baydin_gradients_2022, ren_scaling_2023} uses both families interchangeably; in the independent-component class Rademacher is optimal and this is the class we work with throughout.

\subparagraph{Minimum-cost allocation at target accuracy.}
We now allow the directional derivatives to be estimated with $M$ shots each, so the estimator picks up an additional shot-noise contribution of $\bar\sigma^2_\nabla N/(VM)$ (derived in the same way as the diagonal term above with $\operatorname{Var}_m[d_\ell] = \bar\sigma^2_\nabla$). Using Rademacher directions, the mean-squared error of the full estimator is
\begin{equation} \label{eqn:total_var}
    \operatorname{MSE}(V, M)
    :=
    \mathbb{E}\|\widehat{\boldsymbol{g}}_V - \boldsymbol{g}\|^2
    = \frac{\|\boldsymbol{g}\|^2(N-1)}{V} + \frac{\bar\sigma^2_\nabla N}{VM}.
\end{equation}
We seek the allocation that minimises the total per-step measurement cost $2VM$ subject to two practical requirements: the estimator MSE must not exceed a user-specified target $\tau^2$, and the per-direction shot count must not fall below a minimum $M_{\min}$. Formally,
\begin{align} \label{eqn:cost_problem}
    \min_{V, M \in \mathbb{R}_{>0}} \;& 2 V M \notag\\
    \text{s.t.}\;\; & \operatorname{MSE}(V, M) \le \tau^2, \\
    & M \ge M_{\min}. \notag
\end{align}

\begin{theorem}[Optimal $(V, M)$ allocation] \label{thm:joint_optimum}
For any $\tau^2 > 0$, $M_{\min} \ge 1$, and $\|\boldsymbol{g}\|^2 > 0$, the continuous relaxation of~\eqref{eqn:cost_problem} has the unique minimiser
\begin{equation} \label{eqn:joint_optimum_full}
    M^{\star} = M_{\min},
    \qquad
    V^{\star} = \frac{(N - 1)\|\boldsymbol{g}\|^2 + N\bar\sigma^2_\nabla / M_{\min}}{\tau^2},
\end{equation}
with both constraints tight, and minimum cost
\begin{equation} \label{eqn:joint_cost_minimum}
    C^{\star} = 2 V^{\star} M_{\min} = \frac{2\bigl[(N - 1)\|\boldsymbol{g}\|^2 M_{\min} + N \bar\sigma^2_\nabla\bigr]}{\tau^2}.
\end{equation}
Writing $\alpha := N\bar\sigma^2_\nabla / (M_{\min}\|\boldsymbol{g}\|^2)$, \eqref{eqn:joint_optimum_full} becomes
\begin{equation} \label{eqn:joint_optimum_alpha}
    V^\star = \frac{(N - 1 + \alpha)\,\|\boldsymbol{g}\|^2}{\tau^2}, \qquad
    M^\star = \frac{N\,\bar\sigma^2_\nabla}{\alpha\,\|\boldsymbol{g}\|^2},
\end{equation}
which is exactly the joint update of~\eqref{eqn:joint_vm}. Thus $(\alpha, \tau^2)$ is a reparameterisation of $(M_{\min}, \tau^2)$, and the joint rule is the constructive minimum-cost allocation (proved in~\appref{app:joint_optimum_proof}).
\end{theorem}

Three remarks. First, $\tau^2$ in~\eqref{eqn:joint_optimum_alpha} is treated as a hyperparameter, tuned once per problem as part of a small $(\alpha, \tau^2)$ grid. Second, $(V^\star, M^\star)$ is the unique minimiser: every other allocation satisfying the two constraints has shot cost strictly greater than $C^\star$. Third, \thmref{thm:joint_optimum} treats $(\bar\sigma^2_\nabla, \|\boldsymbol{g}\|^2)$ as fixed inputs, but both change over training. The algorithm of~\secref{ssec:quiver} re-solves the optimum each epoch with the current EMA estimates $(\widehat\sigma^2_t, \widehat{g}^2_t)$ in place of $(\bar\sigma^2_\nabla, \|\boldsymbol{g}\|^2)$. At every epoch $(V_t, M_t)$ is the optimum for the EMA estimates of that epoch.

\subparagraph{Matching the Cram\'er--Rao lower bound.}
The optimal cost~\eqref{eqn:joint_cost_minimum} is a construction: it tells us the minimum shot budget \emph{achievable} by a forward-gradient estimator meeting the two practical requirements. A natural complementary question is the information-theoretic \emph{lower bound} on shot cost for \emph{any} unbiased estimator of $\boldsymbol{g}$ that queries a shot-noise oracle. Comparing the two tells us whether the joint rule is near-optimal among all possible estimators, not just within its own class. The Cram\'er--Rao inequality provides the lower bound.

\begin{corollary}[CRB-level optimality] \label{cor:crb}
Let a quantum shot oracle return i.i.d.\ outcomes with mean $f(\boldsymbol{\theta}+\boldsymbol{s})$ and variance $\sigma_m^2$ for shifts with $\|\boldsymbol{s}\|^2 \le \varepsilon^2 N$. Any unbiased estimator of $\boldsymbol{g} := \nabla f(\boldsymbol{\theta})$ from $B$ queries satisfies
\begin{equation} \label{eqn:crb_bound}
    \mathbb{E}\|\widehat{\boldsymbol{g}} - \boldsymbol{g}\|^2 \ge \frac{N \sigma_m^2}{B\, \varepsilon^2}.
\end{equation}
In the bias-neutral limit, the forward-gradient estimator at the optimal allocation of \thmref{thm:joint_optimum} achieves $B^\star / B_{\mathrm{CRB}} = (N - 1 + \alpha)/N \to 1$ as $N \to \infty$ (proved in~\appref{app:crb_proof}).
\end{corollary}

The saturation rate $B^\star/B_{\mathrm{CRB}} = (N - 1 + \alpha)/N$ is tight at finite $N$: at $\alpha = 0.2$ it evaluates to $0.985, 0.9925, 0.9950, 0.9963$ at $N = 80, 160, 240, 320$ respectively, so the gap to the information-theoretic lower bound is below $1.5\%$ across all parameter counts at which we report experiments and below $0.5\%$ for $N \geq 240$.

\noindent\textbf{Remark.} \corref{cor:crb} is not a separation result: the same Fisher-information calculation applied to parameter-shift gives $B_{\mathrm{PS}} \ge 4N\sigma_m^2/(\pi^2\tau^2)$, matching~\eqref{eqn:crb_bound} up to a constant, so parameter-shift is equally CRB-optimal on this oracle and both estimators lie on the same information-theoretically optimal surface. Together, \thmref{thm:joint_optimum}, \propref{prop:rademacher_min}, and \corref{cor:crb} give \textsc{quiver} a constructive derivation from a measurement-cost optimisation problem.

\subsection{Validation} \label{ssec:quiver_validation}
\begin{figure*}[t]
    \centering
    \includegraphics[width=\textwidth]{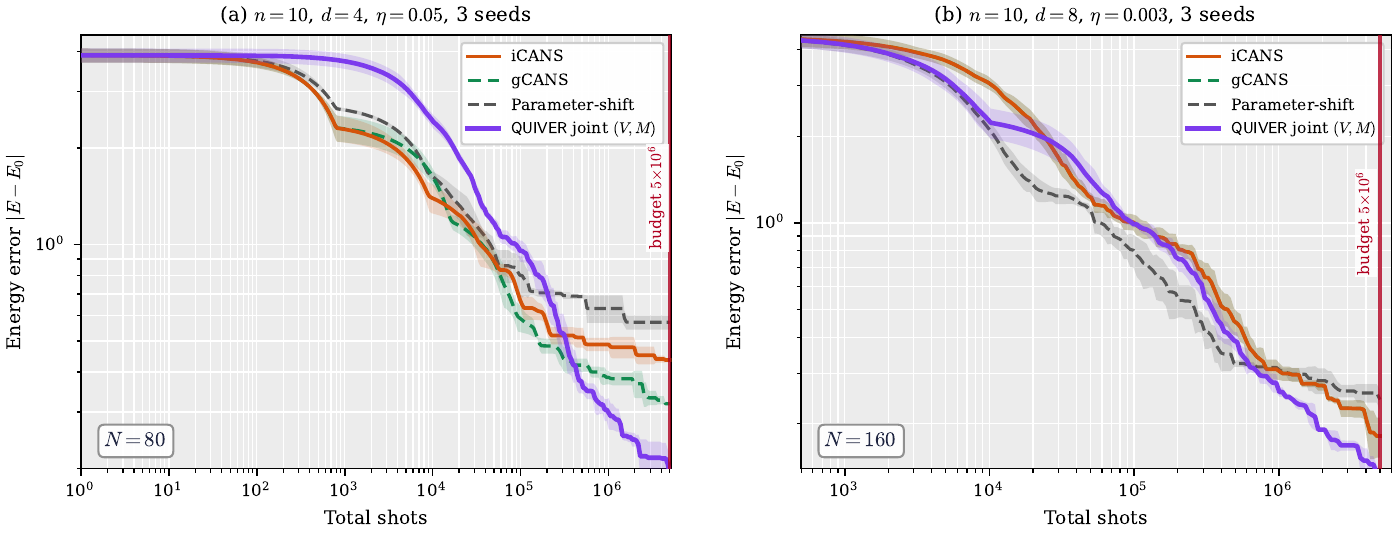}
    \caption{\textbf{\textsc{quiver} outperforms iCANS and gCANS across the $(N,\eta)$ regime crossover.} VQE TFIM $n = 10$, matched $5\times 10^6$-shot budget, three random seeds ($\pm 1\sigma$ bands); the red vertical line marks the training cutoff. (a) $N = 80$, $\eta = 0.05$ ($L\eta \approx 0.95$): iCANS/gCANS allocate above $s_{\min}$ and genuinely adapt; \textsc{quiver} leads all methods. (b) $N = 160$, $\eta = 0.003$ ($L\eta \approx 0.06$): the $L\eta$-dependent prefactor in the iCANS/gCANS shot rule pins the allocation at $s_{\min}$, so both reduce to plain parameter-shift with identical trajectories; \textsc{quiver} leads by a margin well outside the seed-to-seed spread.}
    \label{fig:regime_crossover}
\end{figure*}

\begin{figure}[t]
    \centering
    \includegraphics[width=0.95\columnwidth]{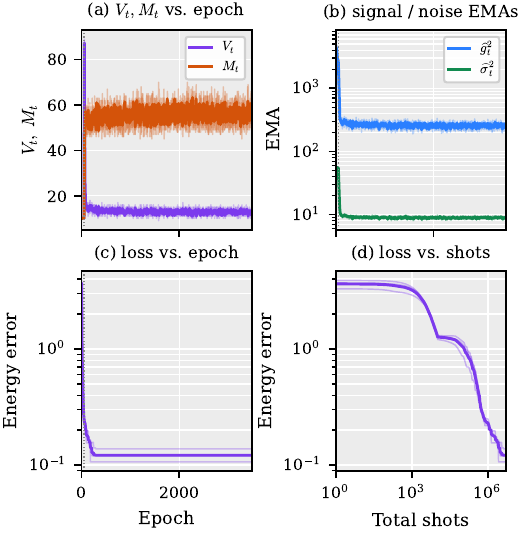}
    \caption{\textbf{\textsc{quiver} trajectory on the large-$N$ VQE benchmark.} TFIM $n = 8$, $d = 20$ ($N = 320$), three random seeds (faded lines) and their mean (bold). (a) Both $V_t$ and $M_t$ adapt after the $50$-epoch warmup, spanning an order-of-magnitude dynamic range over the run. (b) The signal EMA $\widehat{g}^2_t$ decays while the noise EMA $\widehat{\sigma}^2_t$ is nearly flat, driving the joint update. (c,d) Training loss on epoch and shot axes: monotone descent across all seeds.}
    \label{fig:quiver_joint_traj}
\end{figure}

\begin{figure}[t]
    \centering
    \includegraphics[width=0.95\columnwidth]{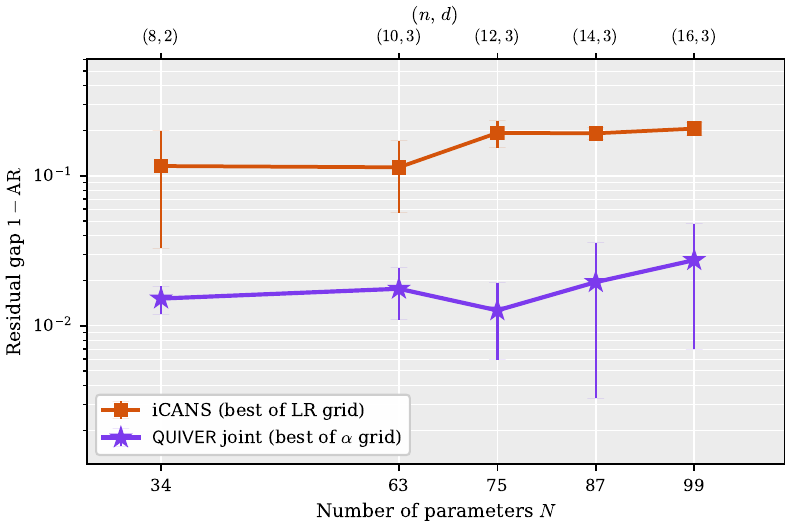}
    \caption{\textbf{\textsc{quiver} on QAOA MaxCut.} Residual gap $1 - \mathrm{AR}$ (log scale, $\pm 1\sigma$) vs.\ parameter count $N$ at matched $10^6$-shot budget, averaged over five weighted Erd\H{o}s--R\'enyi graphs and three seeds. \textsc{quiver} leads iCANS at every $N$ tested and maintains tight across-sample bands throughout, confirming the result is not specific to VQE or the TFIM Hamiltonian.}
    \label{fig:qaoa_crossover}
\end{figure}

We validate the joint rule~\eqref{eqn:joint_vm} in two steps. First, we compare \textsc{quiver} head-to-head against iCANS, gCANS, and plain parameter-shift (with Adam optimiser) at two $(N, \eta)$ operating points: $(N=80, \eta=0.05)$ where iCANS and gCANS can fully adapt, and $(N=160, \eta=0.003)$ where they cannot (\figref{fig:regime_crossover}); \textsc{quiver} outperforms the full PS-based family in both regimes. Second, we inspect the trajectory of $(V_t, M_t)$ during training: both adapt over the run, and the rule is stable across a smooth hyperparameter surface (\figref{fig:quiver_joint_traj}).

\subparagraph{Comparison against iCANS and gCANS.} The iCANS/gCANS shot rule~\eqref{eqn:icans_shot_rule} contains an $L\eta$ prefactor that collapses the allocation to $s_{\min}$ whenever $L\eta \ll 1$; \textsc{quiver}'s joint $(V, M)$ rule~\eqref{eqn:joint_vm} has no such prefactor. In \figref{fig:regime_crossover}a, at $(N{=}80, \eta{=}0.05)$ with $L\eta \approx 0.95$, iCANS and gCANS allocate above $s_{\min}$ and adapt, but \textsc{quiver} still reaches a lower final energy error. In \figref{fig:regime_crossover}b, at $(N{=}160, \eta{=}0.003)$ with $L\eta \approx 0.06$ ($\eta$ is the largest stable value at this depth, and $N{=}160$ is the minimum at which exact gradient descent reaches the target; see \figref{fig:n_scaling}), iCANS and gCANS collapse onto $s_{\min}$ and reduce to plain parameter-shift, while \textsc{quiver} maintains a clear margin. The same ordering holds across a broader $N$ sweep on $n{=}8$ and $n{=}10$ TFIM (\appref{app:icans_baselines}, \figref{fig:n_scaling}) and on QAOA MaxCut (\figref{fig:qaoa_crossover}).

\subparagraph{Trajectory of $(V_t, M_t)$.} \figref{fig:quiver_joint_traj} shows the joint rule on the large-$N$ VQE benchmark (TFIM $n = 8$, $d = 20$, $N = 320$) at the best configuration from the $(\alpha, \tau^2)$ grid ($\alpha = 0.2$, $\tau^2 = 6400$). After the $50$-epoch warmup, $V_t$ responds to the decay of $\widehat{g}^2_t$ and $M_t$ responds to the ratio $\widehat{\sigma}^2_t / \widehat{g}^2_t$; both are shown in \figref{fig:quiver_joint_traj}a,b. The training loss descends monotonically on both the epoch and shot axes, confirming that the warmup period and multiplicative clamp absorb the per-step EMA noise without suppressing adaptation. We selected $(\alpha, \tau^2)$ from a $3\times 3$ hyperparameter sweep (\appref{app:icans_baselines}).

\subparagraph{Robustness to iCANS hyperparameter tuning.} For QAOA on weighted MaxCut the iCANS Lipschitz constant $L$ is upper-bounded by the cost-Hamiltonian operator norm $L_{\mathrm{op}} := \sum_{ij}|w_{ij}|$, the safe choice of~\cite{kubler_adaptive_2020} used in \figref{fig:qaoa_crossover}; tuning $L$ below this bound voids the convergence guarantee.\footnote{We use the iCANS shot-allocation rule with an Adam parameter update rather than Algorithm~2's per-parameter learning-rate cap $\eta_i \leq \chi_i^2/(L(\chi_i^2 + \zeta_i/s_i))$; at $L = L_{\mathrm{op}}$, $\eta = 0.1$ this cap would force $\eta_i \lesssim 0.03$, so the deviation favours iCANS.} We swept iCANS at $n = 16$, $d = 3$ over $\eta \in \{0.01, 0.05, 0.1, 0.3\}$ and $L / L_{\mathrm{op}} \in \{1, 1/3, 1/10\}$, and raised the per-parameter shot ceiling from $s_{\max} = 200$ to $1000$. No configuration beat the safe baseline: across the grid, the shot count saturates at $s_{\max}$ within a few epochs and iCANS reduces to fixed-shot parameter-shift.

\section{Discussion and outlook} \label{sec:conclusion}
Gradient estimation by the parameter-shift rule has been a structural bottleneck in training parameterised quantum circuits: its per-step measurement cost scales linearly in the number of parameters and dominates the total shot budget at any practical scale. In this work we replaced that fixed $\mathcal{O}(N)$ overhead with a continuously tunable lever. Random-direction forward gradients reduce the per-step cost from $\mathcal{O}(N)$ to $\mathcal{O}(V)$ for any $V$ between $1$ (SPSA) and $N$ (full parameter-shift), and the \textsc{quiver} optimiser sets that lever automatically from the measurement statistics already produced during training. The result is a single estimator family that recovers SPSA, RCD, and parameter-shift as limiting cases of the same construction, and a single optimiser that interpolates between them in response to the running signal-to-noise ratio.
Empirically, a fixed $V \ll N$ forward estimator reaches the accuracy of parameter-shift at a fraction of the shot cost across all four benchmarks tested (ECG5000 and MNIST classification, VQE TFIM, QAOA MaxCut), with the saving growing as $N$ increases. \textsc{quiver}'s joint $(V, M)$ rule then outperforms iCANS, gCANS, and Adam-optimised parameter-shift at matched shot budget on VQE TFIM and QAOA MaxCut, including the conservative $L\eta \ll 1$ regime where the iCANS/gCANS allocation collapses onto $s_{\min}$ and reduces to fixed-shot parameter-shift.
On the hardware side, forward-gradient estimators evaluate the circuit only at shifted parameters $\boldsymbol{\theta} \pm \varepsilon \boldsymbol{v}^\ell$ and require no ancillas, mid-circuit measurements, or modified gate sets beyond what parameter-shift already needs. The $V$ directional derivatives per step are independent circuit jobs and run in parallel across QPUs, so wall-clock cost is set by per-step shots rather than by $V$. The framework is compatible with existing PQC training pipelines and transfers directly to any variational algorithm whose loss is linear in circuit expectation values.
Open questions remain in three directions. The convergence analysis of \secref{sec:convergence} shows that the per-step saving cannot be converted into an asymptotic speedup on convex losses, yet the empirical savings on non-convex quantum learning landscapes grow with $N$, and a sharper characterisation of when and why this happens would close that gap. The noise-concentration result underpinning \textsc{quiver} is proved here only for local Hamiltonians on bounded-depth ansätze (\appref{app:noise_concentration_proof}), leaving open the regimes of chemistry Hamiltonians in the molecular-orbital basis, long-range spin systems, and non-translation-invariant problems. Finally, generalising the \textsc{quiver} gain decomposition to non-linear losses (such as negative log-likelihood, in the style of Refoqus~\cite{moussa_resource_2023}), and a head-to-head comparison between the joint optimum $(V^\star, M^\star)$ and the classical repeated-sampling optimum $\ell^\star$ of~\cite{dexheimer_improving_2024} at matched $V \times M$ budgets, are both natural extensions.


\begin{thebibliography}{72}%
\makeatletter
\providecommand \@ifxundefined [1]{%
 \@ifx{#1\undefined}
}%
\providecommand \@ifnum [1]{%
 \ifnum #1\expandafter \@firstoftwo
 \else \expandafter \@secondoftwo
 \fi
}%
\providecommand \@ifx [1]{%
 \ifx #1\expandafter \@firstoftwo
 \else \expandafter \@secondoftwo
 \fi
}%
\providecommand \natexlab [1]{#1}%
\providecommand \enquote  [1]{``#1''}%
\providecommand \bibnamefont  [1]{#1}%
\providecommand \bibfnamefont [1]{#1}%
\providecommand \citenamefont [1]{#1}%
\providecommand \href@noop [0]{\@secondoftwo}%
\providecommand \href [0]{\begingroup \@sanitize@url \@href}%
\providecommand \@href[1]{\@@startlink{#1}\@@href}%
\providecommand \@@href[1]{\endgroup#1\@@endlink}%
\providecommand \@sanitize@url [0]{\catcode `\\12\catcode `\$12\catcode `\&12\catcode `\#12\catcode `\^12\catcode `\_12\catcode `\%12\relax}%
\providecommand \@@startlink[1]{}%
\providecommand \@@endlink[0]{}%
\providecommand \url  [0]{\begingroup\@sanitize@url \@url }%
\providecommand \@url [1]{\endgroup\@href {#1}{\urlprefix }}%
\providecommand \urlprefix  [0]{URL }%
\providecommand \Eprint [0]{\href }%
\providecommand \doibase [0]{https://doi.org/}%
\providecommand \selectlanguage [0]{\@gobble}%
\providecommand \bibinfo  [0]{\@secondoftwo}%
\providecommand \bibfield  [0]{\@secondoftwo}%
\providecommand \translation [1]{[#1]}%
\providecommand \BibitemOpen [0]{}%
\providecommand \bibitemStop [0]{}%
\providecommand \bibitemNoStop [0]{.\EOS\space}%
\providecommand \EOS [0]{\spacefactor3000\relax}%
\providecommand \BibitemShut  [1]{\csname bibitem#1\endcsname}%
\let\auto@bib@innerbib\@empty
\bibitem [{\citenamefont {Cerezo}\ \emph {et~al.}(2021{\natexlab{a}})\citenamefont {Cerezo}, \citenamefont {Arrasmith}, \citenamefont {Babbush}, \citenamefont {Benjamin}, \citenamefont {Endo}, \citenamefont {Fujii}, \citenamefont {McClean}, \citenamefont {Mitarai}, \citenamefont {Yuan}, \citenamefont {Cincio},\ and\ \citenamefont {Coles}}]{cerezo_variational_2021}%
  \BibitemOpen
  \bibfield  {author} {\bibinfo {author} {\bibfnamefont {M.}~\bibnamefont {Cerezo}}, \bibinfo {author} {\bibfnamefont {A.}~\bibnamefont {Arrasmith}}, \bibinfo {author} {\bibfnamefont {R.}~\bibnamefont {Babbush}}, \bibinfo {author} {\bibfnamefont {S.~C.}\ \bibnamefont {Benjamin}}, \bibinfo {author} {\bibfnamefont {S.}~\bibnamefont {Endo}}, \bibinfo {author} {\bibfnamefont {K.}~\bibnamefont {Fujii}}, \bibinfo {author} {\bibfnamefont {J.~R.}\ \bibnamefont {McClean}}, \bibinfo {author} {\bibfnamefont {K.}~\bibnamefont {Mitarai}}, \bibinfo {author} {\bibfnamefont {X.}~\bibnamefont {Yuan}}, \bibinfo {author} {\bibfnamefont {L.}~\bibnamefont {Cincio}},\ and\ \bibinfo {author} {\bibfnamefont {P.~J.}\ \bibnamefont {Coles}},\ }\bibfield  {title} {\bibinfo {title} {Variational quantum algorithms},\ }\href {https://doi.org/10.1038/s42254-021-00348-9} {\bibfield  {journal} {\bibinfo  {journal} {Nat Rev Phys}\ }\textbf {\bibinfo {volume} {3}},\ \bibinfo {pages} {625} (\bibinfo {year} {2021}{\natexlab{a}})}\BibitemShut
  {NoStop}%
\bibitem [{\citenamefont {Bharti}\ \emph {et~al.}(2022)\citenamefont {Bharti} \emph {et~al.}}]{bharti_noisy_2022}%
  \BibitemOpen
  \bibfield  {author} {\bibinfo {author} {\bibfnamefont {K.}~\bibnamefont {Bharti}} \emph {et~al.},\ }\bibfield  {title} {\bibinfo {title} {Noisy intermediate-scale quantum algorithms},\ }\href {https://doi.org/10.1103/RevModPhys.94.015004} {\bibfield  {journal} {\bibinfo  {journal} {Rev. Mod. Phys.}\ }\textbf {\bibinfo {volume} {94}},\ \bibinfo {pages} {015004} (\bibinfo {year} {2022})}\BibitemShut {NoStop}%
\bibitem [{\citenamefont {Larocca}\ \emph {et~al.}(2023)\citenamefont {Larocca}, \citenamefont {Ju}, \citenamefont {Garc\'{i}a-Mart\'{i}n}, \citenamefont {Coles},\ and\ \citenamefont {Cerezo}}]{larocca_theory_2023}%
  \BibitemOpen
  \bibfield  {author} {\bibinfo {author} {\bibfnamefont {M.}~\bibnamefont {Larocca}}, \bibinfo {author} {\bibfnamefont {N.}~\bibnamefont {Ju}}, \bibinfo {author} {\bibfnamefont {D.}~\bibnamefont {Garc\'{i}a-Mart\'{i}n}}, \bibinfo {author} {\bibfnamefont {P.~J.}\ \bibnamefont {Coles}},\ and\ \bibinfo {author} {\bibfnamefont {M.}~\bibnamefont {Cerezo}},\ }\bibfield  {title} {\bibinfo {title} {Theory of overparametrization in quantum neural networks},\ }\href {https://doi.org/10.1038/s43588-023-00467-6} {\bibfield  {journal} {\bibinfo  {journal} {Nat Comput Sci}\ }\textbf {\bibinfo {volume} {3}},\ \bibinfo {pages} {542} (\bibinfo {year} {2023})}\BibitemShut {NoStop}%
\bibitem [{\citenamefont {Delgado}\ \emph {et~al.}(2023)\citenamefont {Delgado}, \citenamefont {Rios},\ and\ \citenamefont {Hamilton}}]{delgado_identifying_2023}%
  \BibitemOpen
  \bibfield  {author} {\bibinfo {author} {\bibfnamefont {A.}~\bibnamefont {Delgado}}, \bibinfo {author} {\bibfnamefont {F.}~\bibnamefont {Rios}},\ and\ \bibinfo {author} {\bibfnamefont {K.~E.}\ \bibnamefont {Hamilton}},\ }\href@noop {} {\bibinfo {title} {Identifying overparameterization in {Quantum} {Circuit} {Born} {Machines}}} (\bibinfo {year} {2023}),\ \Eprint {https://arxiv.org/abs/2307.03292} {arXiv:2307.03292} \BibitemShut {NoStop}%
\bibitem [{\citenamefont {Garc\'{i}a-Mart\'{i}n}\ \emph {et~al.}(2024)\citenamefont {Garc\'{i}a-Mart\'{i}n}, \citenamefont {Larocca},\ and\ \citenamefont {Cerezo}}]{garcia-martin_effects_2024}%
  \BibitemOpen
  \bibfield  {author} {\bibinfo {author} {\bibfnamefont {D.}~\bibnamefont {Garc\'{i}a-Mart\'{i}n}}, \bibinfo {author} {\bibfnamefont {M.}~\bibnamefont {Larocca}},\ and\ \bibinfo {author} {\bibfnamefont {M.}~\bibnamefont {Cerezo}},\ }\bibfield  {title} {\bibinfo {title} {Effects of noise on the overparametrization of quantum neural networks},\ }\href {https://doi.org/10.1103/PhysRevResearch.6.013295} {\bibfield  {journal} {\bibinfo  {journal} {Phys. Rev. Res.}\ }\textbf {\bibinfo {volume} {6}},\ \bibinfo {pages} {013295} (\bibinfo {year} {2024})}\BibitemShut {NoStop}%
\bibitem [{\citenamefont {Holmes}\ \emph {et~al.}(2022)\citenamefont {Holmes}, \citenamefont {Sharma}, \citenamefont {Cerezo},\ and\ \citenamefont {Coles}}]{holmes_connecting_2022}%
  \BibitemOpen
  \bibfield  {author} {\bibinfo {author} {\bibfnamefont {Z.}~\bibnamefont {Holmes}}, \bibinfo {author} {\bibfnamefont {K.}~\bibnamefont {Sharma}}, \bibinfo {author} {\bibfnamefont {M.}~\bibnamefont {Cerezo}},\ and\ \bibinfo {author} {\bibfnamefont {P.~J.}\ \bibnamefont {Coles}},\ }\bibfield  {title} {\bibinfo {title} {Connecting ansatz expressibility to gradient magnitudes and barren plateaus},\ }\href {https://doi.org/10.1103/PRXQuantum.3.010313} {\bibfield  {journal} {\bibinfo  {journal} {PRX Quantum}\ }\textbf {\bibinfo {volume} {3}},\ \bibinfo {pages} {010313} (\bibinfo {year} {2022})}\BibitemShut {NoStop}%
\bibitem [{\citenamefont {Schuld}\ and\ \citenamefont {Killoran}(2022)}]{schuld_quantum_2022}%
  \BibitemOpen
  \bibfield  {author} {\bibinfo {author} {\bibfnamefont {M.}~\bibnamefont {Schuld}}\ and\ \bibinfo {author} {\bibfnamefont {N.}~\bibnamefont {Killoran}},\ }\bibfield  {title} {\bibinfo {title} {Is quantum advantage the right goal for quantum machine learning?},\ }\href {https://doi.org/10.1103/PRXQuantum.3.030101} {\bibfield  {journal} {\bibinfo  {journal} {PRX Quantum}\ }\textbf {\bibinfo {volume} {3}},\ \bibinfo {pages} {030101} (\bibinfo {year} {2022})}\BibitemShut {NoStop}%
\bibitem [{\citenamefont {Rumelhart}\ \emph {et~al.}(1986)\citenamefont {Rumelhart}, \citenamefont {Hinton},\ and\ \citenamefont {Williams}}]{rumelhart_learning_1986}%
  \BibitemOpen
  \bibfield  {author} {\bibinfo {author} {\bibfnamefont {D.~E.}\ \bibnamefont {Rumelhart}}, \bibinfo {author} {\bibfnamefont {G.~E.}\ \bibnamefont {Hinton}},\ and\ \bibinfo {author} {\bibfnamefont {R.~J.}\ \bibnamefont {Williams}},\ }\bibfield  {title} {\bibinfo {title} {Learning representations by back-propagating errors},\ }\href {https://doi.org/10.1038/323533a0} {\bibfield  {journal} {\bibinfo  {journal} {Nature}\ }\textbf {\bibinfo {volume} {323}},\ \bibinfo {pages} {533} (\bibinfo {year} {1986})}\BibitemShut {NoStop}%
\bibitem [{\citenamefont {Baydin}\ \emph {et~al.}(2018)\citenamefont {Baydin}, \citenamefont {Pearlmutter}, \citenamefont {Radul},\ and\ \citenamefont {Siskind}}]{baydin_automatic_2018}%
  \BibitemOpen
  \bibfield  {author} {\bibinfo {author} {\bibfnamefont {A.~G.}\ \bibnamefont {Baydin}}, \bibinfo {author} {\bibfnamefont {B.~A.}\ \bibnamefont {Pearlmutter}}, \bibinfo {author} {\bibfnamefont {A.~A.}\ \bibnamefont {Radul}},\ and\ \bibinfo {author} {\bibfnamefont {J.~M.}\ \bibnamefont {Siskind}},\ }\bibfield  {title} {\bibinfo {title} {Automatic {Differentiation} in {Machine} {Learning}: a {Survey}},\ }\href {https://jmlr.org/papers/v18/17-468.html} {\bibfield  {journal} {\bibinfo  {journal} {Journal of Machine Learning Research}\ }\textbf {\bibinfo {volume} {18}},\ \bibinfo {pages} {1} (\bibinfo {year} {2018})}\BibitemShut {NoStop}%
\bibitem [{\citenamefont {Abbas}\ \emph {et~al.}(2023)\citenamefont {Abbas}, \citenamefont {King}, \citenamefont {Huang}, \citenamefont {Huggins}, \citenamefont {Movassagh}, \citenamefont {Gilboa},\ and\ \citenamefont {McClean}}]{abbas_quantum_2023}%
  \BibitemOpen
  \bibfield  {author} {\bibinfo {author} {\bibfnamefont {A.}~\bibnamefont {Abbas}}, \bibinfo {author} {\bibfnamefont {R.}~\bibnamefont {King}}, \bibinfo {author} {\bibfnamefont {H.-Y.}\ \bibnamefont {Huang}}, \bibinfo {author} {\bibfnamefont {W.~J.}\ \bibnamefont {Huggins}}, \bibinfo {author} {\bibfnamefont {R.}~\bibnamefont {Movassagh}}, \bibinfo {author} {\bibfnamefont {D.}~\bibnamefont {Gilboa}},\ and\ \bibinfo {author} {\bibfnamefont {J.~R.}\ \bibnamefont {McClean}},\ }\bibfield  {title} {\bibinfo {title} {On quantum backpropagation, information reuse, and cheating measurement collapse},\ }in\ \href {https://proceedings.neurips.cc/paper_files/paper/2023/hash/8c3caae2f725c8e2a55ecd600563d172-Abstract-Conference.html} {\emph {\bibinfo {booktitle} {Advances in {Neural} {Information} {Processing} {Systems}}}},\ Vol.~\bibinfo {volume} {36}\ (\bibinfo {year} {2023})\ \Eprint {https://arxiv.org/abs/2305.13362} {arXiv:2305.13362} \BibitemShut {NoStop}%
\bibitem [{\citenamefont {Bowles}\ \emph {et~al.}(2025)\citenamefont {Bowles}, \citenamefont {Wierichs},\ and\ \citenamefont {Park}}]{bowles_backpropagation_2023}%
  \BibitemOpen
  \bibfield  {author} {\bibinfo {author} {\bibfnamefont {J.}~\bibnamefont {Bowles}}, \bibinfo {author} {\bibfnamefont {D.}~\bibnamefont {Wierichs}},\ and\ \bibinfo {author} {\bibfnamefont {C.-Y.}\ \bibnamefont {Park}},\ }\bibfield  {title} {\bibinfo {title} {Backpropagation scaling in parameterised quantum circuits},\ }\href {https://doi.org/10.22331/q-2025-10-02-1873} {\bibfield  {journal} {\bibinfo  {journal} {Quantum}\ }\textbf {\bibinfo {volume} {9}},\ \bibinfo {pages} {1873} (\bibinfo {year} {2025})}\BibitemShut {NoStop}%
\bibitem [{\citenamefont {Coyle}\ \emph {et~al.}(2025)\citenamefont {Coyle}, \citenamefont {Raj}, \citenamefont {Mathur}, \citenamefont {Cherrat}, \citenamefont {Jain}, \citenamefont {Kazdaghli},\ and\ \citenamefont {Kerenidis}}]{coyle_training-efficient_2024}%
  \BibitemOpen
  \bibfield  {author} {\bibinfo {author} {\bibfnamefont {B.}~\bibnamefont {Coyle}}, \bibinfo {author} {\bibfnamefont {S.}~\bibnamefont {Raj}}, \bibinfo {author} {\bibfnamefont {N.}~\bibnamefont {Mathur}}, \bibinfo {author} {\bibfnamefont {E.~A.}\ \bibnamefont {Cherrat}}, \bibinfo {author} {\bibfnamefont {N.}~\bibnamefont {Jain}}, \bibinfo {author} {\bibfnamefont {S.}~\bibnamefont {Kazdaghli}},\ and\ \bibinfo {author} {\bibfnamefont {I.}~\bibnamefont {Kerenidis}},\ }\bibfield  {title} {\bibinfo {title} {Training-efficient density quantum machine learning},\ }\href {https://doi.org/10.1038/s41534-025-01099-6} {\bibfield  {journal} {\bibinfo  {journal} {npj Quantum Inf.}\ }\textbf {\bibinfo {volume} {11}},\ \bibinfo {pages} {172} (\bibinfo {year} {2025})},\ \Eprint {https://arxiv.org/abs/2405.20237} {arXiv:2405.20237} \BibitemShut {NoStop}%
\bibitem [{\citenamefont {Chinzei}\ \emph {et~al.}(2025)\citenamefont {Chinzei}, \citenamefont {Yamano}, \citenamefont {Tran}, \citenamefont {Endo},\ and\ \citenamefont {Oshima}}]{chinzei_trade-off_2024}%
  \BibitemOpen
  \bibfield  {author} {\bibinfo {author} {\bibfnamefont {K.}~\bibnamefont {Chinzei}}, \bibinfo {author} {\bibfnamefont {S.}~\bibnamefont {Yamano}}, \bibinfo {author} {\bibfnamefont {Q.~H.}\ \bibnamefont {Tran}}, \bibinfo {author} {\bibfnamefont {Y.}~\bibnamefont {Endo}},\ and\ \bibinfo {author} {\bibfnamefont {H.}~\bibnamefont {Oshima}},\ }\bibfield  {title} {\bibinfo {title} {Trade-off between {Gradient} {Measurement} {Efficiency} and {Expressivity} in {Deep} {Quantum} {Neural} {Networks}},\ }\href {https://doi.org/10.1038/s41534-025-01036-7} {\bibfield  {journal} {\bibinfo  {journal} {npj Quantum Inf.}\ }\textbf {\bibinfo {volume} {11}},\ \bibinfo {pages} {79} (\bibinfo {year} {2025})}\BibitemShut {NoStop}%
\bibitem [{\citenamefont {Spall}(1992)}]{spall_multivariate_1992}%
  \BibitemOpen
  \bibfield  {author} {\bibinfo {author} {\bibfnamefont {J.}~\bibnamefont {Spall}},\ }\bibfield  {title} {\bibinfo {title} {Multivariate stochastic approximation using a simultaneous perturbation gradient approximation},\ }\href {https://doi.org/10.1109/9.119632} {\bibfield  {journal} {\bibinfo  {journal} {IEEE Transactions on Automatic Control}\ }\textbf {\bibinfo {volume} {37}},\ \bibinfo {pages} {332} (\bibinfo {year} {1992})}\BibitemShut {NoStop}%
\bibitem [{\citenamefont {Ding}\ \emph {et~al.}(2024)\citenamefont {Ding}, \citenamefont {Ko}, \citenamefont {Yao}, \citenamefont {Lin},\ and\ \citenamefont {Li}}]{ding_random_2024}%
  \BibitemOpen
  \bibfield  {author} {\bibinfo {author} {\bibfnamefont {Z.}~\bibnamefont {Ding}}, \bibinfo {author} {\bibfnamefont {T.}~\bibnamefont {Ko}}, \bibinfo {author} {\bibfnamefont {J.}~\bibnamefont {Yao}}, \bibinfo {author} {\bibfnamefont {L.}~\bibnamefont {Lin}},\ and\ \bibinfo {author} {\bibfnamefont {X.}~\bibnamefont {Li}},\ }\bibfield  {title} {\bibinfo {title} {Random coordinate descent: {A} simple alternative for optimizing parameterized quantum circuits},\ }\href {https://doi.org/10.1103/PhysRevResearch.6.033029} {\bibfield  {journal} {\bibinfo  {journal} {Phys. Rev. Res.}\ }\textbf {\bibinfo {volume} {6}},\ \bibinfo {pages} {033029} (\bibinfo {year} {2024})}\BibitemShut {NoStop}%
\bibitem [{\citenamefont {Baydin}\ \emph {et~al.}(2022)\citenamefont {Baydin}, \citenamefont {Pearlmutter}, \citenamefont {Syme}, \citenamefont {Wood},\ and\ \citenamefont {Torr}}]{baydin_gradients_2022}%
  \BibitemOpen
  \bibfield  {author} {\bibinfo {author} {\bibfnamefont {A.~G.}\ \bibnamefont {Baydin}}, \bibinfo {author} {\bibfnamefont {B.~A.}\ \bibnamefont {Pearlmutter}}, \bibinfo {author} {\bibfnamefont {D.}~\bibnamefont {Syme}}, \bibinfo {author} {\bibfnamefont {F.}~\bibnamefont {Wood}},\ and\ \bibinfo {author} {\bibfnamefont {P.}~\bibnamefont {Torr}},\ }\href@noop {} {\bibinfo {title} {Gradients without {Backpropagation}}} (\bibinfo {year} {2022}),\ \Eprint {https://arxiv.org/abs/2202.08587} {arXiv:2202.08587} \BibitemShut {NoStop}%
\bibitem [{\citenamefont {Silver}\ \emph {et~al.}(2022)\citenamefont {Silver}, \citenamefont {Goyal}, \citenamefont {Danihelka}, \citenamefont {Hessel},\ and\ \citenamefont {Hasselt}}]{silver_learning_2022}%
  \BibitemOpen
  \bibfield  {author} {\bibinfo {author} {\bibfnamefont {D.}~\bibnamefont {Silver}}, \bibinfo {author} {\bibfnamefont {A.}~\bibnamefont {Goyal}}, \bibinfo {author} {\bibfnamefont {I.}~\bibnamefont {Danihelka}}, \bibinfo {author} {\bibfnamefont {M.}~\bibnamefont {Hessel}},\ and\ \bibinfo {author} {\bibfnamefont {H.~v.}\ \bibnamefont {Hasselt}},\ }\bibfield  {title} {\bibinfo {title} {Learning by {Directional} {Gradient} {Descent}},\ }in\ \href {https://openreview.net/forum?id=5i7lJLuhTm} {\emph {\bibinfo {booktitle} {International {Conference} on {Learning} {Representations}}}}\ (\bibinfo {year} {2022})\BibitemShut {NoStop}%
\bibitem [{\citenamefont {Hanzely}\ \emph {et~al.}(2018)\citenamefont {Hanzely}, \citenamefont {Mishchenko},\ and\ \citenamefont {Richtarik}}]{hanzely_sega_2018}%
  \BibitemOpen
  \bibfield  {author} {\bibinfo {author} {\bibfnamefont {F.}~\bibnamefont {Hanzely}}, \bibinfo {author} {\bibfnamefont {K.}~\bibnamefont {Mishchenko}},\ and\ \bibinfo {author} {\bibfnamefont {P.}~\bibnamefont {Richtarik}},\ }\bibfield  {title} {\bibinfo {title} {{SEGA}: {Variance} {Reduction} via {Gradient} {Sketching}},\ }in\ \href {https://proceedings.neurips.cc/paper/2018/hash/fc2c7c47b918d0c2d792a719dfb602ef-Abstract.html} {\emph {\bibinfo {booktitle} {Advances in {Neural} {Information} {Processing} {Systems}}}},\ Vol.~\bibinfo {volume} {31}\ (\bibinfo {year} {2018})\ \Eprint {https://arxiv.org/abs/1809.03054} {arXiv:1809.03054} \BibitemShut {NoStop}%
\bibitem [{\citenamefont {Hinton}(2022)}]{hinton_forward-forward_2022}%
  \BibitemOpen
  \bibfield  {author} {\bibinfo {author} {\bibfnamefont {G.}~\bibnamefont {Hinton}},\ }\href@noop {} {\bibinfo {title} {The {Forward}-{Forward} {Algorithm}: {Some} {Preliminary} {Investigations}}} (\bibinfo {year} {2022}),\ \Eprint {https://arxiv.org/abs/2212.13345} {arXiv:2212.13345} \BibitemShut {NoStop}%
\bibitem [{\citenamefont {Fournier}\ \emph {et~al.}(2023)\citenamefont {Fournier}, \citenamefont {Rivaud}, \citenamefont {Belilovsky}, \citenamefont {Eickenberg},\ and\ \citenamefont {Oyallon}}]{fournier_can_2023}%
  \BibitemOpen
  \bibfield  {author} {\bibinfo {author} {\bibfnamefont {L.}~\bibnamefont {Fournier}}, \bibinfo {author} {\bibfnamefont {S.}~\bibnamefont {Rivaud}}, \bibinfo {author} {\bibfnamefont {E.}~\bibnamefont {Belilovsky}}, \bibinfo {author} {\bibfnamefont {M.}~\bibnamefont {Eickenberg}},\ and\ \bibinfo {author} {\bibfnamefont {E.}~\bibnamefont {Oyallon}},\ }\bibfield  {title} {\bibinfo {title} {Can {Forward} {Gradient} {Match} {Backpropagation}?},\ }in\ \href@noop {} {\emph {\bibinfo {booktitle} {Fortieth {International} {Conference} on {Machine} {Learning}}}}\ (\bibinfo {year} {2023})\ \Eprint {https://arxiv.org/abs/2306.06968} {arXiv:2306.06968} \BibitemShut {NoStop}%
\bibitem [{\citenamefont {Ren}\ \emph {et~al.}(2023)\citenamefont {Ren}, \citenamefont {Kornblith}, \citenamefont {Liao},\ and\ \citenamefont {Hinton}}]{ren_scaling_2023}%
  \BibitemOpen
  \bibfield  {author} {\bibinfo {author} {\bibfnamefont {M.}~\bibnamefont {Ren}}, \bibinfo {author} {\bibfnamefont {S.}~\bibnamefont {Kornblith}}, \bibinfo {author} {\bibfnamefont {R.}~\bibnamefont {Liao}},\ and\ \bibinfo {author} {\bibfnamefont {G.}~\bibnamefont {Hinton}},\ }\bibfield  {title} {\bibinfo {title} {Scaling {Forward} {Gradient} {With} {Local} {Losses}},\ }in\ \href@noop {} {\emph {\bibinfo {booktitle} {International {Conference} on {Learning} {Representations}}}}\ (\bibinfo {year} {2023})\ \Eprint {https://arxiv.org/abs/2210.03310} {arXiv:2210.03310} \BibitemShut {NoStop}%
\bibitem [{\citenamefont {Balles}\ \emph {et~al.}(2017)\citenamefont {Balles}, \citenamefont {Romero},\ and\ \citenamefont {Hennig}}]{balles_coupling_2017}%
  \BibitemOpen
  \bibfield  {author} {\bibinfo {author} {\bibfnamefont {L.}~\bibnamefont {Balles}}, \bibinfo {author} {\bibfnamefont {J.}~\bibnamefont {Romero}},\ and\ \bibinfo {author} {\bibfnamefont {P.}~\bibnamefont {Hennig}},\ }\bibfield  {title} {\bibinfo {title} {Coupling {Adaptive} {Batch} {Sizes} with {Learning} {Rates}},\ }in\ \href@noop {} {\emph {\bibinfo {booktitle} {Uncertainty in {Artificial} {Intelligence}}}}\ (\bibinfo {year} {2017})\ \Eprint {https://arxiv.org/abs/1612.05086} {arXiv:1612.05086} \BibitemShut {NoStop}%
\bibitem [{\citenamefont {K\"{u}bler}\ \emph {et~al.}(2020)\citenamefont {K\"{u}bler}, \citenamefont {Arrasmith}, \citenamefont {Cincio},\ and\ \citenamefont {Coles}}]{kubler_adaptive_2020}%
  \BibitemOpen
  \bibfield  {author} {\bibinfo {author} {\bibfnamefont {J.~M.}\ \bibnamefont {K\"{u}bler}}, \bibinfo {author} {\bibfnamefont {A.}~\bibnamefont {Arrasmith}}, \bibinfo {author} {\bibfnamefont {L.}~\bibnamefont {Cincio}},\ and\ \bibinfo {author} {\bibfnamefont {P.~J.}\ \bibnamefont {Coles}},\ }\bibfield  {title} {\bibinfo {title} {An {Adaptive} {Optimizer} for {Measurement}-{Frugal} {Variational} {Algorithms}},\ }\href {https://doi.org/10.22331/q-2020-05-11-263} {\bibfield  {journal} {\bibinfo  {journal} {Quantum}\ }\textbf {\bibinfo {volume} {4}},\ \bibinfo {pages} {263} (\bibinfo {year} {2020})}\BibitemShut {NoStop}%
\bibitem [{\citenamefont {Gu}\ \emph {et~al.}(2021)\citenamefont {Gu}, \citenamefont {Lowe}, \citenamefont {Dub}, \citenamefont {Coles},\ and\ \citenamefont {Arrasmith}}]{gu_adaptive_2021}%
  \BibitemOpen
  \bibfield  {author} {\bibinfo {author} {\bibfnamefont {A.}~\bibnamefont {Gu}}, \bibinfo {author} {\bibfnamefont {A.}~\bibnamefont {Lowe}}, \bibinfo {author} {\bibfnamefont {P.~A.}\ \bibnamefont {Dub}}, \bibinfo {author} {\bibfnamefont {P.~J.}\ \bibnamefont {Coles}},\ and\ \bibinfo {author} {\bibfnamefont {A.}~\bibnamefont {Arrasmith}},\ }\href@noop {} {\bibinfo {title} {Adaptive shot allocation for fast convergence in variational quantum algorithms}} (\bibinfo {year} {2021}),\ \Eprint {https://arxiv.org/abs/2108.10434} {arXiv:2108.10434} \BibitemShut {NoStop}%
\bibitem [{\citenamefont {Landman}\ \emph {et~al.}(2022)\citenamefont {Landman}, \citenamefont {Mathur}, \citenamefont {Li}, \citenamefont {Strahm}, \citenamefont {Kazdaghli}, \citenamefont {Prakash},\ and\ \citenamefont {Kerenidis}}]{landman_quantum_2022}%
  \BibitemOpen
  \bibfield  {author} {\bibinfo {author} {\bibfnamefont {J.}~\bibnamefont {Landman}}, \bibinfo {author} {\bibfnamefont {N.}~\bibnamefont {Mathur}}, \bibinfo {author} {\bibfnamefont {Y.~Y.}\ \bibnamefont {Li}}, \bibinfo {author} {\bibfnamefont {M.}~\bibnamefont {Strahm}}, \bibinfo {author} {\bibfnamefont {S.}~\bibnamefont {Kazdaghli}}, \bibinfo {author} {\bibfnamefont {A.}~\bibnamefont {Prakash}},\ and\ \bibinfo {author} {\bibfnamefont {I.}~\bibnamefont {Kerenidis}},\ }\bibfield  {title} {\bibinfo {title} {Quantum {Methods} for {Neural} {Networks} and {Application} to {Medical} {Image} {Classification}},\ }\href {https://doi.org/10.22331/q-2022-12-22-881} {\bibfield  {journal} {\bibinfo  {journal} {Quantum}\ }\textbf {\bibinfo {volume} {6}},\ \bibinfo {pages} {881} (\bibinfo {year} {2022})}\BibitemShut {NoStop}%
\bibitem [{\citenamefont {Monbroussou}\ \emph {et~al.}(2025)\citenamefont {Monbroussou}, \citenamefont {Landman}, \citenamefont {Grilo}, \citenamefont {Kukla},\ and\ \citenamefont {Kashefi}}]{monbroussou_trainability_2023}%
  \BibitemOpen
  \bibfield  {author} {\bibinfo {author} {\bibfnamefont {L.}~\bibnamefont {Monbroussou}}, \bibinfo {author} {\bibfnamefont {J.}~\bibnamefont {Landman}}, \bibinfo {author} {\bibfnamefont {A.~B.}\ \bibnamefont {Grilo}}, \bibinfo {author} {\bibfnamefont {R.}~\bibnamefont {Kukla}},\ and\ \bibinfo {author} {\bibfnamefont {E.}~\bibnamefont {Kashefi}},\ }\bibfield  {title} {\bibinfo {title} {Trainability and {Expressivity} of {Hamming}-{Weight} {Preserving} {Quantum} {Circuits} for {Machine} {Learning}},\ }\href {https://doi.org/10.22331/q-2025-05-15-1745} {\bibfield  {journal} {\bibinfo  {journal} {Quantum}\ }\textbf {\bibinfo {volume} {9}},\ \bibinfo {pages} {1745} (\bibinfo {year} {2025})}\BibitemShut {NoStop}%
\bibitem [{\citenamefont {Kingma}\ and\ \citenamefont {Ba}(2015)}]{kingma_adam_2017}%
  \BibitemOpen
  \bibfield  {author} {\bibinfo {author} {\bibfnamefont {D.~P.}\ \bibnamefont {Kingma}}\ and\ \bibinfo {author} {\bibfnamefont {J.}~\bibnamefont {Ba}},\ }\bibfield  {title} {\bibinfo {title} {Adam: {A} {Method} for {Stochastic} {Optimization}},\ }in\ \href@noop {} {\emph {\bibinfo {booktitle} {International {Conference} on {Learning} {Representations}}}}\ (\bibinfo {year} {2015})\ \Eprint {https://arxiv.org/abs/1412.6980} {arXiv:1412.6980} \BibitemShut {NoStop}%
\bibitem [{\citenamefont {Bradbury}\ \emph {et~al.}(2018)\citenamefont {Bradbury}, \citenamefont {Frostig}, \citenamefont {Hawkins}, \citenamefont {Johnson}, \citenamefont {Leary}, \citenamefont {Maclaurin}, \citenamefont {Necula}, \citenamefont {Paszke}, \citenamefont {VanderPlas}, \citenamefont {Wanderman-Milne},\ and\ \citenamefont {Zhang}}]{bradbury_jax_2018}%
  \BibitemOpen
  \bibfield  {author} {\bibinfo {author} {\bibfnamefont {J.}~\bibnamefont {Bradbury}}, \bibinfo {author} {\bibfnamefont {R.}~\bibnamefont {Frostig}}, \bibinfo {author} {\bibfnamefont {P.}~\bibnamefont {Hawkins}}, \bibinfo {author} {\bibfnamefont {M.~J.}\ \bibnamefont {Johnson}}, \bibinfo {author} {\bibfnamefont {C.}~\bibnamefont {Leary}}, \bibinfo {author} {\bibfnamefont {D.}~\bibnamefont {Maclaurin}}, \bibinfo {author} {\bibfnamefont {G.}~\bibnamefont {Necula}}, \bibinfo {author} {\bibfnamefont {A.}~\bibnamefont {Paszke}}, \bibinfo {author} {\bibfnamefont {J.}~\bibnamefont {VanderPlas}}, \bibinfo {author} {\bibfnamefont {S.}~\bibnamefont {Wanderman-Milne}},\ and\ \bibinfo {author} {\bibfnamefont {Q.}~\bibnamefont {Zhang}},\ }\href {http://github.com/google/jax} {\bibinfo {title} {{JAX}: composable transformations of {Python}+{NumPy} programs}} (\bibinfo {year} {2018})\BibitemShut {NoStop}%
\bibitem [{\citenamefont {Paszke}\ \emph {et~al.}(2019)\citenamefont {Paszke} \emph {et~al.}}]{paszke_pytorch_2019}%
  \BibitemOpen
  \bibfield  {author} {\bibinfo {author} {\bibfnamefont {A.}~\bibnamefont {Paszke}} \emph {et~al.},\ }\href@noop {} {\bibinfo {title} {{PyTorch}: {An} {Imperative} {Style}, {High}-{Performance} {Deep} {Learning} {Library}}} (\bibinfo {year} {2019}),\ \Eprint {https://arxiv.org/abs/1912.01703} {arXiv:1912.01703} \BibitemShut {NoStop}%
\bibitem [{\citenamefont {{Mart\'{i}n Abadi}}\ \emph {et~al.}(2015)\citenamefont {{Mart\'{i}n Abadi}} \emph {et~al.}}]{martin_abadi_tensorflow_2015}%
  \BibitemOpen
  \bibfield  {author} {\bibinfo {author} {\bibnamefont {{Mart\'{i}n Abadi}}} \emph {et~al.},\ }\href {https://www.tensorflow.org/} {\bibinfo {title} {{TensorFlow}: {Large}-{Scale} {Machine} {Learning} on {Heterogeneous} {Systems}}} (\bibinfo {year} {2015}),\ \bibinfo {note} {software available from tensorflow.org}\BibitemShut {NoStop}%
\bibitem [{\citenamefont {Griewank}\ \emph {et~al.}(2012)\citenamefont {Griewank}, \citenamefont {Kulshreshtha},\ and\ \citenamefont {Walther}}]{griewank_numerical_2012}%
  \BibitemOpen
  \bibfield  {author} {\bibinfo {author} {\bibfnamefont {A.}~\bibnamefont {Griewank}}, \bibinfo {author} {\bibfnamefont {K.}~\bibnamefont {Kulshreshtha}},\ and\ \bibinfo {author} {\bibfnamefont {A.}~\bibnamefont {Walther}},\ }\bibfield  {title} {\bibinfo {title} {On the numerical stability of algorithmic differentiation},\ }\href {https://doi.org/10.1007/s00607-011-0162-z} {\bibfield  {journal} {\bibinfo  {journal} {Computing}\ }\textbf {\bibinfo {volume} {94}},\ \bibinfo {pages} {125} (\bibinfo {year} {2012})}\BibitemShut {NoStop}%
\bibitem [{\citenamefont {Schmidhuber}(2015)}]{schmidhuber_deep_2015}%
  \BibitemOpen
  \bibfield  {author} {\bibinfo {author} {\bibfnamefont {J.}~\bibnamefont {Schmidhuber}},\ }\bibfield  {title} {\bibinfo {title} {Deep learning in neural networks: {An} overview},\ }\href {https://doi.org/10.1016/j.neunet.2014.09.003} {\bibfield  {journal} {\bibinfo  {journal} {Neural Networks}\ }\textbf {\bibinfo {volume} {61}},\ \bibinfo {pages} {85} (\bibinfo {year} {2015})}\BibitemShut {NoStop}%
\bibitem [{\citenamefont {P\'{e}rez-Salinas}\ \emph {et~al.}(2020)\citenamefont {P\'{e}rez-Salinas}, \citenamefont {Cervera-Lierta}, \citenamefont {Gil-Fuster},\ and\ \citenamefont {Latorre}}]{perez-salinas_data_2020}%
  \BibitemOpen
  \bibfield  {author} {\bibinfo {author} {\bibfnamefont {A.}~\bibnamefont {P\'{e}rez-Salinas}}, \bibinfo {author} {\bibfnamefont {A.}~\bibnamefont {Cervera-Lierta}}, \bibinfo {author} {\bibfnamefont {E.}~\bibnamefont {Gil-Fuster}},\ and\ \bibinfo {author} {\bibfnamefont {J.~I.}\ \bibnamefont {Latorre}},\ }\bibfield  {title} {\bibinfo {title} {Data re-uploading for a universal quantum classifier},\ }\href {https://doi.org/10.22331/q-2020-02-06-226} {\bibfield  {journal} {\bibinfo  {journal} {Quantum}\ }\textbf {\bibinfo {volume} {4}},\ \bibinfo {pages} {226} (\bibinfo {year} {2020})}\BibitemShut {NoStop}%
\bibitem [{\citenamefont {Romero}\ \emph {et~al.}(2018)\citenamefont {Romero}, \citenamefont {Babbush}, \citenamefont {McClean}, \citenamefont {Hempel}, \citenamefont {Love},\ and\ \citenamefont {Aspuru-Guzik}}]{romero_strategies_2018}%
  \BibitemOpen
  \bibfield  {author} {\bibinfo {author} {\bibfnamefont {J.}~\bibnamefont {Romero}}, \bibinfo {author} {\bibfnamefont {R.}~\bibnamefont {Babbush}}, \bibinfo {author} {\bibfnamefont {J.~R.}\ \bibnamefont {McClean}}, \bibinfo {author} {\bibfnamefont {C.}~\bibnamefont {Hempel}}, \bibinfo {author} {\bibfnamefont {P.~J.}\ \bibnamefont {Love}},\ and\ \bibinfo {author} {\bibfnamefont {A.}~\bibnamefont {Aspuru-Guzik}},\ }\bibfield  {title} {\bibinfo {title} {Strategies for quantum computing molecular energies using the unitary coupled cluster ansatz},\ }\href {https://doi.org/10.1088/2058-9565/aad3e4} {\bibfield  {journal} {\bibinfo  {journal} {Quantum Sci. Technol.}\ }\textbf {\bibinfo {volume} {4}},\ \bibinfo {pages} {014008} (\bibinfo {year} {2018})}\BibitemShut {NoStop}%
\bibitem [{\citenamefont {Farhi}\ and\ \citenamefont {Neven}(2018)}]{farhi_classification_2018}%
  \BibitemOpen
  \bibfield  {author} {\bibinfo {author} {\bibfnamefont {E.}~\bibnamefont {Farhi}}\ and\ \bibinfo {author} {\bibfnamefont {H.}~\bibnamefont {Neven}},\ }\href@noop {} {\bibinfo {title} {Classification with {Quantum} {Neural} {Networks} on {Near} {Term} {Processors}}} (\bibinfo {year} {2018}),\ \Eprint {https://arxiv.org/abs/1802.06002} {arXiv:1802.06002} \BibitemShut {NoStop}%
\bibitem [{\citenamefont {Mitarai}\ \emph {et~al.}(2018)\citenamefont {Mitarai}, \citenamefont {Negoro}, \citenamefont {Kitagawa},\ and\ \citenamefont {Fujii}}]{mitarai_quantum_2018}%
  \BibitemOpen
  \bibfield  {author} {\bibinfo {author} {\bibfnamefont {K.}~\bibnamefont {Mitarai}}, \bibinfo {author} {\bibfnamefont {M.}~\bibnamefont {Negoro}}, \bibinfo {author} {\bibfnamefont {M.}~\bibnamefont {Kitagawa}},\ and\ \bibinfo {author} {\bibfnamefont {K.}~\bibnamefont {Fujii}},\ }\bibfield  {title} {\bibinfo {title} {Quantum circuit learning},\ }\href {https://doi.org/10.1103/PhysRevA.98.032309} {\bibfield  {journal} {\bibinfo  {journal} {Phys. Rev. A}\ }\textbf {\bibinfo {volume} {98}},\ \bibinfo {pages} {032309} (\bibinfo {year} {2018})}\BibitemShut {NoStop}%
\bibitem [{\citenamefont {Wierichs}\ \emph {et~al.}(2022)\citenamefont {Wierichs}, \citenamefont {Izaac}, \citenamefont {Wang},\ and\ \citenamefont {Lin}}]{wierichs_general_2022}%
  \BibitemOpen
  \bibfield  {author} {\bibinfo {author} {\bibfnamefont {D.}~\bibnamefont {Wierichs}}, \bibinfo {author} {\bibfnamefont {J.}~\bibnamefont {Izaac}}, \bibinfo {author} {\bibfnamefont {C.}~\bibnamefont {Wang}},\ and\ \bibinfo {author} {\bibfnamefont {C.~Y.-Y.}\ \bibnamefont {Lin}},\ }\bibfield  {title} {\bibinfo {title} {General parameter-shift rules for quantum gradients},\ }\href {https://doi.org/10.22331/q-2022-03-30-677} {\bibfield  {journal} {\bibinfo  {journal} {Quantum}\ }\textbf {\bibinfo {volume} {6}},\ \bibinfo {pages} {677} (\bibinfo {year} {2022})}\BibitemShut {NoStop}%
\bibitem [{\citenamefont {Kyriienko}\ and\ \citenamefont {Elfving}(2021)}]{kyriienko_generalized_2021}%
  \BibitemOpen
  \bibfield  {author} {\bibinfo {author} {\bibfnamefont {O.}~\bibnamefont {Kyriienko}}\ and\ \bibinfo {author} {\bibfnamefont {V.~E.}\ \bibnamefont {Elfving}},\ }\bibfield  {title} {\bibinfo {title} {Generalized quantum circuit differentiation rules},\ }\href {https://doi.org/10.1103/PhysRevA.104.052417} {\bibfield  {journal} {\bibinfo  {journal} {Phys. Rev. A}\ }\textbf {\bibinfo {volume} {104}},\ \bibinfo {pages} {052417} (\bibinfo {year} {2021})}\BibitemShut {NoStop}%
\bibitem [{\citenamefont {Anselmetti}\ \emph {et~al.}(2021)\citenamefont {Anselmetti}, \citenamefont {Wierichs}, \citenamefont {Gogolin},\ and\ \citenamefont {Parrish}}]{anselmetti_local_2021}%
  \BibitemOpen
  \bibfield  {author} {\bibinfo {author} {\bibfnamefont {G.-L.~R.}\ \bibnamefont {Anselmetti}}, \bibinfo {author} {\bibfnamefont {D.}~\bibnamefont {Wierichs}}, \bibinfo {author} {\bibfnamefont {C.}~\bibnamefont {Gogolin}},\ and\ \bibinfo {author} {\bibfnamefont {R.~M.}\ \bibnamefont {Parrish}},\ }\bibfield  {title} {\bibinfo {title} {Local, expressive, quantum-number-preserving {VQE} ans\"{a}tze for fermionic systems},\ }\href {https://doi.org/10.1088/1367-2630/ac2cb3} {\bibfield  {journal} {\bibinfo  {journal} {New J. Phys.}\ }\textbf {\bibinfo {volume} {23}},\ \bibinfo {pages} {113010} (\bibinfo {year} {2021})}\BibitemShut {NoStop}%
\bibitem [{\citenamefont {Sweke}\ \emph {et~al.}(2020)\citenamefont {Sweke}, \citenamefont {Wilde}, \citenamefont {Meyer}, \citenamefont {Schuld}, \citenamefont {Faehrmann}, \citenamefont {Meynard-Piganeau},\ and\ \citenamefont {Eisert}}]{sweke_stochastic_2020}%
  \BibitemOpen
  \bibfield  {author} {\bibinfo {author} {\bibfnamefont {R.}~\bibnamefont {Sweke}}, \bibinfo {author} {\bibfnamefont {F.}~\bibnamefont {Wilde}}, \bibinfo {author} {\bibfnamefont {J.}~\bibnamefont {Meyer}}, \bibinfo {author} {\bibfnamefont {M.}~\bibnamefont {Schuld}}, \bibinfo {author} {\bibfnamefont {P.~K.}\ \bibnamefont {Faehrmann}}, \bibinfo {author} {\bibfnamefont {B.}~\bibnamefont {Meynard-Piganeau}},\ and\ \bibinfo {author} {\bibfnamefont {J.}~\bibnamefont {Eisert}},\ }\bibfield  {title} {\bibinfo {title} {Stochastic gradient descent for hybrid quantum-classical optimization},\ }\href {https://doi.org/10.22331/q-2020-08-31-314} {\bibfield  {journal} {\bibinfo  {journal} {Quantum}\ }\textbf {\bibinfo {volume} {4}},\ \bibinfo {pages} {314} (\bibinfo {year} {2020})}\BibitemShut {NoStop}%
\bibitem [{\citenamefont {Moussa}\ \emph {et~al.}(2023)\citenamefont {Moussa}, \citenamefont {Gordon}, \citenamefont {Baczyk}, \citenamefont {Cerezo}, \citenamefont {Cincio},\ and\ \citenamefont {Coles}}]{moussa_resource_2023}%
  \BibitemOpen
  \bibfield  {author} {\bibinfo {author} {\bibfnamefont {C.}~\bibnamefont {Moussa}}, \bibinfo {author} {\bibfnamefont {M.~H.}\ \bibnamefont {Gordon}}, \bibinfo {author} {\bibfnamefont {M.}~\bibnamefont {Baczyk}}, \bibinfo {author} {\bibfnamefont {M.}~\bibnamefont {Cerezo}}, \bibinfo {author} {\bibfnamefont {L.}~\bibnamefont {Cincio}},\ and\ \bibinfo {author} {\bibfnamefont {P.~J.}\ \bibnamefont {Coles}},\ }\bibfield  {title} {\bibinfo {title} {Resource frugal optimizer for quantum machine learning},\ }\href {https://doi.org/10.1088/2058-9565/acef55} {\bibfield  {journal} {\bibinfo  {journal} {Quantum Sci. Technol.}\ }\textbf {\bibinfo {volume} {8}},\ \bibinfo {pages} {045019} (\bibinfo {year} {2023})}\BibitemShut {NoStop}%
\bibitem [{\citenamefont {Spall}(1987)}]{spall_stochastic_1987}%
  \BibitemOpen
  \bibfield  {author} {\bibinfo {author} {\bibfnamefont {J.~C.}\ \bibnamefont {Spall}},\ }\bibfield  {title} {\bibinfo {title} {A {Stochastic} {Approximation} {Technique} for {Generating} {Maximum} {Likelihood} {Parameter} {Estimates}},\ }in\ \href {https://ieeexplore.ieee.org/document/4789489/} {\emph {\bibinfo {booktitle} {1987 {American} {Control} {Conference}}}}\ (\bibinfo {year} {1987})\ pp.\ \bibinfo {pages} {1161--1167}\BibitemShut {NoStop}%
\bibitem [{\citenamefont {Bhatnagar}\ \emph {et~al.}(2013)\citenamefont {Bhatnagar}, \citenamefont {Prasad},\ and\ \citenamefont {Prashanth}}]{bhatnagar_stochastic_2013}%
  \BibitemOpen
  \bibfield  {author} {\bibinfo {author} {\bibfnamefont {S.}~\bibnamefont {Bhatnagar}}, \bibinfo {author} {\bibfnamefont {H.}~\bibnamefont {Prasad}},\ and\ \bibinfo {author} {\bibfnamefont {L.}~\bibnamefont {Prashanth}},\ }\bibfield  {title} {\bibinfo {title} {Stochastic {Approximation} {Algorithms}},\ }in\ \href {https://doi.org/10.1007/978-1-4471-4285-0_3} {\emph {\bibinfo {booktitle} {Stochastic {Recursive} {Algorithms} for {Optimization}}}}\ (\bibinfo  {publisher} {Springer},\ \bibinfo {year} {2013})\ pp.\ \bibinfo {pages} {17--28}\BibitemShut {NoStop}%
\bibitem [{\citenamefont {Cade}\ \emph {et~al.}(2020)\citenamefont {Cade}, \citenamefont {Mineh}, \citenamefont {Montanaro},\ and\ \citenamefont {Stanisic}}]{cade_strategies_2020}%
  \BibitemOpen
  \bibfield  {author} {\bibinfo {author} {\bibfnamefont {C.}~\bibnamefont {Cade}}, \bibinfo {author} {\bibfnamefont {L.}~\bibnamefont {Mineh}}, \bibinfo {author} {\bibfnamefont {A.}~\bibnamefont {Montanaro}},\ and\ \bibinfo {author} {\bibfnamefont {S.}~\bibnamefont {Stanisic}},\ }\bibfield  {title} {\bibinfo {title} {Strategies for solving the {Fermi}-{Hubbard} model on near-term quantum computers},\ }\href {https://doi.org/10.1103/PhysRevB.102.235122} {\bibfield  {journal} {\bibinfo  {journal} {Phys. Rev. B}\ }\textbf {\bibinfo {volume} {102}},\ \bibinfo {pages} {235122} (\bibinfo {year} {2020})}\BibitemShut {NoStop}%
\bibitem [{\citenamefont {Gacon}\ \emph {et~al.}(2021)\citenamefont {Gacon}, \citenamefont {Zoufal}, \citenamefont {Carleo},\ and\ \citenamefont {Woerner}}]{gacon_simultaneous_2021}%
  \BibitemOpen
  \bibfield  {author} {\bibinfo {author} {\bibfnamefont {J.}~\bibnamefont {Gacon}}, \bibinfo {author} {\bibfnamefont {C.}~\bibnamefont {Zoufal}}, \bibinfo {author} {\bibfnamefont {G.}~\bibnamefont {Carleo}},\ and\ \bibinfo {author} {\bibfnamefont {S.}~\bibnamefont {Woerner}},\ }\bibfield  {title} {\bibinfo {title} {Simultaneous {Perturbation} {Stochastic} {Approximation} of the {Quantum} {Fisher} {Information}},\ }\href {https://doi.org/10.22331/q-2021-10-20-567} {\bibfield  {journal} {\bibinfo  {journal} {Quantum}\ }\textbf {\bibinfo {volume} {5}},\ \bibinfo {pages} {567} (\bibinfo {year} {2021})}\BibitemShut {NoStop}%
\bibitem [{\citenamefont {Jain}\ \emph {et~al.}(2022)\citenamefont {Jain}, \citenamefont {Coyle}, \citenamefont {Kashefi},\ and\ \citenamefont {Kumar}}]{jain_graph_2022}%
  \BibitemOpen
  \bibfield  {author} {\bibinfo {author} {\bibfnamefont {N.}~\bibnamefont {Jain}}, \bibinfo {author} {\bibfnamefont {B.}~\bibnamefont {Coyle}}, \bibinfo {author} {\bibfnamefont {E.}~\bibnamefont {Kashefi}},\ and\ \bibinfo {author} {\bibfnamefont {N.}~\bibnamefont {Kumar}},\ }\bibfield  {title} {\bibinfo {title} {Graph neural network initialisation of quantum approximate optimisation},\ }\href {https://doi.org/10.22331/q-2022-11-17-861} {\bibfield  {journal} {\bibinfo  {journal} {Quantum}\ }\textbf {\bibinfo {volume} {6}},\ \bibinfo {pages} {861} (\bibinfo {year} {2022})}\BibitemShut {NoStop}%
\bibitem [{\citenamefont {Sauvage}\ and\ \citenamefont {Mintert}(2020)}]{sauvage_optimal_2020}%
  \BibitemOpen
  \bibfield  {author} {\bibinfo {author} {\bibfnamefont {F.}~\bibnamefont {Sauvage}}\ and\ \bibinfo {author} {\bibfnamefont {F.}~\bibnamefont {Mintert}},\ }\bibfield  {title} {\bibinfo {title} {Optimal quantum control with poor statistics},\ }\href {https://doi.org/10.1103/PRXQuantum.1.020322} {\bibfield  {journal} {\bibinfo  {journal} {PRX Quantum}\ }\textbf {\bibinfo {volume} {1}},\ \bibinfo {pages} {020322} (\bibinfo {year} {2020})}\BibitemShut {NoStop}%
\bibitem [{\citenamefont {Bonet-Monroig}\ \emph {et~al.}(2023)\citenamefont {Bonet-Monroig}, \citenamefont {Wang}, \citenamefont {Vermetten}, \citenamefont {Senjean}, \citenamefont {Moussa}, \citenamefont {B{\"a}ck}, \citenamefont {Dunjko},\ and\ \citenamefont {O'Brien}}]{bonet-monroig_performance_2023}%
  \BibitemOpen
  \bibfield  {author} {\bibinfo {author} {\bibfnamefont {X.}~\bibnamefont {Bonet-Monroig}}, \bibinfo {author} {\bibfnamefont {H.}~\bibnamefont {Wang}}, \bibinfo {author} {\bibfnamefont {D.}~\bibnamefont {Vermetten}}, \bibinfo {author} {\bibfnamefont {B.}~\bibnamefont {Senjean}}, \bibinfo {author} {\bibfnamefont {C.}~\bibnamefont {Moussa}}, \bibinfo {author} {\bibfnamefont {T.}~\bibnamefont {B{\"a}ck}}, \bibinfo {author} {\bibfnamefont {V.}~\bibnamefont {Dunjko}},\ and\ \bibinfo {author} {\bibfnamefont {T.~E.}\ \bibnamefont {O'Brien}},\ }\bibfield  {title} {\bibinfo {title} {Performance comparison of optimization methods on variational quantum algorithms},\ }\href {https://doi.org/10.1103/PhysRevA.107.032407} {\bibfield  {journal} {\bibinfo  {journal} {Physical Review A}\ }\textbf {\bibinfo {volume} {107}},\ \bibinfo {pages} {032407} (\bibinfo {year} {2023})},\ \Eprint {https://arxiv.org/abs/2111.13454} {arXiv:2111.13454 [quant-ph]} \BibitemShut {NoStop}%
\bibitem [{\citenamefont {Nesterov}(2012)}]{nesterov_efficiency_2012}%
  \BibitemOpen
  \bibfield  {author} {\bibinfo {author} {\bibfnamefont {Y.}~\bibnamefont {Nesterov}},\ }\bibfield  {title} {\bibinfo {title} {Efficiency of {Coordinate} {Descent} {Methods} on {Huge}-{Scale} {Optimization} {Problems}},\ }\href {https://doi.org/10.1137/100802001} {\bibfield  {journal} {\bibinfo  {journal} {SIAM J. Optim.}\ }\textbf {\bibinfo {volume} {22}},\ \bibinfo {pages} {341} (\bibinfo {year} {2012})}\BibitemShut {NoStop}%
\bibitem [{\citenamefont {Richt\'{a}rik}\ and\ \citenamefont {Tak\'{a}\v{c}}(2014)}]{richtarik_iteration_2014}%
  \BibitemOpen
  \bibfield  {author} {\bibinfo {author} {\bibfnamefont {P.}~\bibnamefont {Richt\'{a}rik}}\ and\ \bibinfo {author} {\bibfnamefont {M.}~\bibnamefont {Tak\'{a}\v{c}}},\ }\bibfield  {title} {\bibinfo {title} {Iteration complexity of randomized block-coordinate descent methods for minimizing a composite function},\ }\href {https://doi.org/10.1007/s10107-012-0614-z} {\bibfield  {journal} {\bibinfo  {journal} {Math. Program.}\ }\textbf {\bibinfo {volume} {144}},\ \bibinfo {pages} {1} (\bibinfo {year} {2014})}\BibitemShut {NoStop}%
\bibitem [{\citenamefont {Arrasmith}\ \emph {et~al.}(2020)\citenamefont {Arrasmith}, \citenamefont {Cincio}, \citenamefont {Somma},\ and\ \citenamefont {Coles}}]{arrasmith_operator_2020}%
  \BibitemOpen
  \bibfield  {author} {\bibinfo {author} {\bibfnamefont {A.}~\bibnamefont {Arrasmith}}, \bibinfo {author} {\bibfnamefont {L.}~\bibnamefont {Cincio}}, \bibinfo {author} {\bibfnamefont {R.~D.}\ \bibnamefont {Somma}},\ and\ \bibinfo {author} {\bibfnamefont {P.~J.}\ \bibnamefont {Coles}},\ }\href@noop {} {\bibinfo {title} {Operator {Sampling} for {Shot}-frugal {Optimization} in {Variational} {Algorithms}}} (\bibinfo {year} {2020}),\ \Eprint {https://arxiv.org/abs/2004.06252} {arXiv:2004.06252} \BibitemShut {NoStop}%
\bibitem [{\citenamefont {van Straaten}\ and\ \citenamefont {Koczor}(2021)}]{vanstraaten_measurement_2021}%
  \BibitemOpen
  \bibfield  {author} {\bibinfo {author} {\bibfnamefont {B.}~\bibnamefont {van Straaten}}\ and\ \bibinfo {author} {\bibfnamefont {B.}~\bibnamefont {Koczor}},\ }\bibfield  {title} {\bibinfo {title} {Measurement cost of metric-aware variational quantum algorithms},\ }\href {https://doi.org/10.1103/PRXQuantum.2.030324} {\bibfield  {journal} {\bibinfo  {journal} {PRX Quantum}\ }\textbf {\bibinfo {volume} {2}},\ \bibinfo {pages} {030324} (\bibinfo {year} {2021})}\BibitemShut {NoStop}%
\bibitem [{\citenamefont {Boyd}\ and\ \citenamefont {Koczor}(2022)}]{boyd_training_2022}%
  \BibitemOpen
  \bibfield  {author} {\bibinfo {author} {\bibfnamefont {G.}~\bibnamefont {Boyd}}\ and\ \bibinfo {author} {\bibfnamefont {B.}~\bibnamefont {Koczor}},\ }\bibfield  {title} {\bibinfo {title} {Training variational quantum circuits with {CoVaR}: Covariance root finding with classical shadows},\ }\href {https://doi.org/10.1103/PhysRevX.12.041022} {\bibfield  {journal} {\bibinfo  {journal} {Phys. Rev. X}\ }\textbf {\bibinfo {volume} {12}},\ \bibinfo {pages} {041022} (\bibinfo {year} {2022})}\BibitemShut {NoStop}%
\bibitem [{\citenamefont {Garc\'{i}a-P\'{e}rez}\ \emph {et~al.}(2021)\citenamefont {Garc\'{i}a-P\'{e}rez}, \citenamefont {Rossi}, \citenamefont {Sokolov}, \citenamefont {Tacchino}, \citenamefont {Barkoutsos}, \citenamefont {Mazzola}, \citenamefont {Tavernelli},\ and\ \citenamefont {Maniscalco}}]{garcia-perez_learning_2021}%
  \BibitemOpen
  \bibfield  {author} {\bibinfo {author} {\bibfnamefont {G.}~\bibnamefont {Garc\'{i}a-P\'{e}rez}}, \bibinfo {author} {\bibfnamefont {M.~A.~C.}\ \bibnamefont {Rossi}}, \bibinfo {author} {\bibfnamefont {B.}~\bibnamefont {Sokolov}}, \bibinfo {author} {\bibfnamefont {F.}~\bibnamefont {Tacchino}}, \bibinfo {author} {\bibfnamefont {P.~K.}\ \bibnamefont {Barkoutsos}}, \bibinfo {author} {\bibfnamefont {G.}~\bibnamefont {Mazzola}}, \bibinfo {author} {\bibfnamefont {I.}~\bibnamefont {Tavernelli}},\ and\ \bibinfo {author} {\bibfnamefont {S.}~\bibnamefont {Maniscalco}},\ }\bibfield  {title} {\bibinfo {title} {Learning to measure: Adaptive informationally complete generalized measurements for quantum algorithms},\ }\href {https://doi.org/10.1103/PRXQuantum.2.040342} {\bibfield  {journal} {\bibinfo  {journal} {PRX Quantum}\ }\textbf {\bibinfo {volume} {2}},\ \bibinfo {pages} {040342} (\bibinfo {year} {2021})}\BibitemShut {NoStop}%
\bibitem [{\citenamefont {Pramanik}\ and\ \citenamefont {Chandra}(2025)}]{pramanik_stochastic_2025}%
  \BibitemOpen
  \bibfield  {author} {\bibinfo {author} {\bibfnamefont {S.}~\bibnamefont {Pramanik}}\ and\ \bibinfo {author} {\bibfnamefont {M.~G.}\ \bibnamefont {Chandra}},\ }\href@noop {} {\bibinfo {title} {Stochastic {S}hadow {D}escent: {T}raining {P}arametrized {Q}uantum {C}ircuits with {S}hadows of {G}radients}} (\bibinfo {year} {2025}),\ \Eprint {https://arxiv.org/abs/2511.12168} {arXiv:2511.12168} \BibitemShut {NoStop}%
\bibitem [{\citenamefont {Fl\"{u}gel}\ \emph {et~al.}(2024)\citenamefont {Fl\"{u}gel}, \citenamefont {Coquelin}, \citenamefont {G\"{o}tz},\ and\ \citenamefont {Debus}}]{flugel_beyond_2024}%
  \BibitemOpen
  \bibfield  {author} {\bibinfo {author} {\bibfnamefont {K.}~\bibnamefont {Fl\"{u}gel}}, \bibinfo {author} {\bibfnamefont {D.}~\bibnamefont {Coquelin}}, \bibinfo {author} {\bibfnamefont {M.}~\bibnamefont {G\"{o}tz}},\ and\ \bibinfo {author} {\bibfnamefont {C.}~\bibnamefont {Debus}},\ }\href@noop {} {\bibinfo {title} {Beyond {Backpropagation}: {Optimization} with {Multi}-{Tangent} {Forward} {Gradients}}} (\bibinfo {year} {2024}),\ \Eprint {https://arxiv.org/abs/2410.17764} {arXiv:2410.17764} \BibitemShut {NoStop}%
\bibitem [{\citenamefont {Bos}\ and\ \citenamefont {Schmidt-Hieber}(2024)}]{bos_convergence_2024}%
  \BibitemOpen
  \bibfield  {author} {\bibinfo {author} {\bibfnamefont {T.}~\bibnamefont {Bos}}\ and\ \bibinfo {author} {\bibfnamefont {J.}~\bibnamefont {Schmidt-Hieber}},\ }\bibfield  {title} {\bibinfo {title} {Convergence guarantees for forward gradient descent in the linear regression model},\ }\href {https://doi.org/10.1016/j.jspi.2024.106174} {\bibfield  {journal} {\bibinfo  {journal} {Journal of Statistical Planning and Inference}\ }\textbf {\bibinfo {volume} {233}},\ \bibinfo {pages} {106174} (\bibinfo {year} {2024})}\BibitemShut {NoStop}%
\bibitem [{\citenamefont {Dexheimer}\ and\ \citenamefont {Schmidt-Hieber}(2024)}]{dexheimer_improving_2024}%
  \BibitemOpen
  \bibfield  {author} {\bibinfo {author} {\bibfnamefont {N.}~\bibnamefont {Dexheimer}}\ and\ \bibinfo {author} {\bibfnamefont {J.}~\bibnamefont {Schmidt-Hieber}},\ }\href@noop {} {\bibinfo {title} {Improving the {Convergence} {Rates} of {Forward} {Gradient} {Descent} with {Repeated} {Sampling}}} (\bibinfo {year} {2024}),\ \Eprint {https://arxiv.org/abs/2411.17567} {arXiv:2411.17567} \BibitemShut {NoStop}%
\bibitem [{\citenamefont {Singhal}\ \emph {et~al.}(2023)\citenamefont {Singhal}, \citenamefont {Cheung}, \citenamefont {Chandra}, \citenamefont {Ragan-Kelley}, \citenamefont {Tenenbaum}, \citenamefont {Poggio},\ and\ \citenamefont {Yu}}]{singhal_how_2023}%
  \BibitemOpen
  \bibfield  {author} {\bibinfo {author} {\bibfnamefont {U.}~\bibnamefont {Singhal}}, \bibinfo {author} {\bibfnamefont {B.}~\bibnamefont {Cheung}}, \bibinfo {author} {\bibfnamefont {K.}~\bibnamefont {Chandra}}, \bibinfo {author} {\bibfnamefont {J.}~\bibnamefont {Ragan-Kelley}}, \bibinfo {author} {\bibfnamefont {J.~B.}\ \bibnamefont {Tenenbaum}}, \bibinfo {author} {\bibfnamefont {T.~A.}\ \bibnamefont {Poggio}},\ and\ \bibinfo {author} {\bibfnamefont {S.~X.}\ \bibnamefont {Yu}},\ }\href@noop {} {\bibinfo {title} {How to guess a gradient}} (\bibinfo {year} {2023}),\ \Eprint {https://arxiv.org/abs/2312.04709} {arXiv:2312.04709} \BibitemShut {NoStop}%
\bibitem [{\citenamefont {Wang}\ \emph {et~al.}(2025)\citenamefont {Wang}, \citenamefont {Markou},\ and\ \citenamefont {Campbell}}]{wang_towards_2025}%
  \BibitemOpen
  \bibfield  {author} {\bibinfo {author} {\bibfnamefont {Z.}~\bibnamefont {Wang}}, \bibinfo {author} {\bibfnamefont {S.}~\bibnamefont {Markou}},\ and\ \bibinfo {author} {\bibfnamefont {A.}~\bibnamefont {Campbell}},\ }\href@noop {} {\bibinfo {title} {Towards {Scalable} {Backpropagation}-{Free} {Gradient} {Estimation}}} (\bibinfo {year} {2025}),\ \Eprint {https://arxiv.org/abs/2511.03110} {arXiv:2511.03110} \BibitemShut {NoStop}%
\bibitem [{\citenamefont {Panchal}\ \emph {et~al.}(2025)\citenamefont {Panchal}, \citenamefont {Choudhary}, \citenamefont {Brun},\ and\ \citenamefont {Guan}}]{panchal_cost_2025}%
  \BibitemOpen
  \bibfield  {author} {\bibinfo {author} {\bibfnamefont {K.}~\bibnamefont {Panchal}}, \bibinfo {author} {\bibfnamefont {S.}~\bibnamefont {Choudhary}}, \bibinfo {author} {\bibfnamefont {Y.}~\bibnamefont {Brun}},\ and\ \bibinfo {author} {\bibfnamefont {H.}~\bibnamefont {Guan}},\ }\href@noop {} {\bibinfo {title} {The {Cost} of {Avoiding} {Backpropagation}}} (\bibinfo {year} {2025}),\ \Eprint {https://arxiv.org/abs/2506.21833} {arXiv:2506.21833} \BibitemShut {NoStop}%
\bibitem [{\citenamefont {Cobb}\ \emph {et~al.}(2025)\citenamefont {Cobb}, \citenamefont {Baydin}, \citenamefont {Pearlmutter},\ and\ \citenamefont {Jha}}]{cobb_second_2025}%
  \BibitemOpen
  \bibfield  {author} {\bibinfo {author} {\bibfnamefont {A.~D.}\ \bibnamefont {Cobb}}, \bibinfo {author} {\bibfnamefont {A.~G.}\ \bibnamefont {Baydin}}, \bibinfo {author} {\bibfnamefont {B.~A.}\ \bibnamefont {Pearlmutter}},\ and\ \bibinfo {author} {\bibfnamefont {S.}~\bibnamefont {Jha}},\ }\bibfield  {title} {\bibinfo {title} {Second-{Order} {Forward}-{Mode} {Automatic} {Differentiation} for {Optimization}},\ }in\ \href@noop {} {\emph {\bibinfo {booktitle} {International {Conference} on {Learning} {Representations}}}}\ (\bibinfo {year} {2025})\ \Eprint {https://arxiv.org/abs/2408.10419} {arXiv:2408.10419} \BibitemShut {NoStop}%
\bibitem [{\citenamefont {Yu}\ \emph {et~al.}(2024)\citenamefont {Yu}, \citenamefont {Xia}, \citenamefont {Ma}, \citenamefont {Lengyel},\ and\ \citenamefont {Hennequin}}]{yu_sofo_2024}%
  \BibitemOpen
  \bibfield  {author} {\bibinfo {author} {\bibfnamefont {Y.}~\bibnamefont {Yu}}, \bibinfo {author} {\bibfnamefont {R.}~\bibnamefont {Xia}}, \bibinfo {author} {\bibfnamefont {Q.}~\bibnamefont {Ma}}, \bibinfo {author} {\bibfnamefont {M.}~\bibnamefont {Lengyel}},\ and\ \bibinfo {author} {\bibfnamefont {G.}~\bibnamefont {Hennequin}},\ }\bibfield  {title} {\bibinfo {title} {Second-{Order} {Forward}-{Mode} {Optimization} of {Recurrent} {Neural} {Networks} for {Neuroscience}},\ }in\ \href {https://proceedings.neurips.cc/paper_files/paper/2024/hash/a791a086d7643ecf53608e57cd5889f0-Abstract-Conference.html} {\emph {\bibinfo {booktitle} {Advances in {Neural} {Information} {Processing} {Systems}}}},\ Vol.~\bibinfo {volume} {37}\ (\bibinfo {year} {2024})\BibitemShut {NoStop}%
\bibitem [{\citenamefont {Stokes}\ \emph {et~al.}(2020)\citenamefont {Stokes}, \citenamefont {Izaac}, \citenamefont {Killoran},\ and\ \citenamefont {Carleo}}]{stokes_quantum_2020}%
  \BibitemOpen
  \bibfield  {author} {\bibinfo {author} {\bibfnamefont {J.}~\bibnamefont {Stokes}}, \bibinfo {author} {\bibfnamefont {J.}~\bibnamefont {Izaac}}, \bibinfo {author} {\bibfnamefont {N.}~\bibnamefont {Killoran}},\ and\ \bibinfo {author} {\bibfnamefont {G.}~\bibnamefont {Carleo}},\ }\bibfield  {title} {\bibinfo {title} {Quantum {Natural} {Gradient}},\ }\href {https://doi.org/10.22331/q-2020-05-25-269} {\bibfield  {journal} {\bibinfo  {journal} {Quantum}\ }\textbf {\bibinfo {volume} {4}},\ \bibinfo {pages} {269} (\bibinfo {year} {2020})}\BibitemShut {NoStop}%
\bibitem [{\citenamefont {Mari}\ \emph {et~al.}(2021)\citenamefont {Mari}, \citenamefont {Bromley},\ and\ \citenamefont {Killoran}}]{mari_estimating_2021}%
  \BibitemOpen
  \bibfield  {author} {\bibinfo {author} {\bibfnamefont {A.}~\bibnamefont {Mari}}, \bibinfo {author} {\bibfnamefont {T.~R.}\ \bibnamefont {Bromley}},\ and\ \bibinfo {author} {\bibfnamefont {N.}~\bibnamefont {Killoran}},\ }\bibfield  {title} {\bibinfo {title} {Estimating the gradient and higher-order derivatives on quantum hardware},\ }\href {https://doi.org/10.1103/PhysRevA.103.012405} {\bibfield  {journal} {\bibinfo  {journal} {Physical Review A}\ }\textbf {\bibinfo {volume} {103}},\ \bibinfo {pages} {012405} (\bibinfo {year} {2021})}\BibitemShut {NoStop}%
\bibitem [{\citenamefont {Parrish}\ \emph {et~al.}(2021)\citenamefont {Parrish}, \citenamefont {Anselmetti},\ and\ \citenamefont {Gogolin}}]{parrish_analytical_2021}%
  \BibitemOpen
  \bibfield  {author} {\bibinfo {author} {\bibfnamefont {R.~M.}\ \bibnamefont {Parrish}}, \bibinfo {author} {\bibfnamefont {G.-L.~R.}\ \bibnamefont {Anselmetti}},\ and\ \bibinfo {author} {\bibfnamefont {C.}~\bibnamefont {Gogolin}},\ }\href@noop {} {\bibinfo {title} {Analytical {Ground}- and {Excited}-{State} {Gradients} for {Molecular} {Electronic} {Structure} {Theory} from {Hybrid} {Quantum}/{Classical} {Methods}}} (\bibinfo {year} {2021}),\ \Eprint {https://arxiv.org/abs/2110.05040} {arXiv:2110.05040} \BibitemShut {NoStop}%
\bibitem [{\citenamefont {Wolf}(2023)}]{wolf_mathematical_2023}%
  \BibitemOpen
  \bibfield  {author} {\bibinfo {author} {\bibfnamefont {M.~M.}\ \bibnamefont {Wolf}},\ }\href@noop {} {\emph {\bibinfo {title} {Mathematical {Foundations} of {Supervised} {Learning}}}}\ (\bibinfo  {publisher} {Lecture notes, Technical University of Munich},\ \bibinfo {year} {2023})\BibitemShut {NoStop}%
\bibitem [{\citenamefont {Talagrand}(1995)}]{talagrand_concentration_1995}%
  \BibitemOpen
  \bibfield  {author} {\bibinfo {author} {\bibfnamefont {M.}~\bibnamefont {Talagrand}},\ }\bibfield  {title} {\bibinfo {title} {Concentration of measure and isoperimetric inequalities in product spaces},\ }\href {https://doi.org/10.1007/BF02699376} {\bibfield  {journal} {\bibinfo  {journal} {Publications Math\'ematiques de l'IH\'ES}\ }\textbf {\bibinfo {volume} {81}},\ \bibinfo {pages} {73} (\bibinfo {year} {1995})}\BibitemShut {NoStop}%
\bibitem [{\citenamefont {Cerezo}\ \emph {et~al.}(2021{\natexlab{b}})\citenamefont {Cerezo}, \citenamefont {Sone}, \citenamefont {Volkoff}, \citenamefont {Cincio},\ and\ \citenamefont {Coles}}]{cerezo_cost_2021}%
  \BibitemOpen
  \bibfield  {author} {\bibinfo {author} {\bibfnamefont {M.}~\bibnamefont {Cerezo}}, \bibinfo {author} {\bibfnamefont {A.}~\bibnamefont {Sone}}, \bibinfo {author} {\bibfnamefont {T.}~\bibnamefont {Volkoff}}, \bibinfo {author} {\bibfnamefont {L.}~\bibnamefont {Cincio}},\ and\ \bibinfo {author} {\bibfnamefont {P.~J.}\ \bibnamefont {Coles}},\ }\bibfield  {title} {\bibinfo {title} {Cost function dependent barren plateaus in shallow parametrized quantum circuits},\ }\href {https://doi.org/10.1038/s41467-021-21728-w} {\bibfield  {journal} {\bibinfo  {journal} {Nature Communications}\ }\textbf {\bibinfo {volume} {12}},\ \bibinfo {pages} {1791} (\bibinfo {year} {2021}{\natexlab{b}})}\BibitemShut {NoStop}%
\bibitem [{\citenamefont {Kandala}\ \emph {et~al.}(2017)\citenamefont {Kandala}, \citenamefont {Mezzacapo}, \citenamefont {Temme}, \citenamefont {Takita}, \citenamefont {Brink}, \citenamefont {Chow},\ and\ \citenamefont {Gambetta}}]{kandala_hardware-efficient_2017}%
  \BibitemOpen
  \bibfield  {author} {\bibinfo {author} {\bibfnamefont {A.}~\bibnamefont {Kandala}}, \bibinfo {author} {\bibfnamefont {A.}~\bibnamefont {Mezzacapo}}, \bibinfo {author} {\bibfnamefont {K.}~\bibnamefont {Temme}}, \bibinfo {author} {\bibfnamefont {M.}~\bibnamefont {Takita}}, \bibinfo {author} {\bibfnamefont {M.}~\bibnamefont {Brink}}, \bibinfo {author} {\bibfnamefont {J.~M.}\ \bibnamefont {Chow}},\ and\ \bibinfo {author} {\bibfnamefont {J.~M.}\ \bibnamefont {Gambetta}},\ }\bibfield  {title} {\bibinfo {title} {Hardware-efficient variational quantum eigensolver for small molecules and quantum magnets},\ }\href {https://doi.org/10.1038/nature23879} {\bibfield  {journal} {\bibinfo  {journal} {Nature}\ }\textbf {\bibinfo {volume} {549}},\ \bibinfo {pages} {242} (\bibinfo {year} {2017})}\BibitemShut {NoStop}%
\bibitem [{\citenamefont {Farhi}\ \emph {et~al.}(2014)\citenamefont {Farhi}, \citenamefont {Goldstone},\ and\ \citenamefont {Gutmann}}]{farhi_quantum_2014}%
  \BibitemOpen
  \bibfield  {author} {\bibinfo {author} {\bibfnamefont {E.}~\bibnamefont {Farhi}}, \bibinfo {author} {\bibfnamefont {J.}~\bibnamefont {Goldstone}},\ and\ \bibinfo {author} {\bibfnamefont {S.}~\bibnamefont {Gutmann}},\ }\href@noop {} {\bibinfo {title} {A quantum approximate optimization algorithm}} (\bibinfo {year} {2014}),\ \Eprint {https://arxiv.org/abs/1411.4028} {arXiv:1411.4028 [quant-ph]} \BibitemShut {NoStop}%
\bibitem [{\citenamefont {Herrman}\ \emph {et~al.}(2022)\citenamefont {Herrman}, \citenamefont {Lotshaw}, \citenamefont {Ostrowski}, \citenamefont {Humble},\ and\ \citenamefont {Siopsis}}]{herrman_multiangle_2022}%
  \BibitemOpen
  \bibfield  {author} {\bibinfo {author} {\bibfnamefont {R.}~\bibnamefont {Herrman}}, \bibinfo {author} {\bibfnamefont {P.~C.}\ \bibnamefont {Lotshaw}}, \bibinfo {author} {\bibfnamefont {J.}~\bibnamefont {Ostrowski}}, \bibinfo {author} {\bibfnamefont {T.~S.}\ \bibnamefont {Humble}},\ and\ \bibinfo {author} {\bibfnamefont {G.}~\bibnamefont {Siopsis}},\ }\bibfield  {title} {\bibinfo {title} {Multi-angle quantum approximate optimization algorithm},\ }\href {https://doi.org/10.1038/s41598-022-10555-8} {\bibfield  {journal} {\bibinfo  {journal} {Scientific Reports}\ }\textbf {\bibinfo {volume} {12}},\ \bibinfo {pages} {6781} (\bibinfo {year} {2022})},\ \Eprint {https://arxiv.org/abs/2109.11455} {arXiv:2109.11455} \BibitemShut {NoStop}%
\end{thebibliography}

%

\appendix

\onecolumngrid

\section{Unbiasedness of the forward gradient estimator}\label{app:fwd_unbiased}

We prove that the $V$-direction, $M$-shot forward gradient estimator~\eqref{eqn:fwd_estimator_full} is unbiased in the $\varepsilon \to 0$ limit, adapting the classical argument of~\cite{baydin_gradients_2022} to the quantum finite-difference setting with shot noise.

\begin{proposition}[Unbiasedness in the $\varepsilon \to 0$ limit]\label{prop:fwd_unbiased}
Let $f(\boldsymbol{\theta}) = \langle \psi(\boldsymbol{\theta})| \mathcal{O} |\psi(\boldsymbol{\theta})\rangle$ be a circuit expectation value, and write $\tilde{o}_m(\boldsymbol{\theta}')$ for the outcome of the $m$-th single-shot measurement of $\mathcal{O}$ on $|\psi(\boldsymbol{\theta}')\rangle$, satisfying $\mathbb{E}_m[\tilde{o}_m(\boldsymbol{\theta}')] = f(\boldsymbol{\theta}')$. Let $\{\boldsymbol{v}^\ell\}_{\ell=1}^V$ be drawn i.i.d.\ from a distribution $p$ on $\mathbb{R}^N$ satisfying $\mathbb{E}[\boldsymbol{v}] = \boldsymbol{0}$ and $\mathbb{E}[\boldsymbol{v}\boldsymbol{v}^\top] = I_N$. Suppose the central finite-difference stencil is consistent, in the sense that for every measurement realisation $m$,
\begin{equation}\label{eqn:stencil_consistency}
    \lim_{\varepsilon \to 0} \sum_{q=1}^{Q} \gamma_q^{\boldsymbol{v}, \varepsilon}\, \tilde{o}_m\bigl(\boldsymbol{\theta}_{\boldsymbol{v}, q}^\varepsilon\bigr) = \nabla_{\boldsymbol{v}} f_m(\boldsymbol{\theta}) = \boldsymbol{v}^\top \nabla f_m(\boldsymbol{\theta}),
\end{equation}
where $f_m$ denotes the single-shot estimator. Then the $V$-direction, $M$-shot forward gradient estimator
\begin{equation}\label{eqn:fwd_estimator_full}
    \widetilde{\boldsymbol{g}}^{\mathsf{F}}(\boldsymbol{\theta}) = \frac{1}{V} \sum_{\ell=1}^V \biggl(\frac{1}{M}\sum_{m=1}^M \sum_{q=1}^Q \gamma_q^{\boldsymbol{v}^\ell, \varepsilon}\, \tilde{o}_m\bigl(\boldsymbol{\theta}_{\boldsymbol{v}^\ell, q}^\varepsilon\bigr)\biggr) \boldsymbol{v}^\ell
\end{equation}
satisfies $\lim_{\varepsilon \to 0} \mathbb{E}\bigl[\widetilde{\boldsymbol{g}}^{\mathsf{F}}(\boldsymbol{\theta})\bigr] = \nabla f(\boldsymbol{\theta})$, with expectation taken jointly over directions and measurement noise.
\end{proposition}

\begin{proof}
By the tower property, $\mathbb{E}[\widetilde{\boldsymbol{g}}^{\mathsf{F}}] = \mathbb{E}_{\boldsymbol{v}}\bigl[\mathbb{E}_m[\widetilde{\boldsymbol{g}}^{\mathsf{F}} \mid \boldsymbol{v}]\bigr]$. The inner expectation over the unbiased single-shot quantum estimator $\mathbb{E}_m[\tilde{o}_m(\boldsymbol{\theta})] = \langle \psi(\boldsymbol{\theta})|\mathcal{O}|\psi(\boldsymbol{\theta})\rangle = f(\boldsymbol{\theta})$, combined with stencil consistency~\eqref{eqn:stencil_consistency}, gives
\begin{equation}
    \lim_{\varepsilon \to 0} \mathbb{E}_m\bigl[\widetilde{\boldsymbol{g}}^{\mathsf{F}} \mid \boldsymbol{v}\bigr] = \frac{1}{V}\sum_{\ell=1}^V \bigl(\boldsymbol{v}^{\ell\top} \nabla f(\boldsymbol{\theta})\bigr)\, \boldsymbol{v}^\ell.
\end{equation}
Taking the expectation over directions and using the i.i.d.\ assumption together with the second-moment identity $\mathbb{E}[\boldsymbol{v}\boldsymbol{v}^\top] = I_N$,
\begin{equation}
    \mathbb{E}_{\boldsymbol{v}}\bigl[(\boldsymbol{v}^\top \nabla f)\, \boldsymbol{v}\bigr] = \mathbb{E}_{\boldsymbol{v}}[\boldsymbol{v}\boldsymbol{v}^\top]\, \nabla f = \nabla f(\boldsymbol{\theta}),
\end{equation}
which is independent of $V$. Hence $\lim_{\varepsilon \to 0} \mathbb{E}[\widetilde{\boldsymbol{g}}^{\mathsf{F}}] = \nabla f(\boldsymbol{\theta})$.
\end{proof}

The proposition extends to losses linear in circuit observables (such as the MSE loss used in the main text) by linearity of expectation: the gradient of $\mathcal{L}(\boldsymbol{\theta}) = \nicefrac{1}{D}\sum_{d=1}^D \ell_d(\boldsymbol{\theta})$ inherits unbiasedness from each term $\ell_d$ by linearity. The finite-$\varepsilon$ bias is $\mathcal{O}(\varepsilon^2)$ for the central stencil~\cite{sweke_stochastic_2020}; the bias--variance balance leading to the $\varepsilon^\star \propto M^{-1/6}$ rule is derived in \appref{app:eps_star}.

\section{Bias--variance trade-off and the optimal step size}\label{app:eps_star}

We derive the closed-form optimum $\varepsilon^\star$ that balances finite-difference bias against shot-noise amplification, and verify it against the empirical U-curve of \figref{fig:eps_sweep}.

\subparagraph{Variance.} The central-difference estimator divides by $2\varepsilon$, so shot noise enters the variance as $\sigma_m^2/(2M\varepsilon^2)$ (derived in full below as~\eqref{eqn:fd_var}), amplified as $1/\varepsilon^2$, and taking $\varepsilon$ large suppresses noise but introduces $\mathcal{O}(\varepsilon^2 \|\nabla^3 f\|)$ Taylor-remainder bias. Because quantum expectation values are band-limited trigonometric polynomials~\cite{parrish_analytical_2021}, the bias term remains small up to surprisingly large $\varepsilon$, and the optimal choice is set almost entirely by the noise amplification.

\subparagraph{Closed-form $\varepsilon^\star$.} Expanding $f(\boldsymbol{\theta} \pm \varepsilon \boldsymbol{v})$ to fourth order around $\boldsymbol{\theta}$, the central difference estimator~\eqref{eqn:directional_derivative_approx_quantum} has expectation
\begin{equation} \label{eqn:fd_bias}
    \mathbb{E}_m\left[\widetilde{\nabla}_{\boldsymbol{v}}^{\varepsilon} f\right] = (\boldsymbol{\nabla} f \cdot \boldsymbol{v}) + \frac{\varepsilon^2}{6}\,\partial_{\boldsymbol{v}}^3 f + O(\varepsilon^4),
\end{equation}
where $\partial_{\boldsymbol{v}}^3 f := \sum_{i,j,k} v_i v_j v_k\, \partial_i \partial_j \partial_k f$ is the third directional derivative. The leading bias is therefore $\varepsilon^2 C_3 / 6$ where $C_3 := |\partial_{\boldsymbol{v}}^3 f|$ depends on $\boldsymbol{v}$ and $\boldsymbol{\theta}$ but not on $\varepsilon$, while the variance is
\begin{equation} \label{eqn:fd_var}
    \operatorname{Var}_m\left[\widetilde{\nabla}_{\boldsymbol{v}}^{\varepsilon} f\right] = \frac{\sigma_{m}^2}{2 M \varepsilon^2},
\end{equation}
where $\sigma_m^2$ is the per-shot variance of the observable at $\boldsymbol{\theta} \pm \varepsilon \boldsymbol{v}$ (upper-bounded by the operator norm of the Hamiltonian). The mean-squared error of the estimator is the sum of squared bias and variance,
\begin{equation} \label{eqn:fd_mse}
    \operatorname{MSE}(\varepsilon) = \frac{\varepsilon^4 C_3^2}{36} + \frac{\sigma_m^2}{2 M \varepsilon^2}.
\end{equation}
Minimising over $\varepsilon$, the first-order condition $\nicefrac{d}{d\varepsilon}\operatorname{MSE} = 0$ yields the closed-form optimum
\begin{equation} \label{eqn:eps_star}
    \varepsilon^{\star} = \left(\,\frac{9\, \sigma_m^2}{M\, C_3^2}\,\right)^{1/6}
\end{equation}
with minimum MSE $\operatorname{MSE}(\varepsilon^\star) = \nicefrac{1}{2}\,[3^{-1/3} + 3^{2/3}]\,(C_3\,\sigma_m^2/M)^{2/3}$.

\subparagraph{Scaling properties.} Three properties follow directly from~\eqref{eqn:eps_star}: (i)~$\varepsilon^\star \propto M^{-1/6}$, so doubling $M$ decreases the optimal step by only $\approx 11\%$ and transferring $\varepsilon^\star$ between runs of different $M$ is safe over the factor-of-ten range typical in practice; (ii)~$\varepsilon^\star \propto \sigma_m^{1/3}$, so a noisier observable calls for a \emph{larger} step since more shot noise can be absorbed into a larger denominator; (iii)~$\varepsilon^\star \propto C_3^{-1/3}$, so a flatter (more band-limited) landscape tolerates a larger step. Band-limitation to frequency $\sqrt{N}$~\cite{parrish_analytical_2021} means $C_3 = O(N^{3/2} \|H\|)$ in the worst case, so $\varepsilon^\star$ is weakly $N$-dependent, decreasing as $N^{-1/2}$.

\subparagraph{Calibration and transferability.}
\begin{wrapfigure}{r}{0.46\linewidth}
    \vspace{-6pt}
    \centering
    \includegraphics[width=\linewidth]{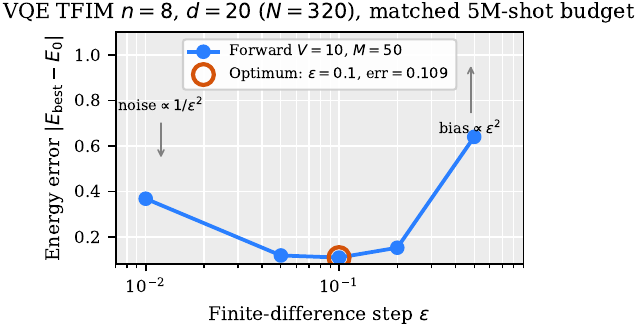}
    \caption{\textbf{$\varepsilon = 0.1$ sits at the bias--variance optimum.} VQE TFIM $n{=}8$, $d{=}20$, $V{=}10$, $M{=}50$, matched 5M-shot budget, seed 0. Noise amplification ($\propto 1/\varepsilon^2$) dominates at small $\varepsilon$; finite-difference bias ($\propto \varepsilon^2$) dominates at large $\varepsilon$.}
    \label{fig:eps_sweep}
    \vspace{-4pt}
\end{wrapfigure}
\figref{fig:eps_sweep} shows the empirical U-curve on a representative VQE problem: at fixed $V = 10$ and matched shot budget, the best energy error is achieved at $\varepsilon = 0.1$, with a near-$3\times$ degradation at $\varepsilon = 0.01$ (noise-dominated) and a $5\times$ degradation at $\varepsilon = 0.5$ (bias-dominated). The closed form~\eqref{eqn:eps_star} is most usefully read as a transferability tool rather than as a calibration-free predictor: the constant $C_3$ is problem-specific and is recovered from a single U-curve fit (here, $C_3 \sim 10^3$ at $M = 50$ and $\sigma_m^2 \sim \|H\|^2 \sim 10$, recovering the empirical $\varepsilon^\star \approx 0.1$ used throughout the main text), but once $C_3$ is calibrated at any one operating point the scaling laws $\varepsilon^\star \propto M^{-1/6}$, $\sigma_m^{1/3}$, $N^{-1/2}$ make the optimal step transferable across $(M, \sigma_m, N)$ without re-running the full sweep.

\section{Second-moment expansion of the forward gradient estimator}\label{app:second_moment}

We prove \lemref{lem:exp_over_vectors_norm_grad_squared}, which bounds the second moment of the $V$-direction forward gradient estimator and underlies the convergence rates of \secref{sec:convergence}.

\begin{lemma}\label{lem:exp_over_vectors_norm_grad_squared_app}
Let $\{\boldsymbol{v}^\ell\}_{\ell=1}^V$ be drawn i.i.d.\ with $\mathbb{E}[v_i^\ell] = 0$, $\mathbb{E}[v_i^\ell v_{i'}^\ell] = \delta_{ii'}$, and let $\kappa := \mathbb{E}[(v_i^\ell)^4]$ be the fourth moment of one component (so $\kappa = 1$ for Rademacher and $\kappa = 3$ for standard Gaussian). Let $\widetilde{\boldsymbol{\nabla}}^{\mathsf{F}} \mathcal{L}(\boldsymbol{\theta}) = \nicefrac{1}{V}\sum_{\ell=1}^V (\widetilde{\nabla}_{\boldsymbol{v}^\ell} \mathcal{L}^M)\, \boldsymbol{v}^\ell$ be the $V$-direction $M$-shot forward gradient estimator. Then
\begin{equation}\label{eqn:second_moment_app}
    \mathbb{E}_{\boldsymbol{v}}\left[\bigl\|\widetilde{\boldsymbol{\nabla}}^{\mathsf{F}} \mathcal{L}(\boldsymbol{\theta})\bigr\|^2\right] = \frac{N + V + \kappa - 2}{V}\, \bigl\|\widetilde{\boldsymbol{\nabla}} \mathcal{L}^M\bigr\|^2,
\end{equation}
where $\widetilde{\boldsymbol{\nabla}} \mathcal{L}^M$ is the (full-coordinate) $M$-shot estimator of the true gradient. In the noiseless limit $M \to \infty$ this reduces to $\bigl\|\widetilde{\boldsymbol{\nabla}}^{\mathsf{F}} \mathcal{L}\bigr\|^2 \to (N{+}V{+}\kappa{-}2)/V \cdot \|\nabla \mathcal{L}\|^2$, which is the bound used in part~2 of \lemref{thm:sgd_pl_app} below.
\end{lemma}

\begin{proof}
Write $Z_k := [\widetilde{\boldsymbol{\nabla}}^{\mathsf{F}} \mathcal{L}]_k = \nicefrac{1}{VM}\sum_{m=1}^M \sum_{\ell=1}^V \sum_{j=1}^N v_j^\ell\, (\partial_j f_m)\, v_k^\ell$, where $\partial_j f_m$ is the (single-shot) partial derivative of the loss in coordinate $j$. Squaring and summing,
\begin{equation}
    \mathbb{E}_{\boldsymbol{v}}\left[\|\boldsymbol{Z}\|^2\right] = \frac{1}{V^2 M^2} \sum_{k=1}^N \sum_{\substack{m,m'=1}}^M \sum_{\substack{\ell, \ell'=1}}^V \sum_{\substack{j, j'=1}}^N (\partial_j f_m)(\partial_{j'} f_{m'})\, \mathbb{E}\left[v_j^\ell v_k^\ell v_{j'}^{\ell'} v_k^{\ell'}\right].
\end{equation}
The fourth-moment object $\mathbb{E}[v_j^\ell v_k^\ell v_{j'}^{\ell'} v_k^{\ell'}]$ depends on whether $\ell = \ell'$ and on the index pattern in $(j, j', k)$:
\begin{itemize}
    \item \textbf{$\ell \neq \ell'$.} By independence, $\mathbb{E}[v_j^\ell v_k^\ell] \mathbb{E}[v_{j'}^{\ell'} v_k^{\ell'}] = \delta_{jk}\delta_{j'k}$.
    \item \textbf{$\ell = \ell'$.} The expectation is $\mathbb{E}[v_j^\ell v_{j'}^\ell (v_k^\ell)^2]$, with three non-vanishing index patterns:
    \begin{itemize}
        \item $j = j' = k$: $\mathbb{E}[(v_k^\ell)^4] = \kappa$;
        \item $j = j' \neq k$: $\mathbb{E}[(v_j^\ell)^2 (v_k^\ell)^2] = 1$;
        \item $j \neq j'$, or any pattern where a lone $v$ component appears: $0$ by zero-mean.
    \end{itemize}
\end{itemize}
Substituting these cases,
\begin{align}
    \mathbb{E}_{\boldsymbol{v}}\left[\|\boldsymbol{Z}\|^2\right] &= \frac{1}{V^2 M^2} \sum_{m, m'=1}^M\, V\, \biggl(\kappa \sum_{k=1}^N (\partial_k f_m)(\partial_k f_{m'}) + (N - 1)\sum_{j=1}^N (\partial_j f_m)(\partial_j f_{m'}) \notag \\
    &\hspace{12em}+ (V - 1)\sum_{k=1}^N (\partial_k f_m)(\partial_k f_{m'})\biggr) \\
    &= \frac{N + V + \kappa - 2}{V}\, \frac{1}{M^2}\sum_{k=1}^N \sum_{m, m'=1}^M (\partial_k f_m)(\partial_k f_{m'}) \\
    &= \frac{N + V + \kappa - 2}{V}\, \bigl\|\widetilde{\boldsymbol{\nabla}} \mathcal{L}^M\bigr\|^2,
\end{align}
which is~\eqref{eqn:second_moment_app}. Taking $M \to \infty$ replaces $\widetilde{\boldsymbol{\nabla}} \mathcal{L}^M$ by the true gradient $\nabla \mathcal{L}$.
\end{proof}

For Rademacher directions $\kappa = 1$ the prefactor reduces to $(N + V - 1)/V$; for standard Gaussian $\kappa = 3$ it becomes $(N + V + 1)/V$. Both choices are bounded above by $(N + V + 1)/V$, so the convergence rate~\eqref{eqn:fwd_convergence_exact} obtained by substituting $\beta^2 = (N + V + \kappa - 2)/V$ into \lemref{thm:sgd_pl_app} (\appref{app:pl_sgd}) is robust across reasonable direction distributions.

\section{Polyak--\L{}ojasiewicz inequality and stochastic gradient descent}\label{app:pl_sgd}\label{app:pl_convergence_fwd}

We state the Polyak--\L{}ojasiewicz (PL) inequality and the PL-SGD convergence theorem~\cite{wolf_mathematical_2023, sweke_stochastic_2020}, then prove \propref{prop:fwd_convergence} by substituting the second-moment bound of \appref{app:second_moment} into the SGD bound.

\subparagraph{The PL inequality.} A function $f \in C^1(\mathbb{R}^N)$ that attains its global minimum at $\boldsymbol{\theta}^\star$ is said to satisfy the (global) PL inequality with constant $\mu > 0$ if
\begin{equation} \label{eqn:global_PL_inequality_app}
    \tfrac{1}{2}\,\|\nabla f(\boldsymbol{\theta})\|^2 \geq \mu\,\bigl(f(\boldsymbol{\theta}) - f(\boldsymbol{\theta}^\star)\bigr) \qquad \forall \boldsymbol{\theta} \in \mathbb{R}^N.
\end{equation}
The PL condition is in general not satisfied globally for parameterised quantum circuits~\cite{ding_random_2024}, but it is often a reasonable assumption locally. Writing $\mathcal{X} = \{\boldsymbol{\theta}^\star\}$ for the set of global minima, the \emph{local} PL condition states that there exist $\delta_f, \mu > 0$ such that for all $\boldsymbol{\theta}$ in the sublevel neighbourhood $\mathcal{N}(\mathcal{X}) := f^{-1}([f(\boldsymbol{\theta}^\star), \delta_f])$,
\begin{equation} \label{eqn:local_PL_inequality_app}
    \tfrac{1}{2}\,\|\nabla f(\boldsymbol{\theta})\|^2 \geq \mu\,\bigl(f(\boldsymbol{\theta}) - f(\boldsymbol{\theta}^\star)\bigr) \qquad \forall \boldsymbol{\theta} \in \mathcal{N}(\mathcal{X}).
\end{equation}
The local condition is what is invoked implicitly throughout \secref{sec:convergence}: convergence is to a neighbourhood of a minimum, not to the global minimum.

\begin{proposition}[\propref{prop:fwd_convergence} restated]\label{prop:fwd_convergence_app}
Let $f \in C^1(\mathbb{R}^N)$ be $L$-smooth and satisfy the PL inequality with constant $\mu > 0$, and let directions be drawn from the Rademacher distribution ($\kappa = 1$).

\emph{(i) Convergence of forward gradient descent.} With learning rate $\eta = V / (L(N + V - 1))$, forward gradient descent satisfies
\begin{equation}
    \mathbb{E}[f(\boldsymbol{\theta}^{(T)})] - f^\star \leq\left(1 - \frac{\mu V}{L(N{+}V{-}1)}\right)^{\!T} \!\bigl(f(\boldsymbol{\theta}^{(0)}) - f^\star\bigr).
\end{equation}

\emph{(ii) Convergence under shot noise.} With per-direction shot noise variance $\sigma^2_{\mathrm{shot}}$ and $\eta \in [0, 1/(2\mu)]$,
\begin{equation}
    \mathbb{E}[f(\boldsymbol{\theta}^{(T)})] - f^\star \leq\left(1 - 2\mu\eta\right)^T \!\bigl(f(\boldsymbol{\theta}^{(0)}) - f^\star\bigr) + \frac{L\eta(N{+}V{-}1)\,\sigma^2_{\mathrm{shot}}}{4\mu VM}.
\end{equation}
\end{proposition}

The proof uses the following cited lemma.

\begin{lemma}[SGD under PL and bounded second moment, \cite{wolf_mathematical_2023, sweke_stochastic_2020}]\label{thm:sgd_pl_app}
Let $f \in C^1(\mathbb{R}^N)$ be $L$-smooth, satisfy the PL inequality~\eqref{eqn:global_PL_inequality_app} with constant $\mu > 0$, and attain its global minimum at $\boldsymbol{\theta}^\star$. Let $\boldsymbol{g}^{(1)}, \dots, \boldsymbol{g}^{(T)}$ be unbiased gradient estimators and consider the iterates $\boldsymbol{\theta}^{(t+1)} = \boldsymbol{\theta}^{(t)} - \eta\, \boldsymbol{g}^{(t)}(\boldsymbol{\theta}^{(t)})$.
\begin{enumerate}
    \item If $\mathbb{E}[\|\boldsymbol{g}^{(t)}(\boldsymbol{\theta})\|^2] \leq \gamma^2$ for all $\boldsymbol{\theta}, t$ and $\eta \in [0, 1/(2\mu)]$, then
    \begin{equation}
        \mathbb{E}[f(\boldsymbol{\theta}^{(T)})] - f(\boldsymbol{\theta}^\star) \leq (1 - 2\mu\eta)^T \bigl(f(\boldsymbol{\theta}^{(0)}) - f(\boldsymbol{\theta}^\star)\bigr) + \frac{L \eta\, \gamma^2}{4\mu}.
    \end{equation}
    \item If $\mathbb{E}[\|\boldsymbol{g}^{(t)}(\boldsymbol{\theta})\|^2] \leq \beta^2 \|\nabla f(\boldsymbol{\theta})\|^2$ for all $\boldsymbol{\theta}, t$ and $\eta = 1/(L\beta^2)$, then
    \begin{equation}
        \mathbb{E}[f(\boldsymbol{\theta}^{(T)})] - f(\boldsymbol{\theta}^\star) \leq \Bigl(1 - \frac{\mu}{L\beta^2}\Bigr)^T \bigl(f(\boldsymbol{\theta}^{(0)}) - f(\boldsymbol{\theta}^\star)\bigr).
    \end{equation}
\end{enumerate}
\end{lemma}

\begin{proof}[Proof of \propref{prop:fwd_convergence_app}]
\textbf{Part (i).}
By \appref{app:fwd_unbiased}, $\mathbb{E}[\widetilde{\boldsymbol{g}}^{\mathsf{F}}(\boldsymbol{\theta})] = \nabla f(\boldsymbol{\theta})$, so the estimator is unbiased.
By \lemref{lem:exp_over_vectors_norm_grad_squared_app} with $\kappa = 1$ (Rademacher),
\begin{equation*}
    \mathbb{E}\bigl[\|\widetilde{\boldsymbol{g}}^{\mathsf{F}}\|^2\bigr] = \frac{N+V-1}{V}\,\|\nabla f\|^2 =:\; \beta^2\,\|\nabla f\|^2.
\end{equation*}
This is the bounded relative second moment condition of part~2 of \lemref{thm:sgd_pl_app}. Setting $\eta = 1/(L\beta^2) = V/(L(N+V-1))$ and substituting $\beta^2 = (N+V-1)/V$,
\begin{equation*}
    \mathbb{E}[f(\boldsymbol{\theta}^{(T)})] - f^\star \leq \Bigl(1 - \frac{\mu}{L\beta^2}\Bigr)^T(f(\boldsymbol{\theta}^{(0)}) - f^\star)
    = \Bigl(1 - \frac{\mu V}{L(N+V-1)}\Bigr)^T(f(\boldsymbol{\theta}^{(0)}) - f^\star),
\end{equation*}
which is part~(i).

\textbf{Part (ii).}
With $M$ shots per direction the $M$-shot estimator $\widetilde{\boldsymbol{\nabla}}^{\mathsf{F}}\mathcal{L}^M$ adds a shot-noise contribution. By the variance decomposition (proof of \lemref{lem:exp_over_vectors_norm_grad_squared_app}), the full second moment is
\begin{equation*}
    \mathbb{E}\bigl[\|\widetilde{\boldsymbol{g}}^{\mathsf{F}}\|^2\bigr]
    = \frac{N+V-1}{V}\,\|\nabla f\|^2 + \frac{N+V-1}{V}\cdot\frac{\sigma^2_{\mathrm{shot}}}{M}
    \leq \frac{N+V-1}{V}\|\nabla f\|^2 + \gamma^2,
\end{equation*}
where $\gamma^2 := (N+V-1)\sigma^2_{\mathrm{shot}}/(VM)$ is the shot-noise residual. This satisfies the bounded absolute second moment condition of part~1 of \lemref{thm:sgd_pl_app} with this $\gamma^2$. Substituting into part~1 gives
\begin{equation*}
    \mathbb{E}[f(\boldsymbol{\theta}^{(T)})] - f^\star \leq (1-2\mu\eta)^T(f(\boldsymbol{\theta}^{(0)}) - f^\star) + \frac{L\eta\gamma^2}{4\mu}
    = (1-2\mu\eta)^T(f(\boldsymbol{\theta}^{(0)}) - f^\star) + \frac{L\eta(N+V-1)\sigma^2_{\mathrm{shot}}}{4\mu VM},
\end{equation*}
which is part~(ii).
\end{proof}

\section{Expected gain decomposition and per-direction shot allocation}\label{app:gain_proof}

We derive the per-direction gain decomposition~\eqref{eqn:gain_signal_noise_main} for forward gradient descent, then derive the optimal per-direction shot allocation $M_\ell^\star$ that maximises gain per measurement.

\subparagraph{Proof strategy.} The argument proceeds in four steps.
\begin{enumerate}
    \item \emph{Simplify the gain via unbiasedness.} Starting from the one-step expected loss decrease~\eqref{eqn:gain_definition}, the cross-term $\nabla\mathcal{L}^\top\mathbb{E}[\widetilde{\boldsymbol{\nabla}}^{\mathsf{F}}\mathcal{L}]$ reduces to $\|\nabla\mathcal{L}\|^2$ by unbiasedness, leaving the second moment $\mathbb{E}[\|\widetilde{\boldsymbol{\nabla}}^{\mathsf{F}}\mathcal{L}\|^2]$ as the only quantity to bound.
    \item \emph{Decompose the second moment per direction.} \lemref{lem:variance_with_meas_exp_app} (proved below) splits $\mathbb{E}[\|\widetilde{\boldsymbol{\nabla}}^{\mathsf{F}}\mathcal{L}\|^2]$ into a sum over directions, each contributing a signal term $(\nabla_{\boldsymbol{v}^\ell}\mathcal{L})^2$ and a shot-noise term $\operatorname{Var}_m[\widetilde{\nabla}_{\boldsymbol{v}^\ell}\mathcal{L}_m]/M$.
    \item \emph{Derive the per-direction gain formula.} Substituting Steps 1--2 into the gain and grouping by direction yields~\eqref{eqn:gain_per_direction_app}, from which the learning-rate criterion follows directly.
    \item \emph{Optimise the shot allocation.} Treating the per-direction gain as a function of $M_\ell$ alone and solving the first-order condition gives the closed-form $M_\ell^\star$~\eqref{eqn:optimal_M_ell_app}. Under noise concentration (\assumpref{assump:noise_concentration_main}) the same approach applied to $V$ at fixed $M$ gives the fixed-$M$ optimal $V^\star$~\eqref{eqn:adaptive_v_rule}.
\end{enumerate}

\subparagraph{Notation.} Throughout this appendix we work with two sources of randomness: the random forward directions $\{\boldsymbol{v}^\ell\}_{\ell=1}^V$ (drawn i.i.d.\ with zero mean, unit variance, kurtosis $\kappa$) and the per-shot quantum measurement noise. We write $\mathbb{E}_{\boldsymbol{v}}$, $\mathbb{E}_m$ for the corresponding expectations, and $\mathbb{E} = \mathbb{E}_{\boldsymbol{v}}\mathbb{E}_m$ for the joint expectation. For a directional derivative we write $\nabla_{\boldsymbol{v}}\mathcal{L}$ for the true value and $\widetilde{\nabla}_{\boldsymbol{v}}\mathcal{L}^M = \nicefrac{1}{M}\sum_m \widetilde{\nabla}_{\boldsymbol{v}}\mathcal{L}_m$ for the $M$-shot Monte Carlo estimator.

\subparagraph{Setup.} Starting from~\eqref{eqn:gain_definition}, the unbiasedness of the forward gradient estimator (\propref{prop:fwd_unbiased}) implies $\nabla\mathcal{L}^\top \mathbb{E}[\widetilde{\boldsymbol{\nabla}}^{\mathsf{F}}\mathcal{L}] = \|\nabla\mathcal{L}\|^2$, so the gain expression simplifies to
\begin{equation}\label{eqn:gain_simplified_app}
    \mathbb{E}[\mathcal{G}^{\mathsf{F}}] = \eta\, \|\nabla\mathcal{L}\|^2  -  \frac{L\eta^2}{2}\, \mathbb{E}\left[\bigl\|\widetilde{\boldsymbol{\nabla}}^{\mathsf{F}} \mathcal{L}\bigr\|^2\right].
\end{equation}
The second-moment term is exactly the object bounded by \lemref{lem:exp_over_vectors_norm_grad_squared} of~\appref{app:second_moment}.

\subsection{Variance decomposition over measurements}

To express the gain in a form where each random direction contributes a separate signal and noise term we decompose the measurement-side expectation of the directional-derivative variance.

\begin{lemma}[Variance-with-measurement decomposition]\label{lem:variance_with_meas_exp_app}
With unbiased single-shot estimators $\mathbb{E}_m[\widetilde{\nabla}_{\boldsymbol{v}^\ell}\mathcal{L}_m] = \nabla_{\boldsymbol{v}^\ell}\mathcal{L}$ and i.i.d.\ measurement trials,
\begin{equation}\label{eqn:variance_with_meas_app}
    \mathbb{E}_m\left[\operatorname{Var}_{\boldsymbol{v}}\widetilde{\nabla}_{\boldsymbol{v}}\mathcal{L}^M\right] \approx \frac{1}{V}\sum_{\ell=1}^V \biggl(\bigl(\nabla_{\boldsymbol{v}^\ell}\mathcal{L}\bigr)^2 + \frac{1}{M}\operatorname{Var}_m\left[\widetilde{\nabla}_{\boldsymbol{v}^\ell}\mathcal{L}_m\right]\biggr).
\end{equation}
\end{lemma}

\begin{proof}
Since the directional derivatives have zero mean over directions (by isotropy), $\operatorname{Var}_{\boldsymbol{v}}\widetilde{\nabla}_{\boldsymbol{v}}\mathcal{L}^M = \mathbb{E}_{\boldsymbol{v}}[(\widetilde{\nabla}_{\boldsymbol{v}}\mathcal{L}^M)^2]$. Approximate this expectation by the empirical average over the $V$ sampled directions:
\begin{equation*}
    \mathbb{E}_{\boldsymbol{v}}\bigl[(\widetilde{\nabla}_{\boldsymbol{v}}\mathcal{L}^M)^2\bigr] \approx \frac{1}{V}\sum_{\ell=1}^V \bigl(\widetilde{\nabla}_{\boldsymbol{v}^\ell}\mathcal{L}^M\bigr)^2.
\end{equation*}
Expand each squared term using $\widetilde{\nabla}_{\boldsymbol{v}^\ell}\mathcal{L}^M = \nicefrac{1}{M}\sum_{m=1}^M \widetilde{\nabla}_{\boldsymbol{v}^\ell}\mathcal{L}_m$:
\begin{equation*}
    \bigl(\widetilde{\nabla}_{\boldsymbol{v}^\ell}\mathcal{L}^M\bigr)^2 = \frac{1}{M^2}\sum_{m,m'=1}^M \widetilde{\nabla}_{\boldsymbol{v}^\ell}\mathcal{L}_m\,\widetilde{\nabla}_{\boldsymbol{v}^\ell}\mathcal{L}_{m'}.
\end{equation*}
Taking $\mathbb{E}_m$ of each summand splits into two cases.

\textbf{Off-diagonal ($m \neq m'$).} The $M(M-1)$ such terms factorise by independence of the two shots:
\begin{equation*}
    \mathbb{E}_m\left[\widetilde{\nabla}_{\boldsymbol{v}^\ell}\mathcal{L}_m \cdot \widetilde{\nabla}_{\boldsymbol{v}^\ell}\mathcal{L}_{m'}\right] = \bigl(\nabla_{\boldsymbol{v}^\ell}\mathcal{L}\bigr)^2.
\end{equation*}

\textbf{Diagonal ($m = m'$).} The $M$ such terms use $\mathbb{E}_m[X^2] = (\mathbb{E}_m X)^2 + \operatorname{Var}_m X$:
\begin{equation*}
    \mathbb{E}_m\left[\bigl(\widetilde{\nabla}_{\boldsymbol{v}^\ell}\mathcal{L}_m\bigr)^2\right] = \bigl(\nabla_{\boldsymbol{v}^\ell}\mathcal{L}\bigr)^2 + \operatorname{Var}_m\left[\widetilde{\nabla}_{\boldsymbol{v}^\ell}\mathcal{L}_m\right].
\end{equation*}
Combining diagonal and off-diagonal contributions and dividing by $M^2$:
\begin{align*}
    \mathbb{E}_m\left[\bigl(\widetilde{\nabla}_{\boldsymbol{v}^\ell}\mathcal{L}^M\bigr)^2\right]
    &= \frac{1}{M^2}\Bigl[M\bigl((\nabla_{\boldsymbol{v}^\ell}\mathcal{L})^2 + \operatorname{Var}_m[\widetilde{\nabla}_{\boldsymbol{v}^\ell}\mathcal{L}_m]\bigr) + M(M-1)(\nabla_{\boldsymbol{v}^\ell}\mathcal{L})^2\Bigr] \\
    &= \bigl(\nabla_{\boldsymbol{v}^\ell}\mathcal{L}\bigr)^2 + \frac{1}{M}\operatorname{Var}_m\left[\widetilde{\nabla}_{\boldsymbol{v}^\ell}\mathcal{L}_m\right].
\end{align*}
Averaging over $\ell$ yields~\eqref{eqn:variance_with_meas_app}.
\end{proof}

\subsection{Per-direction gain and learning-rate criterion}

Substituting \lemref{lem:exp_over_vectors_norm_grad_squared} and \lemref{lem:variance_with_meas_exp_app} into~\eqref{eqn:gain_simplified_app}:
\begin{align*}
    \mathbb{E}[\mathcal{G}^{\mathsf{F}}]
    &= \eta\,\|\nabla\mathcal{L}\|^2 - \frac{L\eta^2}{2}\,\mathbb{E}\left[\|\widetilde{\boldsymbol{\nabla}}^{\mathsf{F}}\mathcal{L}\|^2\right] \\
    &\approx \eta\,\|\nabla\mathcal{L}\|^2 - \frac{L\eta^2}{2}\cdot\frac{N+V+\kappa-2}{V}\cdot\frac{1}{V}\sum_{\ell=1}^V\biggl((\nabla_{\boldsymbol{v}^\ell}\mathcal{L})^2 + \frac{\operatorname{Var}_m[\widetilde{\nabla}_{\boldsymbol{v}^\ell}\mathcal{L}_m]}{M}\biggr) \\
    &= \frac{1}{V}\sum_{\ell=1}^V \underbrace{\biggl[\eta\,\|\nabla\mathcal{L}\|^2 - \frac{L\eta^2}{2}\frac{N+V+\kappa-2}{V}\biggl((\nabla_{\boldsymbol{v}^\ell}\mathcal{L})^2 + \frac{\operatorname{Var}_m[\widetilde{\nabla}_{\boldsymbol{v}^\ell}\mathcal{L}_m]}{M}\biggr)\biggr]}_{\displaystyle =:\,\gamma_{\boldsymbol{v}^\ell}},
\end{align*}
where the second line uses \lemref{lem:exp_over_vectors_norm_grad_squared} for the second-moment term and \lemref{lem:variance_with_meas_exp_app} for the per-direction variance, and the third line folds the constant signal term $\eta\|\nabla\mathcal{L}\|^2$ inside the sum. This gives the per-direction decomposition
\begin{equation}\label{eqn:gain_per_direction_app}
    \mathbb{E}[\mathcal{G}^{\mathsf{F}}] \approx \frac{1}{V}\sum_{\ell=1}^V \gamma_{\boldsymbol{v}^\ell}, \qquad \gamma_{\boldsymbol{v}^\ell} := \eta\, \|\nabla\mathcal{L}\|^2  -  \frac{L\eta^2}{2}\frac{N + V + \kappa - 2}{V} \biggl(\bigl(\nabla_{\boldsymbol{v}^\ell}\mathcal{L}\bigr)^2 + \frac{1}{M}\operatorname{Var}_m\left[\widetilde{\nabla}_{\boldsymbol{v}^\ell}\mathcal{L}_m\right]\biggr),
\end{equation}
which is~\eqref{eqn:gain_signal_noise_main} of the main text. Each direction contributes a (constant) signal term $\eta\,\|\nabla\mathcal{L}\|^2$ and a penalty composed of a directional-derivative term $(\nabla_{\boldsymbol{v}^\ell}\mathcal{L})^2$ (the part of the gradient projected onto $\boldsymbol{v}^\ell$) and a measurement-noise term $\operatorname{Var}_m[\widetilde{\nabla}_{\boldsymbol{v}^\ell}\mathcal{L}_m]/M$.

Requiring $\mathbb{E}[\mathcal{G}^{\mathsf{F}}] > 0$ (i.e.\ the loss decreases in expectation) gives the learning-rate criterion
\begin{equation}\label{eqn:lr_criterion_app}
    \eta < \frac{2V\, \|\nabla\mathcal{L}\|^2}{L(N + V + \kappa - 2)\, \mathbb{E}_m\!\bigl[\operatorname{Var}_{\boldsymbol{v}}\widetilde{\nabla}_{\boldsymbol{v}}\mathcal{L}^M\bigr]},
\end{equation}
the forward-gradient analogue of the CABS~\cite{balles_coupling_2017} and iCANS~\cite{kubler_adaptive_2020} learning-rate criteria.

\subsection{Optimal per-direction shot allocation}

Allowing the number of shots to depend on the direction, $M \to M_\ell$, the per-direction gain $\gamma_{\boldsymbol{v}^\ell}$ from~\eqref{eqn:gain_per_direction_app} becomes a function of $M_\ell$ alone. Maximising the gain-per-shot $\gamma_{\boldsymbol{v}^\ell}/M_\ell$ over $M_\ell$ and rearranging yields the optimal per-direction allocation referenced from~\secref{ssec:gain_forward_gradients}.

\begin{lemma}[Optimal per-direction shot allocation]\label{lem:per_direction_shots_app}
Let $\mathcal{L}$ have $L$-Lipschitz gradients and let $\{\boldsymbol{v}^\ell\}_{\ell=1}^V$ have zero mean, unit variance, and kurtosis $\kappa$. Then the per-direction shot count that maximises the per-shot gain is
\begin{equation}\label{eqn:optimal_M_ell_app}
    M_\ell^\star = \frac{2\, \operatorname{Var}_m\left[\widetilde{\nabla}_{\boldsymbol{v}^\ell}\mathcal{L}_m\right]}{\dfrac{2V}{L\eta(N + V + \kappa - 2)}\, \|\nabla\mathcal{L}\|^2  -  \bigl(\nabla_{\boldsymbol{v}^\ell}\mathcal{L}\bigr)^2}.
\end{equation}
\end{lemma}

\begin{proof}
Write $\sigma^2_{\nabla,\ell} := \operatorname{Var}_m[\widetilde{\nabla}_{\boldsymbol{v}^\ell}\mathcal{L}_m]$, $g_\ell := (\nabla_{\boldsymbol{v}^\ell}\mathcal{L})^2$, and $c := \nicefrac{L\eta^2}{2}\cdot\nicefrac{N+V+\kappa-2}{V}$. Then the per-shot gain is
Since $\gamma_{\boldsymbol{v}^\ell}$ contains a $1/M_\ell$ term from the noise, the gain-per-shot is
\begin{equation*}
    \frac{\gamma_{\boldsymbol{v}^\ell}}{M_\ell}
    = \frac{1}{M_\ell}\Bigl[\eta\|\nabla\mathcal{L}\|^2 - c\,g_\ell\Bigr] - \frac{c\,\sigma^2_{\nabla,\ell}}{M_\ell^2}.
\end{equation*}
Differentiating with respect to $M_\ell$ (treating $\|\nabla\mathcal{L}\|^2$, $g_\ell$, $\sigma^2_{\nabla,\ell}$ as $M_\ell$-independent):
\begin{equation*}
    \frac{d}{dM_\ell}\frac{\gamma_{\boldsymbol{v}^\ell}}{M_\ell}
    = -\frac{\eta\|\nabla\mathcal{L}\|^2 - c\,g_\ell}{M_\ell^2} + \frac{2c\,\sigma^2_{\nabla,\ell}}{M_\ell^3} = 0.
\end{equation*}
Multiplying through by $M_\ell^3 > 0$:
\begin{equation*}
    2c\,\sigma^2_{\nabla,\ell} = M_\ell\bigl(\eta\|\nabla\mathcal{L}\|^2 - c\,g_\ell\bigr).
\end{equation*}
Solving for $M_\ell$ and re-expanding $c = \nicefrac{L\eta^2}{2}\cdot\nicefrac{N+V+\kappa-2}{V}$:
\begin{equation*}
    M_\ell^\star = \frac{2c\,\sigma^2_{\nabla,\ell}}{\eta\|\nabla\mathcal{L}\|^2 - c\,g_\ell}
    = \frac{2\,\sigma^2_{\nabla,\ell}}{\dfrac{\eta\|\nabla\mathcal{L}\|^2}{c} - g_\ell}
    = \frac{2\,\operatorname{Var}_m[\widetilde{\nabla}_{\boldsymbol{v}^\ell}\mathcal{L}_m]}{\dfrac{2V}{L\eta(N+V+\kappa-2)}\|\nabla\mathcal{L}\|^2 - (\nabla_{\boldsymbol{v}^\ell}\mathcal{L})^2},
\end{equation*}
which is~\eqref{eqn:optimal_M_ell_app}.
\end{proof}

The denominator of~\eqref{eqn:optimal_M_ell_app} can be negative when the directional derivative is well-aligned with the true gradient; the regime analysis of when this happens, and the practical mitigation, is the subject of \appref{app:meas_alloc}.

\subsection{Fixed-\texorpdfstring{$M$}{M} optimal \texorpdfstring{$V$}{V}}\label{app:fixed_m_v_star_proof}

Under \assumpref{assump:noise_concentration_main} the per-direction measurement variance concentrates, $\operatorname{Var}_m[\widetilde{\nabla}_{\boldsymbol{v}^\ell}\mathcal{L}_m] \approx \bar\sigma^2_\nabla$ for all $\ell$. For isotropic zero-mean unit-variance directions, $\mathbb{E}_{\boldsymbol{v}}[(\nabla_{\boldsymbol{v}^\ell}\mathcal{L})^2] = \|\boldsymbol{\nabla}\mathcal{L}\|^2$. Taking this expectation in the per-direction gain~\eqref{eqn:gain_per_direction_app} and averaging over the $V$ directions:
\begin{equation}\label{eqn:expected_gain_averaged_app}
    \mathbb{E}[\mathcal{G}^{\mathsf F}]
    \approx
    \eta\,\|\boldsymbol{\nabla}\mathcal{L}\|^2
     - 
    \frac{L\eta^2}{2}\,\frac{N+V+\kappa-2}{V}
    \Bigl(\|\boldsymbol{\nabla}\mathcal{L}\|^2 + \tfrac{\bar\sigma^2_\nabla}{M}\Bigr).
\end{equation}
With $M$ fixed, the gain-per-shot objective is $f(V) := \mathbb{E}[\mathcal{G}^{\mathsf F}]/(2VM)$. Define the shorthand
\begin{equation*}
    A := \eta\|\boldsymbol{\nabla}\mathcal{L}\|^2, \qquad B := \tfrac{L\eta^2}{2}\bigl(\|\boldsymbol{\nabla}\mathcal{L}\|^2 + \bar\sigma^2_\nabla/M\bigr),
\end{equation*}
so that~\eqref{eqn:expected_gain_averaged_app} reads $\mathbb{E}[\mathcal{G}^{\mathsf F}] \approx A - B(N+V+\kappa-2)/V = (A-B) - B(N+\kappa-2)/V$. Dividing by $2VM$:
\begin{equation*}
    f(V) = \frac{A - B}{2VM} - \frac{B(N+\kappa-2)}{2V^2 M}.
\end{equation*}
Differentiating with respect to $V$:
\begin{equation*}
    f'(V) = -\frac{A-B}{2V^2 M} + \frac{B(N+\kappa-2)}{V^3 M}.
\end{equation*}
Setting $f'(V^\star) = 0$ gives $(A-B)/(2(V^\star)^2 M) = B(N+\kappa-2)/(V^\star)^3 M$, so $V^\star(A-B) = 2B(N+\kappa-2)$ and
\begin{equation*}
    V^\star = \frac{2B(N+\kappa-2)}{A-B}.
\end{equation*}
Re-substituting $A$ and $B$ and cancelling $L\eta^2/2$ from numerator and denominator:
\begin{equation*}
    V^\star = \frac{L\eta^2(\|\boldsymbol{\nabla}\mathcal{L}\|^2 + \bar\sigma^2_\nabla/M)(N+\kappa-2)}
    {\eta\|\boldsymbol{\nabla}\mathcal{L}\|^2 - \frac{L\eta^2}{2}(\|\boldsymbol{\nabla}\mathcal{L}\|^2 + \bar\sigma^2_\nabla/M)},
\end{equation*}
which after dividing numerator and denominator by $\eta$ yields~\eqref{eqn:adaptive_v_rule}. This is a maximum (not minimum) since $f(V) \to -\infty$ as $V \to 0^+$ and $f(V) \to 0^-$ as $V \to \infty$. The denominator $A - B > 0$ is precisely the learning-rate criterion~\eqref{eqn:lr_criterion_app}; when it fails the gain-per-shot is decreasing in $V$ for all $V > 0$, and $V^\star$ is set to $V_{\max} = N$.

\section{Regime analysis of per-direction shot allocation}\label{app:meas_alloc}

The per-direction allocation $M_\ell^\star$ of \lemref{lem:per_direction_shots_app} is well-defined only when its denominator is positive. We analyse when this fails and what it means operationally.

\subparagraph{Alignment decomposition.} Decompose an arbitrary direction as $\boldsymbol{v}^\ell = \|\boldsymbol{v}^\ell\|\, \widehat{\boldsymbol{v}}^\ell$ with $\|\widehat{\boldsymbol{v}}^\ell\| = 1$, and let $\phi_\ell$ be the angle between $\nabla\mathcal{L}$ and $\widehat{\boldsymbol{v}}^\ell$. Then
\begin{equation}\label{eqn:dir_deriv_squared_app}
    \bigl(\nabla_{\boldsymbol{v}^\ell}\mathcal{L}\bigr)^2 = \bigl(\nabla\mathcal{L} \cdot \boldsymbol{v}^\ell\bigr)^2 = \|\nabla\mathcal{L}\|^2\, \|\boldsymbol{v}^\ell\|^2\, \cos^2\phi_\ell.
\end{equation}
Defining the dimensionless constant $C := 2V/(L\eta(N + V + \kappa - 2))$, the denominator of~\eqref{eqn:optimal_M_ell_app} becomes
\begin{equation}\label{eqn:denominator_app}
    D_\ell = \|\nabla\mathcal{L}\|^2\, \Bigl(C - \|\boldsymbol{v}^\ell\|^2 \cos^2\phi_\ell\Bigr).
\end{equation}
The sign of $D_\ell$ depends on the product $\|\boldsymbol{v}^\ell\|^2 \cos^2\phi_\ell$: even at moderate misalignment, a large vector norm $\|\boldsymbol{v}^\ell\|$ can drive $D_\ell < 0$ purely from scale. Restricting to unit-norm directions ($\|\widehat{\boldsymbol{v}}^\ell\| = 1$) decouples scale from alignment,
\begin{equation}\label{eqn:denominator_unit}
    D_\ell = \|\nabla\mathcal{L}\|^2\, (C - \cos^2\phi_\ell),
\end{equation}
so the sign depends purely on alignment and on $(V, L, \eta, N, \kappa)$.

\subparagraph{Three regimes.} For unit-norm directions:
\begin{itemize}
    \item \textbf{$C > 1$:} $D_\ell > 0$ for every $\boldsymbol{v}^\ell$, so $M_\ell^\star > 0$ and grows with alignment ($\partial M_\ell^\star / \partial \cos^2\phi_\ell > 0$). Better-aligned directions receive more shots, as one would expect intuitively.
    \item \textbf{$C = 1$:} $D_\ell \geq 0$, with equality only at perfect alignment ($\cos^2\phi_\ell = 1$). Allocation diverges at exact alignment but is finite elsewhere.
    \item \textbf{$C < 1$:} cones around $\pm\nabla\mathcal{L}$ where $\cos^2\phi_\ell > C$ have $D_\ell < 0$, formally requiring negative shot counts.
\end{itemize}
The third regime occurs when an aggressive choice of $\eta$, $L$, or small $V$ drives the per-direction gain-per-shot negative, signalling that the squared projection $(\nabla\mathcal{L} \cdot \boldsymbol{v}^\ell)^2$ already dominates the gain expression and additional shots in that direction provide no benefit. The remedy is to push $C \geq 1$ (e.g.\ by tightening the learning-rate criterion~\eqref{eqn:lr_criterion_app}) or to enforce a minimum allocation $M_\ell \geq M_{\min}$. For a representative configuration ($\kappa = 3$, $\eta = 0.1$, $N = 100$, $V = 20$, $L = 1$) we have $C \approx 3.31$, comfortably in the first regime.

\subparagraph{Normalised random directions.} Working with normalised directions $\widehat{\boldsymbol{v}} = \boldsymbol{v}/\|\boldsymbol{v}\|$ rather than raw i.i.d.\ vectors changes the per-coordinate moments. Since $\sum_{i=1}^N \widehat{v}_i^2 \equiv 1$, the i.i.d.\ assumption gives
\begin{equation}\label{eqn:hat_moments_app}
    \mathbb{E}[\widehat{v}_i^2] = \frac{1}{N}, \qquad \operatorname{Var}(\widehat{v}_i) = \frac{1}{N} - \bigl(\mathbb{E}[\widehat{v}_i]\bigr)^2,
\end{equation}
and zero-mean for normalised coordinates is \emph{not} guaranteed by zero-mean of $\boldsymbol{v}$ alone. A sufficient extra condition is symmetry under $\boldsymbol{v} \mapsto -\boldsymbol{v}$, which makes $t \mapsto t/\sqrt{t^2 + S}$ an odd map (with $S = \sum_{k \neq i} v_k^2$) and yields $\mathbb{E}[\widehat{v}_i] = 0$. Under this assumption the second moment is $1/N$, so a normalised forward gradient estimator $\widehat{\boldsymbol{g}}^{\mathsf F} = N\,(\nabla f \cdot \widehat{\boldsymbol{v}})\,\widehat{\boldsymbol{v}}$ recovers unbiasedness with the rescaling factor $N$.

\section{Gradient-Lipschitz constant of the VQE loss}\label{app:vqe_lipschitz}

The gradient-Lipschitz constant $L$ of the VQE loss bounds the step size for convergence and enters the iCANS/gCANS shot-allocation multiplier. We derive it here for TFIM with $J = h = 1$ and open boundary conditions so that it is available to the noise-concentration proof that follows.

The VQE loss is $\mathcal{L}(\boldsymbol{\theta}) := \langle\psi(\boldsymbol{\theta})|H|\psi(\boldsymbol{\theta})\rangle$ with $|\psi(\boldsymbol{\theta})\rangle := U(\boldsymbol{\theta})|0\rangle^{\otimes n}$. The ansatz is hardware-efficient with each layer a product of single-qubit rotations $e^{-i\theta_j P_j/2}$ for $P_j \in \{X, Y\}$ followed by a fixed CZ ring. All generators satisfy $G_j^2 = I$ and $\|G_j\| = 1$. For two parameter vectors $\boldsymbol{\theta}_1, \boldsymbol{\theta}_2$,
\[
\|\nabla\mathcal{L}(\boldsymbol{\theta}_1) - \nabla\mathcal{L}(\boldsymbol{\theta}_2)\|
\leq
L \,\|\boldsymbol{\theta}_1 - \boldsymbol{\theta}_2\|,
\qquad
L \leq \|H\|,
\]
which follows from standard parameter-shift arguments: the $j$-th gradient component is $\partial_j\mathcal{L} = \nicefrac{i}{2}\langle\psi(\boldsymbol{\theta})|[U_{>j}^\dagger H U_{>j}, G_j]|\psi_{<j}\rangle$ with $|\psi_{<j}\rangle$ the state at parameter $j$, so $|\partial_i\partial_j\mathcal{L}| \leq \|[[H, G_i], G_j]\| \leq 4\|H\|\,\|G_i\|\,\|G_j\| = 4\|H\|$, and the operator-norm bound on the Hessian gives $L \leq \|H\|$ (with the factor of 4 absorbed by the half-angle rotation convention; see~\cite{sweke_stochastic_2020, kubler_adaptive_2020} for the same convention).

For TFIM,
\[
H_{\mathrm{TFIM}} = -J\sum_{i=1}^{n-1} Z_i Z_{i+1} - h\sum_{i=1}^n X_i,
\]
the operator norm is $\|H\| \leq (n-1)|J| + n|h|$ by the triangle inequality. At $n = 8$, $J = h = 1$, OBC, this gives $\|H\| \leq 15$, and so $L \leq 15$. The bound is tight within an $O(1)$ factor.

\section{Proof of the noise-concentration assumption for local Hamiltonians}\label{app:noise_concentration_proof}

We first restate the assumption from \secref{ssec:noise_concentration}, then prove it for the circuits and Hamiltonians used in this paper via the following lemma.

\begin{assumption}[\assumpref{assump:noise_concentration_main} restated]\label{assump:noise_concentration_app}
Let $\boldsymbol{v}^\ell$ have i.i.d.\ Rademacher components and let $\sigma^2_{\nabla,\ell} := \operatorname{Var}_m[\widetilde{\nabla}_{\boldsymbol{v}^{\ell}}\mathcal{L}_m]$ denote the per-direction measurement variance, with mean $\bar{\sigma}^2_{\nabla} := \mathbb{E}_{\boldsymbol{v}}[\sigma^2_{\nabla,\ell}]$. The function $\sigma^2_{\nabla,\ell}(\boldsymbol{v}^\ell)$ is Lipschitz in $\boldsymbol{v}^\ell$ with Lipschitz constant $L_\sigma$ independent of $N$, so that $\sigma^2_{\nabla,\ell} \approx \bar\sigma^2_\nabla$ for all $\ell$ with fluctuations $\mathcal{O}(L_\sigma)$. This is verified empirically for the benchmarks of this paper in \figref{fig:noise_concentration}.
\end{assumption}

The following lemma gives a concrete, quantitative version of this condition.

\begin{lemma}[Noise concentration for local Hamiltonians]\label{lem:noise_concentration_local}
Let $H = \sum_{j=1}^J h_j$ be a $k$-local Hamiltonian with $\|h_j\| \leq h_{\max}$, and let $U(\boldsymbol{\theta})$ be a 1D brick-layer hardware-efficient ansatz of depth $d$ with $N$ parameters and Pauli generators $G_i$ with $\|G_i\|=1$. Let $\sigma^2_{\nabla,\ell}(\boldsymbol{v}) := \operatorname{Var}_m[\widetilde{\nabla}_{\boldsymbol{v}^\ell}\mathcal{L}_m]$ for Rademacher $\boldsymbol{v}^\ell \in \{-1,+1\}^N$, and let $\bar\sigma^2_\nabla := \mathbb{E}_{\boldsymbol{v}}[\sigma^2_{\nabla,\ell}]$. Then for every $\delta > 0$,
\begin{equation*}
    \Pr\bigl(|\sigma^2_{\nabla,\ell} - \bar\sigma^2_\nabla| > \delta\bigr)
    \leq 2\exp\left(-\frac{\delta^2}{2L_\sigma^2}\right),
    \qquad
    L_\sigma = \frac{\sqrt{\xi(k,d)}\,h_{\max}^2}{\varepsilon M},
\end{equation*}
where $\xi(k,d) = kd + d(d+1)$ is the parameter count in the backward light cone of a $k$-local term at depth $d$, and $L_\sigma$ is independent of $N$ at fixed $(k,d)$. In particular, \assumpref{assump:noise_concentration_main} holds with Lipschitz constant $L_\sigma$.
\end{lemma}

\subparagraph{Setting.} Throughout the proof, $\mathcal{O}$ is a Pauli operator entering one of the $h_j$ terms, $U(\boldsymbol{\theta}) = U_N\cdots U_1$ with $U_i = e^{-i\theta_i G_i/2}$, and $\langle\mathcal{O}\rangle(\boldsymbol{\theta}) := \langle\psi(\boldsymbol{\theta})|\mathcal{O}|\psi(\boldsymbol{\theta})\rangle$.

\subparagraph{Variance formula for $\sigma^2_{\nabla,\ell}$.} Write $\sigma^2_\pm(\boldsymbol{v}) := \operatorname{Var}_m[o_m(\boldsymbol{\theta} \pm \varepsilon \boldsymbol{v})]$ for the per-shot variance at the two shifted parameter vectors. The variance of the central-difference estimator is then
\begin{equation}\label{eqn:variance_formula_app}
    \sigma^2_{\nabla,\ell} := \operatorname{Var}\left[\frac{f^M(\boldsymbol{\theta}+\varepsilon\boldsymbol{v}) - f^M(\boldsymbol{\theta}-\varepsilon\boldsymbol{v})}{2\varepsilon}\right]
    = \frac{\sigma^2_+(\boldsymbol{v}) + \sigma^2_-(\boldsymbol{v})}{4\varepsilon^2 M}.
\end{equation}
Each measurement of $\mathcal{O}$ returns an eigenvalue, which by definition lies in $[\lambda_{\min}, \lambda_{\max}]$. Popoviciu's inequality ($\operatorname{Var}[X] \leq (b-a)^2/4$ for $X \in [a,b]$) then gives $\sigma^2_\pm \leq \nicefrac{1}{4}(\lambda_{\max} - \lambda_{\min})^2$, so
\begin{equation}\label{eqn:sigma_bound_app}
    \sigma^2_{\nabla,\ell}(\boldsymbol{v})
    = \frac{\sigma^2_+(\boldsymbol{v}) + \sigma^2_-(\boldsymbol{v})}{4\varepsilon^2 M}
    \leq \frac{(\lambda_{\max} - \lambda_{\min})^2}{8\varepsilon^2 M}.
\end{equation}
\begin{proof}[Proof of \lemref{lem:noise_concentration_local}]
Since $\boldsymbol{v}^\ell$ has $N$ i.i.d.\ bounded ($\pm 1$) components and $\sigma^2_{\nabla,\ell}$ is a smooth function of them, the bound follows from Talagrand's inequality once we establish that $\sigma^2_{\nabla,\ell}$ is Lipschitz with an $N$-independent constant.

\begin{lemma}[Talagrand's concentration inequality for Lipschitz functions, \cite{talagrand_concentration_1995}]\label{lem:talagrand}
Let $X_1, \dots, X_N$ be independent random variables with $|X_i| \leq 1$, and let $f : \mathbb{R}^N \to \mathbb{R}$ be $L$-Lipschitz with respect to the Euclidean norm. Then for any $t > 0$,
\begin{equation*}
    \Pr\bigl(|f(X_1,\dots,X_N) - \mathbb{E}[f]| > t\bigr) \leq 2\exp\left(-\frac{t^2}{2L^2}\right).
\end{equation*}
\end{lemma}

We apply this to $f = \sigma^2_{\nabla,\ell}$ as a function of the $N$ i.i.d.\ $\pm 1$ Rademacher components of $\boldsymbol{v}^\ell$. The lemma requires $\sigma^2_{\nabla,\ell}$ to be Lipschitz; the Lipschitz constant of a differentiable function is its gradient-norm supremum, so we need an $N$-independent bound on $\|\nabla_{\boldsymbol{v}}\sigma^2_{\nabla,\ell}\|$, which we derive in three steps.

\begin{enumerate}
    \item \emph{Light-cone bound~\cite{cerezo_cost_2021}.} Write $U(\boldsymbol{\theta}) = U_N U_{N-1} \cdots U_1$ where each $U_i = e^{-i\theta_i G_i/2}$ is a single-qubit rotation. The expectation value decomposes as
    \begin{equation*}
        \langle\mathcal{O}\rangle(\boldsymbol{\theta})
        = \langle 0|^{\otimes n} U^\dagger(\boldsymbol{\theta})\,\mathcal{O}\,U(\boldsymbol{\theta})|0\rangle^{\otimes n}.
    \end{equation*}
    Gate $U_i$ acts on qubit $q_i$. By the product rule, $\partial\langle\mathcal{O}\rangle/\partial\theta_i$ involves commuting $G_i$ through the subsequent gates to the measurement. If qubit $q_i$ is never in the support of any gate between layer $i$ and the support of $\mathcal{O}$, then $G_i$ acts as the identity on the relevant subspace and $\partial\langle\mathcal{O}\rangle/\partial\theta_i = 0$ exactly. Following Cerezo et al.~\cite{cerezo_cost_2021}, the \emph{backward light cone} $\mathcal{L}_B(\mathcal{O})$ is the set of gates whose output qubits are causally connected to the support of $\mathcal{O}$; all parameters outside $\mathcal{L}_B(\mathcal{O})$ have zero gradient.

    For a $k$-local term $h_j$ in a 1D brick-layer ansatz of depth $d$, each layer extends $\mathcal{L}_B(h_j)$ by one qubit in each direction, so at layer $l$ (counted back from the measurement) the cone has width $k + 2l$. Summing over all $d$ layers gives the total parameter count:
    \begin{equation}\label{eqn:xi_count_app}
        \xi(k, d) = \sum_{l=1}^{d} (k + 2l) = kd + d(d+1) = \mathcal{O}\!\bigl((k+d)\,d\bigr),
    \end{equation}
    which is independent of $N$ at fixed $(k, d)$.

    \item \emph{Chain rule for $\nabla_{\boldsymbol{v}}\sigma^2_\pm$.} Since $\sigma^2_\pm(\boldsymbol{v}) = \operatorname{Var}_m[o_m(\boldsymbol{\theta} \pm \varepsilon\boldsymbol{v})]$ and $\operatorname{Var}_m[o_m(\boldsymbol{\theta}')] = \langle\mathcal{O}^2\rangle(\boldsymbol{\theta}') - \langle\mathcal{O}\rangle^2(\boldsymbol{\theta}')$, differentiating in $v_i$ gives
    \begin{equation}\label{eqn:chain_rule_app}
        \frac{\partial \sigma^2_\pm}{\partial v_i}
        = \pm\varepsilon\,\frac{\partial}{\partial\theta_i}\bigl[\langle\mathcal{O}^2\rangle - \langle\mathcal{O}\rangle^2\bigr]_{\boldsymbol{\theta}\pm\varepsilon\boldsymbol{v}}
        = \mp 2\varepsilon\,\langle\mathcal{O}\rangle(\boldsymbol{\theta}\pm\varepsilon\boldsymbol{v})\,\partial_{\theta_i}\langle\mathcal{O}\rangle(\boldsymbol{\theta}\pm\varepsilon\boldsymbol{v}).
    \end{equation}
    By the light-cone bound of step~(i), $\partial_{v_i}\sigma^2_\pm = 0$ for all $\theta_i \notin \mathcal{L}_B(h_j)$, leaving at most $\xi(k,d)$~\eqref{eqn:xi_count_app} non-zero components. For the in-cone indices, we bound each factor separately:
    \begin{itemize}
        \item $|\partial_{\theta_i}\langle\mathcal{O}\rangle| \leq \|G_i\|\,\|\mathcal{O}\| \leq h_{\max}$, by the parameter-shift Lipschitz bound of \appref{app:vqe_lipschitz} with $\|G_i\| = 1$; here $h_{\max} := \max_j \|h_j\|$ is the maximum local-term operator norm defined in the Setting above;
        \item $|\langle\mathcal{O}\rangle(\boldsymbol{\theta}')| \leq \|\mathcal{O}\| \leq h_{\max}$ for any $\boldsymbol{\theta}'$.
    \end{itemize}
    Hence $|\partial_{v_i}\sigma^2_\pm| \leq 2\varepsilon \cdot h_{\max} \cdot h_{\max} = 2\varepsilon h_{\max}^2$.

    \item \emph{Bounding $\|\nabla_{\boldsymbol{v}}\sigma^2_{\nabla,\ell}\|$.} Differentiating~\eqref{eqn:variance_formula_app} with respect to $v_i$ gives
    \begin{equation}\label{eqn:sigma_grad_app}
        \partial_{v_i}\sigma^2_{\nabla,\ell} = \frac{\partial_{v_i}\sigma^2_+(\boldsymbol{v}) + \partial_{v_i}\sigma^2_-(\boldsymbol{v})}{4\varepsilon^2 M}.
    \end{equation}
    Step~(ii) showed $|\partial_{v_i}\sigma^2_\pm| \leq 2\varepsilon h_{\max}^2$ via~\eqref{eqn:chain_rule_app}. Substituting into~\eqref{eqn:sigma_grad_app}:
    \begin{equation*}
        |\partial_{v_i}\sigma^2_{\nabla,\ell}| \leq \frac{2\varepsilon h_{\max}^2 + 2\varepsilon h_{\max}^2}{4\varepsilon^2 M} = \frac{h_{\max}^2}{\varepsilon M}.
    \end{equation*}
    Crucially, each $v_i$ appears in the light cone of at most $\mathcal{O}(1)$ Hamiltonian terms by $k$-locality, so $\nabla_{\boldsymbol{v}}\sigma^2_{\nabla,\ell}$ has at most $\mathcal{O}(\xi(k,d))$ non-zero components independently of $J$ (the number of Hamiltonian terms) and hence of $N$. Summing the squared components over the non-zero indices:
    \begin{equation*}
        \|\nabla_{\boldsymbol{v}}\sigma^2_{\nabla,\ell}\|^2
        =\sum_{i\,:\,\theta_i\in\mathcal{L}_B}|\partial_{v_i}\sigma^2_{\nabla,\ell}|^2
        \leq\xi(k,d)\cdot\left(\frac{h_{\max}^2}{\varepsilon M}\right)^{\!2}.
    \end{equation*}
    Taking the square root defines
    \begin{equation}\label{eqn:lipschitz_V_ell_app}
        L_\sigma := \|\nabla_{\boldsymbol{v}} \sigma^2_{\nabla,\ell}\| \leq \frac{\sqrt{\xi(k,d)}\,h_{\max}^2}{\varepsilon\, M} = \mathcal{O}\left(\frac{\sqrt{(k+d)d}\,h_{\max}^2}{\varepsilon M}\right),
    \end{equation}
    which is independent of $N$ at fixed $(k,d)$.
    \item \emph{Concentration via \lemref{lem:talagrand}.} With $L_\sigma$ from~\eqref{eqn:lipschitz_V_ell_app} established, $\sigma^2_{\nabla,\ell}(\boldsymbol{v}^\ell)$ is $L_\sigma$-Lipschitz in the $N$ i.i.d.\ $\pm 1$ Rademacher components of $\boldsymbol{v}^\ell$, and $\mathbb{E}_{\boldsymbol{v}}[\sigma^2_{\nabla,\ell}] = \bar\sigma^2_\nabla$ by definition. Applying \lemref{lem:talagrand} with $f = \sigma^2_{\nabla,\ell}$ and $L = L_\sigma$: for every $\delta > 0$,
    \begin{equation}\label{eqn:talagrand_app}
        \Pr\bigl(|\sigma^2_{\nabla,\ell} - \bar\sigma^2_\nabla| > \delta\bigr) \leq 2\exp\left(-\frac{\delta^2}{2 L_\sigma^2}\right).
    \end{equation}
    Equivalently, for any $\eta \in (0,1)$, choosing $\delta = L_\sigma\sqrt{2\log(2/\eta)}$:
    \begin{equation}\label{eqn:concentration_final_app}
        \Pr\bigl(|\sigma^2_{\nabla,\ell} - \bar\sigma^2_\nabla| \leq L_\sigma\sqrt{2\log(2/\eta)}\bigr) \geq 1 - \eta.
    \end{equation}
    So $\sigma^2_{\nabla,\ell}$ lies within $\mathcal{O}(L_\sigma)$ of $\bar\sigma^2_\nabla$ with probability at least $1-\eta$, where $L_\sigma = \sqrt{\xi(k,d)}\,h_{\max}^2/(\varepsilon M)$ from~\eqref{eqn:lipschitz_V_ell_app} is independent of $N$.
\end{enumerate}
\end{proof}

\section{Proof of \propref{prop:rademacher_min}: Rademacher minimises estimator variance}\label{app:rademacher_proof}

\propref{prop:rademacher_min} claims that among all isotropic independent-component direction distributions, Rademacher uniquely minimises $\mathbb{E}\|\widehat{\boldsymbol{g}}_V - \boldsymbol{g}\|^2$. To prove this we first derive a closed-form expression for the MSE as a function of the component kurtosis $\kappa$, then minimise over $\kappa$. Since $\widehat{\boldsymbol{g}}_V$ is unbiased ($\mathbb{E}[\widehat{\boldsymbol{g}}_V] = \boldsymbol{g}$), the MSE equals the variance: $\mathbb{E}\|\widehat{\boldsymbol{g}}_V - \boldsymbol{g}\|^2 = \mathbb{E}\|\widehat{\boldsymbol{g}}_V\|^2 - \|\boldsymbol{g}\|^2$.

\subparagraph{Diagonal and off-diagonal decomposition.} By i.i.d.\ sampling of the $V$ directions,
\begin{equation}
    \mathbb{E}\|\widehat{\boldsymbol{g}}_V\|^2
    =
    \frac{1}{V^2}\,\mathbb{E}\left[\sum_\ell (\boldsymbol{g}\cdot\boldsymbol{v}^\ell)^2 \|\boldsymbol{v}^\ell\|^2\right]
     + 
    \frac{1}{V^2}\,\mathbb{E}\left[\sum_{\ell\ne\ell'} (\boldsymbol{g}\cdot\boldsymbol{v}^\ell)(\boldsymbol{g}\cdot\boldsymbol{v}^{\ell'})\,(\boldsymbol{v}^\ell\cdot\boldsymbol{v}^{\ell'})\right].
\end{equation}
For the off-diagonal terms ($\ell \ne \ell'$), independence of the directions and isotropy ($\mathbb{E}[\boldsymbol{v}\boldsymbol{v}^\top] = I_N$) give $\mathbb{E}[(\boldsymbol{g}\cdot\boldsymbol{v}^\ell)(\boldsymbol{g}\cdot\boldsymbol{v}^{\ell'})(\boldsymbol{v}^\ell\cdot\boldsymbol{v}^{\ell'})] = (\mathbb{E}[(\boldsymbol{g}\cdot\boldsymbol{v})\boldsymbol{v}])^2 = \|\boldsymbol{g}\|^2$ component-wise, so the off-diagonal sum contributes $V(V-1)\|\boldsymbol{g}\|^2/V^2 = (V-1)\|\boldsymbol{g}\|^2/V$.

\subparagraph{Diagonal term and kurtosis.} For the diagonal terms, expand $\|\boldsymbol{v}\|^2 = \sum_j v_j^2$ and use component independence:
\begin{align}
    \mathbb{E}[(\boldsymbol{g}\cdot\boldsymbol{v})^2 \|\boldsymbol{v}\|^2]
    &= \sum_{i,j} g_i g_j\, \mathbb{E}[v_i v_j \|\boldsymbol{v}\|^2] \notag \\
    &= \sum_{i,j} g_i^2\, \mathbb{E}[v_i^2 v_j^2] \notag \\
    &= \sum_i g_i^2\bigl(\mathbb{E}[v_i^4] + \sum_{j \ne i}\mathbb{E}[v_i^2]\mathbb{E}[v_j^2]\bigr) \notag \\
    &= \sum_i g_i^2\,(\kappa + N - 1)
    = \|\boldsymbol{g}\|^2\,(\kappa + N - 1),
\end{align}
where the cross terms $i \ne j$ in the second line vanish because $\mathbb{E}[v_i v_j] = 0$ for $i \ne j$ (isotropy), and we used $\mathbb{E}[v_i^2] = 1$ and $\mathbb{E}[v_i^4] = \kappa$ for $i = j$. The diagonal contribution to $\mathbb{E}\|\widehat{\boldsymbol{g}}_V\|^2$ is therefore $\|\boldsymbol{g}\|^2(\kappa + N - 1)/V$.

\subparagraph{Combining.} Adding diagonal and off-diagonal contributions and subtracting $\|\boldsymbol{g}\|^2$:
\begin{align}
    \mathbb{E}\|\widehat{\boldsymbol{g}}_V - \boldsymbol{g}\|^2
    &= \frac{\|\boldsymbol{g}\|^2(\kappa + N - 1)}{V} + \frac{(V-1)\|\boldsymbol{g}\|^2}{V} - \|\boldsymbol{g}\|^2 \notag \\
    &= \frac{\|\boldsymbol{g}\|^2(\kappa + N - 2)}{V},
\end{align}
which is~\eqref{eqn:var_direction_formula}.

\subparagraph{Minimisation over the direction distribution.} Since $\kappa = \mathbb{E}[v_i^4]$ and $\mathbb{E}[v_i^2] = 1$, Jensen's inequality gives $\kappa = \mathbb{E}[v_i^4] \ge (\mathbb{E}[v_i^2])^2 = 1$, with equality iff $v_i^2$ is almost surely constant. Combined with $\mathbb{E}[v_i^2] = 1$, this forces $|v_i| = 1$ a.s., and combined with $\mathbb{E}[v_i] = 0$ it forces $v_i \in \{+1, -1\}$ with equal probability (the Rademacher distribution). This uniquely minimises~\eqref{eqn:var_direction_formula} over all isotropic independent-component distributions with unit variance, giving the bound $\mathbb{E}\|\widehat{\boldsymbol{g}}_V - \boldsymbol{g}\|^2 \ge \|\boldsymbol{g}\|^2(N-1)/V$.

\subsection{Proof of \thmref{thm:joint_optimum}: optimal allocation}\label{app:joint_optimum_proof}

\subparagraph{Setup.} The MSE of the Rademacher forward gradient estimator with $V$ directions and $M$ shots per direction (\eqref{eqn:total_var}) is
\begin{equation}\label{eqn:mse_proof}
    \operatorname{MSE}(V, M) = \frac{(N-1)\|\boldsymbol{g}\|^2}{V} + \frac{N\bar\sigma^2_\nabla}{VM},
\end{equation}
and the cost-minimisation problem is
\begin{equation}\label{eqn:cost_problem_proof}
    \min_{V,\,M>0} \; 2VM \qquad \text{s.t.} \quad \operatorname{MSE}(V,M) \le \tau^2, \quad M \ge M_{\min}.
\end{equation}
The proof has five steps.

\begin{proof}
\begin{enumerate}
    \item \emph{Eliminate $V$.} The MSE constraint in~\eqref{eqn:cost_problem_proof} is equivalent to
    \begin{equation}\label{eqn:phi_def_proof}
        V \ge \phi(M) := \frac{(N-1)\|\boldsymbol{g}\|^2 + N\bar\sigma^2_\nabla/M}{\tau^2}.
    \end{equation}
    Since $2VM$ is strictly increasing in $V$ at fixed $M > 0$, the infimum is attained at $V = \phi(M)$. Substituting reduces~\eqref{eqn:cost_problem_proof} to
    \begin{equation}\label{eqn:reduced_problem}
        \min_{M \ge M_{\min}}\, 2M\phi(M)
        = \min_{M \ge M_{\min}}\, \frac{2\bigl[(N-1)\|\boldsymbol{g}\|^2\, M + N\bar\sigma^2_\nabla\bigr]}{\tau^2}.
    \end{equation}

    \item \emph{Minimise over $M$.} The reduced objective~\eqref{eqn:reduced_problem} is affine in $M$ with slope $2(N-1)\|\boldsymbol{g}\|^2/\tau^2 > 0$ ($N \ge 2$ and $\|\boldsymbol{g}\|^2 > 0$ by assumption), so its minimum on $[M_{\min}, \infty)$ is
    \begin{equation}\label{eqn:mstar_proof}
        M^\star = M_{\min}.
    \end{equation}

    \item \emph{Recover $V^\star$ and $C^\star$.} Substituting~\eqref{eqn:mstar_proof} into~\eqref{eqn:phi_def_proof} gives
    \begin{equation}\label{eqn:vstar_proof}
        V^\star = \phi(M_{\min}) = \frac{(N-1)\|\boldsymbol{g}\|^2 + N\bar\sigma^2_\nabla/M_{\min}}{\tau^2},
    \end{equation}
    which is~\eqref{eqn:joint_optimum_full}, and the minimum cost is
    \begin{equation}\label{eqn:cstar_proof}
        C^\star = 2V^\star M_{\min} = \frac{2\bigl[(N-1)\|\boldsymbol{g}\|^2\,M_{\min} + N\bar\sigma^2_\nabla\bigr]}{\tau^2},
    \end{equation}
    which is~\eqref{eqn:joint_cost_minimum}.

    \item \emph{Uniqueness.} Both constraints in~\eqref{eqn:cost_problem_proof} are tight at $(V^\star, M^\star)$. Any feasible $(V,M)$ with $M > M_{\min}$ has strictly larger cost (Step~2), and any with $V > \phi(M)$ at the same $M$ also has strictly larger cost (Step~1). Hence $(V^\star, M^\star)$ is the unique minimiser.

    \item \emph{Reparameterisation.} Set $\alpha := N\bar\sigma^2_\nabla/(M_{\min}\|\boldsymbol{g}\|^2) > 0$, so $M_{\min} = N\bar\sigma^2_\nabla/(\alpha\|\boldsymbol{g}\|^2)$. Substituting into~\eqref{eqn:vstar_proof}:
    \begin{equation}\label{eqn:vstar_alpha_proof}
        V^\star
        = \frac{(N-1)\|\boldsymbol{g}\|^2 + \alpha\|\boldsymbol{g}\|^2}{\tau^2}
        = \frac{(N-1+\alpha)\|\boldsymbol{g}\|^2}{\tau^2},
    \end{equation}
    and $M^\star = N\bar\sigma^2_\nabla/(\alpha\|\boldsymbol{g}\|^2)$. These are~\eqref{eqn:joint_optimum_alpha}.
\end{enumerate}
\end{proof}

\subsection{Proof of \corref{cor:crb}: CRB-level optimality}\label{app:crb_proof}

The proof has two parts. First we establish the Cram\'er--Rao lower bound~\eqref{eqn:crb_bound} on the MSE of any unbiased estimator of $\boldsymbol{g}$ that queries the shot-noise oracle a total of $B$ times. Second we show that the forward-gradient estimator at the optimal allocation of \thmref{thm:joint_optimum} attains this bound up to a constant that vanishes as $N \to \infty$.

\subparagraph{Lower bound.}
We treat the oracle output at each query as a noisy linear measurement of the gradient. A query at shift $\boldsymbol{s}_k$ returns a single shot with mean $f(\boldsymbol{\theta}+\boldsymbol{s}_k) \approx f(\boldsymbol{\theta}) + \boldsymbol{s}_k^\top\boldsymbol{g}$ (local linearisation of $f$ around $\boldsymbol{\theta}$, valid for the small shifts $\|\boldsymbol{s}_k\| = \mathcal{O}(\varepsilon)$ we use) and variance $\sigma_m^2$. Since the mean depends linearly on $\boldsymbol{g}$ and the variance is independent of $\boldsymbol{g}$, the Fisher information of this query for the gradient parameter $\boldsymbol{g}$ is
\[
    \mathcal{I}_k = \frac{\boldsymbol{s}_k\boldsymbol{s}_k^\top}{\sigma_m^2}.
\]

After $B$ i.i.d.\ queries the total Fisher information is
\[
    \mathcal{I}_B = \frac{1}{\sigma_m^2}\sum_{k=1}^B \boldsymbol{s}_k\boldsymbol{s}_k^\top,
\]
and the multi-parameter Cram\'er--Rao inequality bounds the MSE of any unbiased estimator $\widehat{\boldsymbol{g}}$ by the trace of the inverse Fisher information matrix,
\[
    \mathbb{E}\|\widehat{\boldsymbol{g}} - \boldsymbol{g}\|^2 \ge \operatorname{tr}(\mathcal{I}_B^{-1}).
\]
For Rademacher shifts $\boldsymbol{s}_k \in \{\pm\varepsilon\}^N$ drawn i.i.d.\ with $\mathbb{E}[\boldsymbol{s}_k\boldsymbol{s}_k^\top] = \varepsilon^2 I_N$, the expected Fisher information is the diagonal matrix $\mathbb{E}[\mathcal{I}_B] = (B\varepsilon^2/\sigma_m^2)\, I_N$, whose inverse trace is $N\sigma_m^2/(B\varepsilon^2)$. This is~\eqref{eqn:crb_bound}: $B_{\mathrm{CRB}} = N\sigma_m^2/(\varepsilon^2\tau^2)$ queries are required to drive the MSE to a target $\tau^2$.

\subparagraph{Achievability.}
\thmref{thm:joint_optimum} expresses the minimum cost $C^\star$ in terms of the per-direction shot variance $\bar\sigma^2_\nabla$. In the bias-neutral limit (small $\varepsilon$, finite-difference bias negligible), this variance reduces to
\[
    \bar\sigma^2_\nabla = \frac{\sigma_m^2}{2M\varepsilon^2}
\]
(the central-difference numerator has variance $2\sigma_m^2$ from two independent $M$-shot means, divided by $(2\varepsilon)^2$). Substituting this expression into~\eqref{eqn:joint_cost_minimum} gives the QUIVER shot count
\[
    B^\star = \frac{2(N-1+\alpha)\,N\,\bar\sigma^2_\nabla}{\alpha\,\tau^2}
    = \frac{(N-1+\alpha)\,N\,\sigma_m^2}{\varepsilon^2\,\alpha\,\tau^2}.
\]
Comparing to $B_{\mathrm{CRB}}$ above,
\[
    \frac{B^\star}{B_{\mathrm{CRB}}} = \frac{N-1+\alpha}{N} \to 1 \quad\text{as } N \to \infty,
\]
for any fixed hyperparameter $\alpha > 0$. \textsc{quiver} therefore saturates the Cram\'er--Rao lower bound up to a multiplicative factor that vanishes as $1/N$.

\section{iCANS/gCANS baselines and $N$-scaling of the VQE regime crossover}\label{app:icans_baselines}

iCANS~\cite{kubler_adaptive_2020} and gCANS~\cite{gu_adaptive_2021} are the closest parameter-shift-based shot-adaptive baselines to \textsc{quiver}: both adapt the per-parameter shot count via a gain-per-shot argument on the parameter-shift estimator. This appendix details the mechanism by which they collapse onto the floor $s_{\min}$ at the conservative learning rates required at large $N$, reports the hyperparameter sweeps used to set the iCANS/gCANS baselines in \figref{fig:regime_crossover} and \figref{fig:n_scaling}, and extends the regime-crossover comparison of \figref{fig:regime_crossover}b across a broader $N$ sweep.

\subparagraph{Mechanism.} The iCANS allocation for parameter $i$ is
\[
M^{\text{iCANS}}_i = \frac{2L\eta}{2-L\eta}\cdot\frac{\widehat{\sigma}^2_i}{\widehat{g}^2_i},
\]
where $\widehat\sigma^2_i, \widehat g^2_i$ are EMA estimates of the per-coordinate variance and squared gradient. For our VQE TFIM benchmark $L = 15$ (\appref{app:vqe_lipschitz}); at $\eta = 0.01$ the $L\eta$ prefactor is $2L\eta/(2-L\eta) \approx 0.16$, so to allocate above $s_{\min}$ a parameter must have $\widehat\sigma^2_i/\widehat g^2_i \gtrsim 12$. Across all $N \in \{80, 160, 240, 320\}$ we tested, this is rarely the case and the allocation is dominated by $s_{\min}$ for almost every parameter. gCANS uses a global-norm variant of the same rule (one shot count per step rather than one per parameter), but inherits the same $L\eta$ prefactor and converges to a similar final error. Only at small $N$ and aggressive $\eta$ (e.g.\ $\eta = 0.05$ at $N = 80$) does the prefactor exceed unity and both methods adapt above $s_{\min}$; in that regime training is unstable and the final error is worse.

\subparagraph{Sweep configurations.} For the $N$-scaling comparison of \figref{fig:n_scaling} we ran iCANS at each $N$ over four $(\eta, s_{\text{init}})$ pairs: $(0.003, 5)$, $(0.003, 10)$, $(0.01, 10)$, $(0.05, 50)$, with $s_{\min} = 2$ throughout. For the regime-crossover figure (\figref{fig:regime_crossover}) we additionally swept $s_{\text{init}} \in \{2, 5, 10\}$ at $(N=80, \eta=0.05)$ for panel (a); panel (b) at $(N=160, \eta=0.003)$ uses the single representative configuration $(0.003, 5)$. In both figures we plot the best-performing iCANS and gCANS configuration at each $N$.

\begin{figure*}[t]
    \centering
    \includegraphics[width=\textwidth]{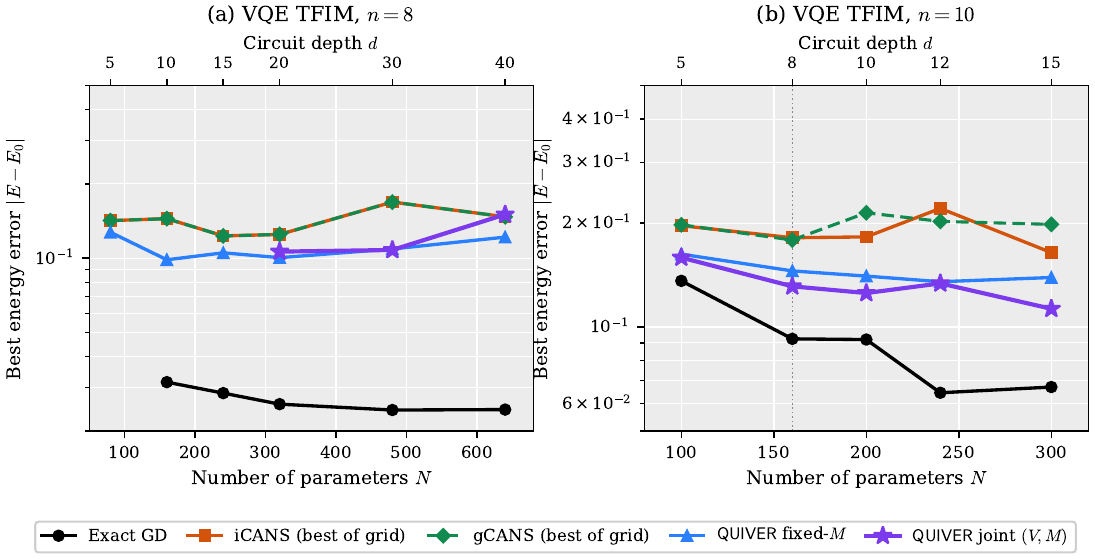}
    \caption{\textbf{\textsc{quiver} leads at all depths and budgets on the VQE $N$-scaling benchmark.} Best energy error as a function of parameter count $N$ for VQE TFIM at matched $5\times 10^6$-shot budget. (a) $n = 8$, $N \in [80, 480]$. (b) $n = 10$, $N \in [80, 300]$. For each method we plot the best error reached across a small hyperparameter grid at each $N$. Black: exact gradient descent. Orange: best iCANS. Green: best gCANS. Blue: best fixed-$M$ \textsc{quiver}. Purple: best \textsc{quiver} joint $(V, M)$. On (b), the dotted vertical line marks the smallest depth at which exact gradient descent meets the target energy error $\delta_E = 0.1$ on $n=10$ TFIM ($d = 8$, $N = 160$); the next-smaller depth $d = 5$ ($N = 100$) fails to reach the target.}
    \label{fig:n_scaling}
\end{figure*}

\section{Circuit descriptions}\label{app:circuits}

The three benchmark circuit families used in this work are described below.
Figure~\ref{fig:circuit_diagrams} shows representative $n = 4$ examples for each.

\begin{figure*}[t]
\centering
    \includegraphics[width=\textwidth]{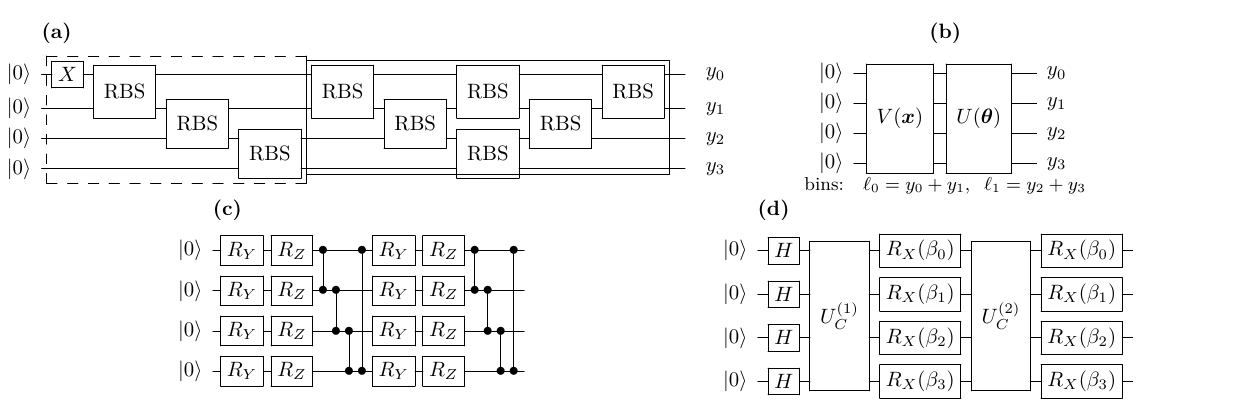}
\caption{\textbf{Circuit diagrams for the three benchmark families ($n{=}4$).}
(a)~OQNN~\cite{landman_quantum_2022, monbroussou_trainability_2023,
coyle_training-efficient_2024}: the \emph{loader} (dashed) encodes $\boldsymbol{x}$ into the unary state $|\phi(\boldsymbol{x})\rangle = \sum_i x_i|e_i\rangle$; the \emph{variational} pyramid (solid, $N{=}n(n{-}1)/2$ RBS gates) yields amplitudes $y_i = \langle e_i|\psi\rangle$.
(b)~Compact view ($C{=}2$): loader $V(\boldsymbol{x})$ and variational $U(\boldsymbol{\theta})$ produce $\{y_i\}$, which are partitioned into $C$ contiguous bins giving logits $\ell_c = \sum_{j\in\mathrm{bin}_c} y_j$. The original orthogonal-QNN proposal uses truncation ($\ell_c = y_{n-C+c}$); binning uses the full distribution and reduces variance at large $n$.
(c)~HEA for VQE~\cite{kandala_hardware-efficient_2017}: per-qubit $R_Y R_Z$ followed by CZ entanglers on adjacent pairs and the closing long-range $\mathrm{CZ}(0,n{-}1)$; $N{=}2nd$.
(d)~Multi-angle QAOA~\cite{farhi_quantum_2014, herrman_multiangle_2022}: $H^{\otimes n}$ initialises the uniform superposition; each layer applies cost block $U_C^{(l)}$ (one global ZZ angle and $n$ per-qubit $Z$ angles) and per-qubit $R_X(\beta_{j,l})$ mixers; $N{=}(2n{+}1)d$.}
\label{fig:circuit_diagrams}
\end{figure*}

\subsection*{Orthogonal quantum neural network (OQNN)}\label{app:oqnn_circuit}

The classification benchmarks (ECG5000 and MNIST) use a Hamming-weight-preserving orthogonal
quantum neural network following~\cite{landman_quantum_2022, monbroussou_trainability_2023,
coyle_training-efficient_2024}.
An $n$-qubit circuit acts entirely within the single-excitation (Hamming-weight-1) subspace,
spanned by the computational basis states $\{|e_i\rangle\}_{i=1}^n$ where $|e_i\rangle$ has
a single $1$ in position $i$.
Within this subspace the reconfigurable beam-splitter (RBS) gate
$\mathrm{RBS}(\theta)$ on qubits $(i,j)$ acts as a 2D Givens rotation,
\begin{equation*}
    \mathrm{RBS}(\theta)\,|e_i\rangle = \cos\theta\,|e_i\rangle + \sin\theta\,|e_j\rangle,
    \quad
    \mathrm{RBS}(\theta)\,|e_j\rangle = -\sin\theta\,|e_i\rangle + \cos\theta\,|e_j\rangle,
\end{equation*}
and the full circuit of $N$ RBS gates arranged in a triangular (staircase) layout implements
an arbitrary $n \times n$ special-orthogonal matrix $O(\boldsymbol{\theta}) \in SO(n)$ on the
amplitude vector, with $N = n(n-1)/2$ parameters.
Classical input $\boldsymbol{x} \in \mathbb{R}^n$ (unit-normalised) is encoded as
$|\phi(\boldsymbol{x})\rangle = \sum_i x_i |e_i\rangle$; the circuit maps this to
$|\psi(\boldsymbol{\theta})\rangle = \sum_i y_i(\boldsymbol{\theta}) |e_i\rangle$ where
$\boldsymbol{y} = O(\boldsymbol{\theta})\boldsymbol{x}$.
The output $n$-vector $\boldsymbol{y}$ is the logit vector after the readout step below.

\subparagraph{Readout: binning vs.\ truncation.}
The original proposal for orthogonal QNNs~\cite{landman_quantum_2022, monbroussou_trainability_2023, coyle_training-efficient_2024} reads class scores by \emph{truncation}: $\ell_c = y_{n-C+c}$, discarding the first $n-C$ components of $\boldsymbol{y}$. At the system sizes used here ($n$ up to 60, $C \in \{5, 10\}$) this discards up to 98\% of the output. We instead use a \emph{binning} readout: partition the $n$ outputs into $C$ contiguous bins of size $b = \lfloor n/C \rfloor$ (discarding at most $C-1$ remainder components) and sum within each bin,
\begin{equation*}
    \ell_c = \sum_{j = (c-1)b+1}^{\,cb} y_j, \quad c = 1,\ldots,C.
\end{equation*}
Binning uses the full output vector and improves convergence at large $n$ at no change in parameter count or gradient cost.

\subsection*{Hardware-efficient ansatz for VQE}\label{app:hea_circuit}

The VQE ground-state benchmarks use a hardware-efficient ansatz
(HEA)~\cite{kandala_hardware-efficient_2017} composed of $d$ identical layers, each acting on
$n$ qubits.
In layer $k$ each qubit $q \in \{0,\ldots,n-1\}$ receives a pair of single-qubit rotations
$R_Y(\theta_{q,k})\,R_Z(\phi_{q,k})$, followed by a closed CZ ring:
CZ gates on pairs $(0,1),(1,2),\ldots,(n-2,n-1)$ and a closing CZ$(0,n-1)$.
The initial state is $|0\rangle^{\otimes n}$.
Total parameter count is $N = 2nd$.
The CZ ring provides all-to-nearest-neighbour entanglement within each layer at no additional
parameter cost; for the 1D Ising Hamiltonian used here the ring topology matches the
interaction structure of the Hamiltonian.

\subsection*{Multi-angle QAOA}\label{app:qaoa_circuit}

The combinatorial optimisation benchmarks use a multi-angle QAOA (MA-QAOA)~\cite{farhi_quantum_2014,
herrman_multiangle_2022} variant with $d$ layers.
Standard QAOA uses a single global cost angle $\gamma_l$ and a single global mixer angle $\beta_l$
per layer $l$.
MA-QAOA assigns independent parameters to each qubit and, in the cost layer, additionally
includes per-qubit local-field terms, yielding greater expressivity at shallower depth than the standard QAOA ansatz:
\begin{align*}
    U_C^{(l)}
    &= \exp\!\Bigl(-i\gamma_{zz,l}\,\sum_{(i,j)\in E} w_{ij}\,Z_iZ_j\Bigr)
       \cdot
       \prod_{j=0}^{n-1} \exp\!\bigl(-i\gamma_{z,j,l}\,Z_j\bigr), \\
    U_M^{(l)}
    &= \prod_{j=0}^{n-1} \exp\!\bigl(-i\beta_{j,l}\,X_j\bigr),
\end{align*}
where the cost layer carries one global ZZ angle $\gamma_{zz,l}$ shared across all edges,
$n$ per-qubit Z angles $\{\gamma_{z,j,l}\}$, and the mixer carries $n$ independent
RX angles $\{\beta_{j,l}\}$ per qubit.
Total parameter count is $N = (2n+1)d$; for $n = 16$ and $d = 3$ this gives $N = 99$.
The circuit is initialised in the uniform superposition $|+\rangle^{\otimes n}$ and
the energy measured under the weighted MaxCut Hamiltonian
$H_C = -\nicefrac{1}{2}\sum_{(i,j)\in E} w_{ij}(I - Z_iZ_j)$.

\section{Experimental hyperparameters}\label{app:hyperparams}

All ECG5000 and MNIST experiments use the following setup unless otherwise stated. Training uses Adam~\cite{kingma_adam_2017} with learning rate $\eta = 0.001$ and learning rate scheduling. Directional derivatives are estimated via central finite difference with step size $\varepsilon = 0.1$. Results are averaged over 3 random seeds.

\subparagraph{Shot budget.} Forward gradient methods use a fixed budget of $B = 1000$ shots per gradient step, divided as $B = 2VM$ where $M$ is the number of shots per direction. For example, $V = 10$ gives $M = 50$. The parameter-shift rule uses $M = 10$ shots per parameter, giving $B = 2NM$ total shots per step (which grows with $N$). SPSA corresponds to $V = 1$, $M = 500$; RCD uses $M = 500$ for a single coordinate.

\subparagraph{ECG5000.} 5-class classification (1 normal rhythm and 4 arrhythmia classes) on the UCR ECG5000 dataset. System sizes $n \in \{10, 20, 30, 40, 50, 60\}$ with $N = n(n-1)/2$ parameters. Forward methods train for 10,000 epochs; exact gradient and PS for 2,000 epochs. Shot budget termination at $5 \times 10^8$ total shots.

\subparagraph{MNIST.} 10-class classification on downsampled MNIST. System sizes $n \in \{10, 20, 40, 50\}$. Forward and SPSA/RCD methods train for 5,000 epochs; exact gradient and PS for 1,000 epochs. Same shot budget and Adam learning rate as ECG5000.

\subparagraph{VQE.} TFIM Hamiltonian with $J = h = 1$, open boundary conditions. Hardware-efficient ansatz with $n_L$ layers, $N = 2 n n_L$ parameters. The hero and regime-crossover figures use $n = 10$; the $N$-scaling, trajectory, and $\varepsilon$-sweep figures use $n = 8$. Forward gradient learning rate $\eta = 0.003$; PS learning rate $\eta = 0.01$. Parameter initialisation scale $0.1$. Energy error is reported as $|E_{\mathrm{best}} - E_{\mathrm{exact}}|$ where $E_{\mathrm{exact}}$ is computed by exact diagonalisation of the Hamiltonian.

\subparagraph{QAOA MaxCut.} Weighted Erd\H{o}s--R\'enyi graphs on $n = 16$ vertices with edge probability $p = 0.5$ and weights drawn uniformly from $[0.1, 1.0]$. Depth $d = 3$ ansatz with $N = (n{+}1)d + nd = (2n{+}1)d = 99$ parameters. Forward gradient learning rate $\eta = 0.1$; $\varepsilon = 0.1$. Ground energy $E_{\mathrm{opt}}$ computed by brute force. Approximation ratio reported as $r = E_{\mathrm{best}} / E_{\mathrm{opt}}$.

\subparagraph{Adaptive $V^\star$ and \textsc{quiver} parameters.} The \textsc{quiver} optimiser of \secref{ssec:quiver} uses exponential moving averages with decay $\beta = 0.9$ for both $\widehat{g}^2_t$ and $\widehat{\sigma}_t^2$. The sample estimator for $\bar\sigma_\nabla^2$ at each epoch is the empirical variance of the $V$ directional derivative estimates $\{\widetilde{\nabla}_{\boldsymbol{v}^\ell}\mathcal{L}^M\}_{\ell=1}^V$ around their mean, which under the noise-concentration result of \secref{ssec:noise_concentration} converges to $\bar\sigma_\nabla^2 / M$. The per-direction shot count $M$ is the only hyperparameter of the optimiser; clamp bounds are $V_{\min} = 2$ and $V_{\max} = N$ unless stated otherwise. All runs with \textsc{quiver} use the same Adam learning rate as the corresponding fixed-$V$ baseline in \figref{fig:hero}.

\end{document}